\documentclass[letterpaper,journal]{IEEEtran}

\usepackage[utf8]{inputenc}
\usepackage{cite}
\usepackage{url}
\usepackage{graphicx}
\usepackage{tikz}
\usepackage{pgf}
\usepackage{xcolor}

\usepackage{amsmath}
\usepackage{amsfonts}
\usepackage{amssymb}
\usepackage{amsthm}
\usepackage{mathtools}
\usepackage{bm}
\usepackage{mathrsfs}
\usepackage{wasysym}
\usepackage{stmaryrd}
\usepackage{bbm}
\usepackage{upgreek}
\usepackage{braket}

\newtheorem{lemma}{Lemma}
\newtheorem{theorem}{Theorem}
\newtheorem{corollary}{Corollary}

\newtheorem{prop}{Proposition}
\newtheorem{defn}{Definition}

\usepackage{multirow}
\usepackage{tabularx}
\usepackage{multicol}
\usepackage{enumerate}
\usepackage{paralist}
\usepackage{subfigure}
\usepackage{wrapfig}
\usepackage{sidecap}
\usepackage{setspace}
\usepackage{lineno}
\usepackage{comment}
\usepackage{soul}
\usepackage{eurosym}

\usepackage{algorithmicx}
\usepackage{algpseudocode}
\usepackage{algorithm}

\usepackage{bibunits}
\usepackage[hidelinks]{hyperref}
\usepackage{cleveref}

\usetikzlibrary{arrows}

\AtBeginDocument{%
    \definecolor{figred}{RGB}{162,0,37}%
    \definecolor{figblue}{RGB}{0,80,239}%
    \definecolor{figyellow}{RGB}{240,163,10}%
    \definecolor{figgray}{RGB}{179,179,179}%
    \definecolor{figslate}{RGB}{100,118,135}%
    \definecolor{figgreen}{RGB}{0,153,0}%
    \definecolor{figorange}{RGB}{255,128,0}%
    \definecolor{figmagenta}{RGB}{255,0,255}%
    \definecolor{RED}{rgb}{1,0,0}%
}

\newcommand{\red}[1]{\textcolor{red}{#1}}

\DeclareRobustCommand{\vnode}[1]{%
    \tikz[baseline=-0.6ex]
    \draw[black,fill=#1,line width=0.4pt]
    (0,0) circle (0.65ex);
}

\DeclareRobustCommand{\cnode}[1]{%
    \tikz[baseline=-0.6ex]
    \draw[black,fill=#1,line width=0.4pt]
    (-0.6ex,-0.6ex) rectangle (0.6ex,0.6ex);
}

\title{
Temporal Degeneracy Limits Message Passing
Decoding of QLDPC Codes Under Repeated Noisy Syndrome
Measurements
}
\author{Dimitris~Chytas,~\IEEEmembership{Member,~IEEE,} Nithin Raveendran,~\IEEEmembership{Member,~IEEE,} and Bane~Vasi\'{c},~\IEEEmembership{Fellow,~IEEE}\\
\IEEEauthorblockA{Department of Electrical and Computer Engineering The University of Arizona, Tucson, AZ 85721, USA}
Email: dchytas@arizona.edu, nithin@arizona.edu, vasic@ece.arizona.edu
\thanks{}
}
\date{June 2025}

\begin{document}

\maketitle
\begin{abstract}
Practical quantum computations can require logical error rates of $10^{-12}$ or lower. Quantum low-density parity-check (QLDPC) codes
are promising candidates for reaching such rates with low qubit overhead, provided that their syndromes can be decoded efficiently. Repeated noisy syndrome measurements introduce temporal degeneracy into the space-time decoding graph. We characterize the temporal kernel, whose vectors have zero total data parity, and show that it is generated by elementary temporal codewords of weight $d_v+2$, where $d_v$ is the variable-node degree. For $(d_v,d_c)$-regular
parity-check matrices with fixed degrees $d_v,d_c>2$ and girth greater
than four, we derive counting bounds and construct degenerate error
patterns of weight $(d_v+2)/2$ for even $d_v$ and $d_v+1$ for odd
$d_v$, with multiplicity $\Theta(nr)$ in the blocklength $n$ and
number of measurement rounds $r$. Under the isolation assumption, these patterns cause local failures
of deterministic isotropic message-passing decoders under a flooding
schedule. These failures depend on the node degrees rather than
the minimum distance; for the bivariate bicycle (BB) codes considered, their weight is equal to four. Monte carlo simulations and fixed-weight importance sampling quantify harmful low-weight degenerate error patterns for the min-sum (MS) decoder beyond the isolation assumption. Motivated by the identified
symmetries, we propose space-time anisotropic min-sum decoding with
syndrome-guided realizations. Its complexity is linear in the number
of Tanner graph edges for fixed maximum number of iterations and  number of decoders.
Numerical results on BB codes show threshold-like behavior across
the tested code sizes and improved failure probabilities relative
to MS decoding with ordered statistics decoding post-processing.
\end{abstract}
\begin{IEEEkeywords}
stabilizer codes, iterative decoding, degeneracy, symmetric stabilizers, noisy syndrome measurements, phenomenological noise model, graph automorphisms, decoding limitations.
\end{IEEEkeywords}
\section{Introduction}
\IEEEPARstart{L}arge-scale quantum computations can require logical error rates of $10^{-12}$ or lower. A resource estimation study ~\cite{beverland2022assessing} of three applications, namely quantum dynamics, quantum chemistry and factoring, found that each logical qubit must fail with probability below about $10^{-11}$, $3\times10^{-17}$ and $4\times10^{-16}$, respectively~\cite{beverland2022assessing}, and a logical error rate of $10^{-12}$ per logical qubit per round, the so-called \emph{teraquop} regime, is commonly adopted as a benchmark target~\cite{gidney2021honeycomb}. Superconducting processors have demonstrated surface-code memories operating below threshold~\cite{google2025below}.
Other platforms have also demonstrated high-fidelity entangling operations, with reported two-qubit gate errors of approximately
$8\times10^{-5}$ in trapped ions~\cite{hughes2025trapped} and $1.5\times10^{-3}$ in neutral atoms~\cite{evered2026entangling}. Bridging this gap is the task of quantum error correction (QEC)~\cite{quantumGF4,shor1997faulttolerantquantumcomputation,Gottesman97}, and it must be accomplished by a classical decoder that keeps pace with the hardware, since universal quantum computation demands real-time decoding for non-Clifford gates~\cite{PhysRevX.11.041058}. Errors in quantum systems can be modeled at different levels of realism: the code-capacity model considers only errors on data qubits, the phenomenological model additionally accounts for faulty measurements, and the circuit-level model incorporates errors on data qubits, ancilla qubits, and gates~\cite{vasic2025quantumlowdensityparitycheckcodes}.

Quantum low-density parity-check (QLDPC)  codes have gained attention because they can simultaneously achieve constant rate and linear minimum distance~\cite{panteleev2021quantumLinearMinDLocalTestable, QuantumTannerCodeszemor,vasic2025quantumlowdensityparitycheckcodes}. Families of non-topological QLDPC codes, such as bivariate bicycle (BB) codes, have been shown to outperform the surface code in terms of qubit overhead~\cite{bravyi2024high}, albeit at the cost of increased qubit connectivity. However, decoding remains challenging for QLDPC codes. Unlike classical error correction, where the error estimate must match the induced error exactly, in QEC it suffices to match the induced error up to a stabilizer, so multiple error patterns that differ by a stabilizer are equally valid corrections for a given syndrome~\cite{Gottesman97}. This phenomenon, known as \textit{degeneracy}, 
becomes a bottleneck under iterative decoding~\cite{Poulin_2008,quantumTS} when the minimum distance increases relative to stabilizer weight. 
In the code-capacity model, degenerate error patterns are supported on stabilizers, and their effect can be mitigated by iterative decoders, such as belief propagation (BP) or min-sum (MS), followed by post-processing, e.g., ordered statistics decoding (OSD)~\cite{osd,roffe_decoding_2020}, localized statistics decoding~\cite{Hillmann_2025}, guided decimation~\cite{decim}, or beam search-based decoding~\cite{ye2025beamsearchdecoderquantum,taghipour2026multistagerewindingdecoderqldpc}.
More recently, techniques that introduce anisotropic update rules have been proposed to overcome degeneracy, achieving performance close to that of OSD in linear time~\cite{chytasminsum,chytas2026edgebasedanisotropicdecodinggeneralized}. Related low-complexity decoding approaches have also been demonstrated with more realistic noise models; for example, the Relay-BP decoder introduces disordered memory strengths to break symmetries under circuit-level noise~\cite{müller2025improvedbeliefpropagationsufficient}. Finally, decoders that exploit code symmetries to mitigate degeneracy have also been proposed~\cite{symbreak,koutsioumpas2025,pradhan2026fastfaulttolerantdecodershypergraph}.

Decoding performance also depends on the syndrome information available to the decoder. In the code-capacity model, syndrome measurements are assumed to be perfect. In realistic systems, however, syndrome measurements can be faulty, so decoding must account for measurement errors in addition to errors on the data qubits. Two approaches are commonly used to address measurement errors. One possibility is \emph{single-shot decoding}, where redundancy in the measured syndrome allows data and measurement errors to be corrected using a single noisy measurement round~\cite{bombin,kubica,pryadko_phenom,jacob2026singleshotdecodingfaulttolerantgates}. Achieving this property generally requires sufficient stabilizer redundancy, syndrome distance, or confinement. Such properties can be engineered in several QLDPC constructions, including recent generalized and BB codes with redundant stabilizer measurements~\cite{pryadko_phenom,rowshan2026singleshotfewshotdecodingstabilizer}. However, they are not generic properties of QLDPC codes.

A more general alternative is to perform \emph{repeated syndrome measurements}~\cite{shor1997faulttolerantquantumcomputation,Dennis_2002}. In this setting, multiple noisy syndrome rounds are performed to suppress measurement errors and infer a consistent space-time error configuration. Multi-round processing can also improve the effective distance and reduce the logical error rate per round relative to single-shot decoding, and has motivated sliding-window and temporally parallel decoding strategies aimed at limiting decoding latency~\cite{PhysRevA.110.012453,Skoric_2023}.
Formally, the phenomenological noise model treats data and syndrome errors as independent stochastic processes while idealizing the recovery operation itself. 
This abstraction captures important features of fault-tolerant operation without requiring a full
circuit-level description of gate and ancilla faults~\cite{vasic2025quantumlowdensityparitycheckcodes,ImprovedDecodingSoftInfo_2021}.
Decoding in this model amounts to estimating a consistent space-time error configuration from a sequence of noisy syndrome measurements and applying a correction that does not induce a logical error~\cite{kuophenom}. Similar decoding formulations also arise when stabilizer generators are measured according to nonuniform or partial schedules, where only subsets of checks are available in each round~\cite{Berthusen_2024}.

For a quantum memory protected by repeated syndrome measurements, the relevant figure of merit is the logical error rate, which determines the expected number of rounds before a logical failure occurs~\cite{kuophenom}. Crucially, under iterative decoding, for low error rate regimes, the logical error rate is dominated by specific low-weight error patterns. In classical error correction, such errors are supported on locally dense subgraphs called trapping sets~\cite{ontology}. Trapping-set analysis has also been extended to circuit-level detector error models~\cite{pacenti2026trappingsetsdetectorerror}. In QEC, constant-weight degenerate error patterns  stem from the orthogonality constraints  of the code and persist independently of its distance~\cite{quantumTS}; in this work, we refer to them as spatial degenerate error patterns. Under the code-capacity model, BP-based decoders can achieve good performance for QLDPC codes, often in linear time when appropriately
modified to account for spatial degeneracy~\cite{chytasminsum}. However, optimizing a decoder for the code-capacity model alone does not guarantee comparable performance under phenomenological noise.

The temporal configurations considered here are imposed by repeated noisy syndrome measurements and arise independently of the specific quantum code construction. Sparsity makes message-passing decoding attractive for QLDPC codes, but the resulting space-time graph contains local symmetries that deterministic isotropic decoders preserve under flooding, symmetric initialization, and the isolation assumption. Under the degree and girth assumptions considered here, these symmetries create local decoding failures that cannot be removed by increasing the code distance. This motivates anisotropic update rules that break the relevant symmetries while retaining low-complexity message passing.

Three features make this limitation severe. First, the weight of the failing patterns is fixed by the variable-node degree $d_v$, which is a small constant for any LDPC code, and not by the minimum distance $d$. For the $[[144,12,12]]$ and $[[288,12,18]]$ BB codes, $d_v=3$ and isotropic decoders fail on weight-$4$ patterns, whereas the distances of these codes would allow up to $5$ and $8$ errors, respectively, to be corrected. Second, the number of such patterns grows linearly with both the blocklength $n$ and the number of measurement rounds $r$. Longer codes and more measurement rounds, which is precisely the regime needed to reach low logical error rates since the number of rounds typically grows with the distance, therefore increase the number of harmful patterns instead of suppressing them. Third, the phenomenon depends only on the degrees and the girth of the Tanner graph and not on the algebraic construction of the code, so it is shared by all $(d_v,d_c)$-regular QLDPC codes of girth larger than four. In Section~\ref{sec:secvi}, full-monte carlo simulations and fixed-weight importance sampling
tests confirm that these failures are relevant in practice and examine  remaining low-weight failures of MS-based decoders.

Related spatial-temporal  structures have appeared in previous phenomenological and fault-tolerant decoding formulations~\cite{Dumer2014,pryadko_2020,pryadko_phenom,tsubouchi2026degeneracycuttinglocalefficient}. In particular, the detector-degeneracy formulation of~\cite{tsubouchi2026degeneracycuttinglocalefficient} extends stabilizer-induced degeneracy to phenomenological and circuit-level decoding graphs. The repeated degeneracy-generating matrices in~\cite{pryadko_phenom} and~\cite{tsubouchi2026degeneracycuttinglocalefficient} contain the elementary temporal generators considered here. Building on these representations, we characterize the combinatorial structure of the temporal kernel, derive counting bounds for configurations supported on minimum even-weight temporal codewords, and identify symmetry-induced failures of isotropic flooding decoders under the stated assumptions.

Our analysis starts from the base matrix of the phenomenological noise model. Using the notion of minimal vectors~\cite{705584}, we characterize its kernel and extend the analysis to a single data variable and its neighboring checks across repeated measurement rounds. The resulting elementary temporal codeword generators have weight $(d_v+2)$, and their embedded copies generate the full temporal kernel. Temporal codewords have zero total data parity across rounds and therefore relate error patterns with the same overall data error.

For $(d_v,d_c)$-regular parity-check matrices with fixed degrees $d_v,d_c>2$ and girth greater than four, we derive decoder-agnostic upper bounds on the number of distinct locally minimum half-weight error patterns supported on minimum even-weight temporal codewords with nonzero measurement parts. Under symmetric initialization and the isolation assumption, we identify a subset of these constant-weight patterns whose multiplicity grows linearly with the blocklength and the number of measurement rounds, and which cause local decoding failures for deterministic syndrome-based isotropic decoders under a flooding schedule. Motivated by this analysis, we propose space-time anisotropic min-sum decoding augmented with syndrome-guided realizations. For fixed iteration and ensemble budgets, its complexity is linear in the number of Tanner graph edges. Numerical results show threshold-like behavior for the BB codes considered under one noisy measurement round followed by an ideal one, together with improved performance relative to MS followed by OSD.

The rest of the paper is organized as follows. In Section \ref{sec:pre}, we review QEC preliminaries and the phenomenological noise model, focusing on the concept of temporal degeneracy. In Section \ref{sec:sectemporal}, we analyze the base matrix of the phenomenological noise model and extend the analysis to the case in which a single variable node and its neighboring checks are measured repeatedly.  In the final part of the same section, we derive upper bounds on the number of locally minimum half-weight temporal degenerate error patterns. In Section~\ref{sec:secv}, we provide lower bounds on the number of minimum-weight harmful temporal degenerate error patterns for isotropic decoders and we propose the idea of anisotropic decoding. Finally, in Section~\ref{sec:secvi}, we provide numerical results for isotropic and anisotropic decoders on a family of QLDPC codes under phenomenological noise.

\section{Preliminaries}
\label{sec:pre}
\subsection{Stabilizer formalism}
Let $\mathbb{F}_2^n$ denote the set of binary vectors of length $n$. The \textit{Hamming weight} of an element in $\mathbb{F}_2^n$ is defined as the number of its non-zero entries. An $[n, k, d]$ linear code $C \subset \mathbb{F}_2^n$ is a linear subspace of $\mathbb{F}_2^n$ spanned by $k$ linearly independent vectors, whose nonzero codewords have minimum Hamming weight $d$. A code $C$ can also be represented by a parity-check matrix $\mathbf{H}$, where $C = \ker(\mathbf{H})$. For simplicity, we consider regular parity-check matrices, with column weight $d_v$ and row weight $d_c$.

Let $P_n$ be the $n$-qubit Pauli group. A \textit{stabilizer} group is an Abelian subgroup $S \subset P_n$ not containing $-I$. An $\llbracket n,k,d \rrbracket$ stabilizer code is a $2^k$-dimensional subspace $\mathcal{C}$ of the Hilbert space $(\mathbb{C}^2)^{\otimes n}$ such that 
\[
s_i\ket{\Psi} = \ket{\Psi}, \quad \forall\ s_i \in S, \ \ket{\Psi} \in \mathcal{C}.
\]
A $\llbracket n, k_X - k_Z^{\perp}, d \rrbracket$ CSS code is a stabilizer code constructed from two classical codes $C_X = \ker(\mathbf{H}_X)$ and $C_Z = \ker(\mathbf{H}_Z)$, satisfying $C_Z^{\perp} \subset C_X$ and $C_X^{\perp} \subset C_Z$~\cite{calderbank1996quantum_exists}. The minimum distance $d = \min\{d_X, d_Z\}$ is determined by the minimum Hamming weights of codewords in $C_X \setminus C_Z^{\perp}$ and $C_Z \setminus C_X^{\perp}$. For more details, see~\cite{Gottesman97}.

CSS codes allow $X$ and $Z$ errors, when assumed independent, to be decoded separately as binary problems: $X$ errors are decoded using $\mathbf{H}_Z$, and $Z$ errors are decoded using $\mathbf{H}_X$. Thus, the channel can be modeled as two independent binary symmetric channels (BSCs)~\cite{mackay_quantum}. Let $\mathbf{e}=[\mathbf{e}_X, \mathbf{e}_Z]$ denote the binary representation
of a Pauli error acting on the $n$ qubits. The corresponding input syndromes are obtained as $\mathbf{s}_Z=   \mathbf{H}_Z\mathbf{e}_X$
and $\mathbf{s}_X=   \mathbf H_X\mathbf{e}_Z$, or, in  compact form, $\mathbf{s}=[\mathbf{s}_X, \mathbf{s}_Z]$. This model, in which errors occur only on data qubits and syndrome measurements are noiseless, is referred to as the code-capacity model~\cite{vasic2025quantumlowdensityparitycheckcodes}. 

A zero syndrome indicates that no error has been detected, whereas a nonzero syndrome indicates the presence of errors. Syndrome-based decoders estimate $\hat{\mathbf{e}}$ whose syndrome matches the observed one, i.e., $\mathbf{\hat{s}} = \mathbf{s}$. Considering $X$ errors  without loss of generality, decoding succeeds if the error is recovered up to a stabilizer, i.e., 
\begin{equation}
    \hat{\mathbf{e}}_X= \mathbf{e}_X \oplus \mathbf{h}, \mathbf{h} \in \text{rowspace}(\mathbf{H}_X),
    \label{eq:spatialDegeneracy}
\end{equation}
whereas failure occurs if the syndrome does not match or if a logical error occurs, i.e., $\mathbf{e}_X \oplus \hat{\mathbf{e}}_X$ commutes with all stabilizers but is not in the rowspace of $\mathbf{H}_X$. 

Eq.~(\ref{eq:spatialDegeneracy}) gives rise to the phenomenon of \emph{degeneracy}, in which two error patterns that differ by an element of the rowspace of $\mathbf{H}_X$ produce the same syndrome. Degeneracy is particularly harmful for iterative decoding when the two equivalent error patterns have the same weight and are supported on symmetric stabilizers. A stabilizer is symmetric when its set of variable  nodes forms an induced
subgraph with no odd-degree check nodes and can be partitioned into an even number of disjoint subsets such
that: (i) the subgraphs induced by each subset of variable nodes are isomorphic, and (ii) each subset shares the
same set of neighboring odd-degree check nodes~\cite{quantumTS,vasic2025quantumlowdensityparitycheckcodes}. A typical example of a symmetric stabilizer~\cite{10857350,chytasminsum} is illustrated in Fig.~\ref{fig:symmetricStab}. A quantum code is called degenerate if it has stabilizers whose weight is less than its minimum distance. For QLDPC codes, as the minimum distance grows, the impact of degeneracy becomes more prominent~\cite{mackay_quantum,bravyi2024high,PhysRevA.109.022407}.

The $\mathbf{H}_X$ and $\mathbf{H}_Z$ stabilizer matrices can each be represented as a \textit{Tanner graph}, a bipartite graph whose two node sets are the set of $n$ variable nodes $\mathcal{V}$ and the set of $m$ check nodes $\mathcal{C}$. 
For $v \in \mathcal V, c \in \mathcal C$, if the $(c,v)$-th entry of the binary parity-check matrix is nonzero, variable node $v$ and check node $c$ are connected by an edge and are called neighbors.
We denote by $\mathcal{M}(v)$ the set of indices of the neighboring check nodes of the variable node $v$, and by $\mathcal{N}(c)$ the set of indices of the neighboring variable nodes of the check node $c$. Throughout the paper, we use the shorthand notation $[n]$ for the set $\{1,\dots,n\}$.

\begin{figure}[htbp]
    \centering
\includegraphics[width=0.5\linewidth]{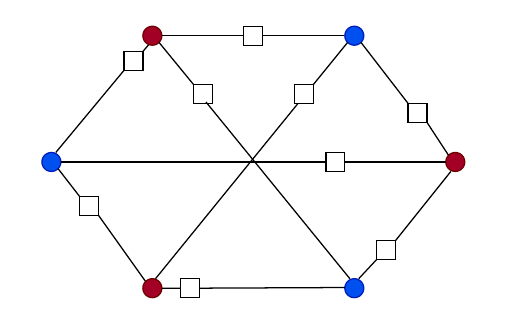}
    \caption[A symmetric stabilizer formed by two sets of variable nodes colored red and blue, respectively.]{A symmetric stabilizer formed by two sets of variable nodes , \vnode{figred} and \vnode{figblue}. Since its induced subgraph is isomorphic to a subdivided $K_{3,3}$ complete bipartite graph, each variable node in one set shares a check node with every variable node in the other set.}
    \label{fig:symmetricStab}
\end{figure}

\subsection{Phenomenological Noise Model}
\label{sec:phenom}
We now describe the phenomenological noise model, in which, in each round of syndrome extraction, an $X$ (or equivalently a $Z$) error is applied to each qubit independently with probability $\alpha$ before the measurements, and each measurement outcome is flipped independently with the same probability~\cite{PhysRevA.110.012453,jacob2026singleshotdecodingfaulttolerantgates}.

We consider $r$ measurement rounds. For each  measurement round $j \in [r]$, the syndrome is given by
\begin{equation}
\mathbf{s}_j=\mathbf{H}_Z(\mathbf{e}_1\oplus\mathbf{e}_2\oplus\dots\oplus\mathbf{e}_j)\oplus\boldsymbol{\epsilon}_j, \boldsymbol{\epsilon}_r=0
\label{eq:naivePhenomeno}
\end{equation}
where $\mathbf{e}_i \in \mathbb{F}_2^{n}, i\in[j]$, are binary vectors representing the data  errors
 on the $n$ qubits 
 at the $i$-th measurement round, and  $\boldsymbol{\epsilon}_j$ is a binary vector representing the measurement errors, modeled as syndrome flips, in measurement round $j$. Conventionally, the last measurement round is assumed to be perfect~\cite{pryadko_phenom}. The parity-check matrix described by Eq.~(\ref{eq:naivePhenomeno}) contains degree-$1$ variable nodes and $4$-cycles and therefore is not well suited for iterative decoding~\cite{kuophenom}. 
By taking the differences between the syndromes of consecutive measurement rounds, errors are detected locally:
\begin{equation}
\boldsymbol{\sigma}_j = \mathbf{s}_j \oplus \mathbf{s}_{j-1} =  \mathbf{H}_Z\mathbf{e}_j \oplus (\boldsymbol{\epsilon}_j \oplus \boldsymbol{\epsilon}_{j-1}) =  \mathbf{H}_Z\mathbf{e}_j \oplus \tilde{\boldsymbol{\epsilon}}_j,
\end{equation}
where $\boldsymbol{\epsilon}_0=\mathbf0$, $\mathbf{s}_0=\mathbf 0$  $j \in [r]$ and $\tilde{\boldsymbol{\epsilon}}_j$ corresponds to the difference between measurement errors in two consecutive rounds. The detectors $\boldsymbol{\sigma}_j$ give rise to an equivalent parity-check matrix given by~\cite{pryadko_phenom,kuophenom}:
\begin{equation}
\setlength{\arraycolsep}{2pt}
{
\mathbf{H}' = 
\left[
\begin{array}{c|c}
\begin{array}{cccc}
\mathbf{I}_m \\
\mathbf{I}_m & \mathbf{I}_m \\
& \ddots & \ddots \\
& & \mathbf{I}_m & \mathbf{I}_m\\
& & & \mathbf{I}_m
\end{array}
&
\begin{array}{ccccc}
\mathbf{H}_Z \\
& \mathbf{H}_Z \\
& & \ddots \\
& & & \mathbf{H}_Z \\
& & & & \mathbf{H}_Z
\end{array}
\end{array}
\right],}
\label{eq:phenomPCM}
\end{equation}
with $\mathbf{H}' \in \mathbb{F}_2^{mr \times [m(r-1)+nr]}$, and $M=mr$, $N=m(r-1)+nr$. After exchanging $X$ and $Z$ and consistently reordering the columns, $\mathbf H'$ coincides with the phenomenological detector
matrix of~\cite{pryadko_phenom}. The corresponding degeneracy-generating matrix is orthogonal to $\mathbf H'$; its rows represent the elementary
temporal codewords analyzed below and spatial stabilizers embedded in individual rounds. We refer to the Tanner graph corresponding to $\mathbf{H}'$ as the \emph{space-time Tanner graph}. We also refer to the left and right blocks of $\mathbf{H}'$ as the measurement part and the data part, respectively. The variable nodes of each part are referred to as measurement and data variable nodes, respectively. 

Now consider $\mathbf{D}\in\ker(\mathbf{H}')$, i.e.,
\[
\mathbf{H}'\mathbf{D}=\mathbf{0}.
\]
We call $\mathbf{D}$ a \emph{spatial-temporal vector} and  represent it as \[\mathbf{D} = (\mathbf{D}_{1},\dots,\mathbf{D}_{r-1}\;|\;\mathbf{D}_{r},\dots,\mathbf{D}_{2r-1} ) \in \mathbb{F}_2^{N},\]
where the left part corresponds to the measurement part and the right part to the data part of the vector. Each block of the measurement part has length $m$, and each block of the data part has length $n$.

\begin{defn}
For $r$ measurement rounds, we define a spatial-temporal codeword as $\mathbf{D} \in \operatorname{ker}(\mathbf{H}')$ that satisfies $\sum_{i=r}^{i=2r-1}\mathbf{D}_i=\mathbf{h} \pmod{2}, \mathbf{h} \in \text{rowspace}(\mathbf{H}_X)$. 
If $\mathbf{h}=\mathbf{0}$, we call $\mathbf{D}$ a temporal codeword; if $\mathbf{h}\neq \mathbf{0}$ and $\mathbf{D}$ has a nonzero measurement part, we call it a mixed spatial-temporal codeword.
\label{definition:definitionSpatialTemporal}
\end{defn}

Let us decompose a spatial-temporal codeword $\mathbf{D}$ into two vectors with disjoint supports:
\begin{equation}
    \mathbf{D}=\mathbf{E}\oplus\mathbf{E}' ,
\end{equation}
where 
\begin{equation}
\begin{split}
\mathbf{E} = (\mathbf{E}_1, \dots, \mathbf{E}_{r-1} \;|\; \mathbf{E}_{r}, \dots, \mathbf{E}_{2r-1}) \in \mathbb{F}_2^{N} \\
\mathbf{E}' = (\mathbf{E}'_1, \dots, \mathbf{E}'_{r-1} \;|\; \mathbf{E}'_{r}, \dots, \mathbf{E}'_{2r-1}) \in \mathbb{F}_2^{N}.
\end{split}
\end{equation}
Since $\mathbf{H}'\mathbf{D}=\mathbf{0}$, $\mathbf{E}$ and $\mathbf{E}'$ produce the same syndrome:
\begin{equation}
\mathbf{H}'\mathbf{E}=\mathbf{H}'\mathbf{E'}.
\end{equation}
The goal of the decoder is to estimate the total data error up to a stabilizer:
\begin{equation}
\mathbf{e}'_{\rm tot} = \mathbf{e}_{\rm tot} \oplus h, \qquad h \in \operatorname{rowspace}(\mathbf{H}_X),
\label{eq:codespace}
\end{equation}
where $\mathbf{e}_{\rm tot},\mathbf{e}'_{\rm tot} \in \mathbb{F}_2^{n}$ denote the overall data error and its estimate~\cite{PhysRevA.110.012453}, obtained as:
\begin{equation}
\mathbf{e}_{\rm tot} := \sum_{i=r}^{2r-1} \mathbf{E}_i, \pmod{2}, \qquad
\mathbf{e}'_{\rm tot} := \sum_{i=r}^{2r-1} \mathbf{E}'_i\pmod{2},
\end{equation}
with difference
\begin{equation}
\Delta \mathbf{e}_{\rm tot} := \mathbf{e}_{\rm tot} \oplus \mathbf{e}'_{\rm tot} = \sum_{i=r}^{2r-1} \mathbf{D}_i=h.
\label{eq:parity}
\end{equation}
Such an equal-weight decomposition exists only when $w(\mathbf D)$ is even. In that case, any partition of $\operatorname{supp}(\mathbf D)$ into two subsets of size $w(\mathbf D)/2$ gives two disjoint vectors $\mathbf E$ and $\mathbf E'$ of equal weight satisfying $\mathbf D=\mathbf E\oplus \mathbf E'$.
Clearly, spatial-temporal degeneracy can hinder iterative decoding, especially when the distance is large compared to $w(\mathbf{D})$, i.e., $w(\mathbf{D}) \ll d$. 
This is analogous to harmful degenerate errors of quantum codes in the code-capacity model~\cite{quantumTS,osd}. For clarity, in the rest of the paper we refer to degeneracy that stems from the rowspace of $\mathbf{H}_X$ as \emph{spatial degeneracy}. Degeneracy stemming from temporal codewords is referred to as temporal degeneracy, and more generally, degeneracy stemming from spatial-temporal codewords as spatial-temporal degeneracy.
Fig.~\ref{fig:simpleTemporal} illustrates an example of a temporal codeword, the structure of which will become clear in the following sections.

In the rest of this paper, we consider only $X$ errors and their decoding; without loss of generality, the notation $\mathbf{H}$ will refer to $\mathbf{H}_Z$, $h$ will denote an element of the rowspace of $\mathbf{H}_X$, $\mathbf{e}$ will refer to $\mathbf{e}_X$, and $\mathbf{s}$ will refer to $\mathbf{s}_Z$. 
\begin{figure}[htbp]

        \centering
       \includegraphics[trim=17 30 17 31, clip, width=0.3\linewidth]{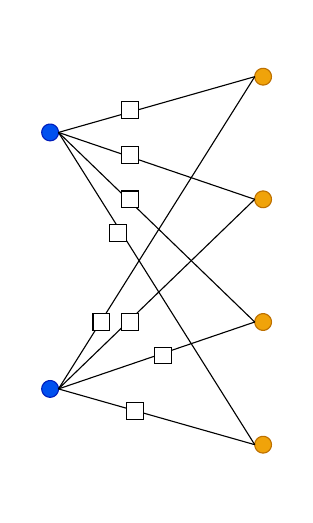}

\caption{The induced subgraph of a temporal codeword formed by a variable node ($d_v=4$) repeated in two measurement rounds.  Observe that the graph is isomorphic to a subdivision of the $K_{2,4}$ complete bipartite graph. Data variable nodes (\vnode{figblue}) and measurement variable nodes (\vnode{figyellow}) form the two partitions.}
    \label{fig:simpleTemporal}
\end{figure}

\section{Evaluation of Temporal Degeneracy}
\label{sec:sectemporal}
The purpose of this section is to identify the error patterns that cause the decoding failures established in Section~\ref{sec:secv}, and to count them. We start from a small base matrix that captures the time structure of $\mathbf{H}'$ independently of the code, extend the analysis to a single variable node repeated in time and then to the full space-time matrix, and finally count the half-weight error patterns supported on the resulting temporal codewords. 

\subsection{Analysis of the base matrix}
\label{sec:proto}
We consider the base matrix of the phenomenological noise model for $r$ measurement rounds, obtained by replacing each nonzero block of Eq.~(\ref{eq:phenomPCM}) with a one:  \begin{equation}
\mathbf{B} = 
\left[
\begin{array}{c|c}
\begin{array}{cccc}
1 \\
1 & 1 \\
& \ddots & \ddots \\
& & 1 & 1\\
&  &  \\
& & & 1
\end{array}
&
\begin{array}{ccccc}
1 \\
& 1 \\
& & \ddots & &\\
&\\
& &  \\
& & & &  1 \\
\end{array}
\end{array}
\right],
\label{eq:BasephenomPCM}
\end{equation}
where $\mathbf{B} \in \mathbb{F}_2^{r \times (2r-1)}$. The first $r-1$ columns correspond to the measurement part and the remaining $r$ columns correspond to the data part.

Let
\[
\mathbf{c}=(\mathbf{z}\mid \mathbf{x}) \in \mathbb{F}_2^{2r-1},
\]
where $\mathbf{z}=(z_1,\dots,z_{r-1})$ denotes the measurement part and $\mathbf{x}=(x_1,\dots,x_r)$ denotes the data part.

Then $\mathbf{c}\in \ker(\mathbf{B})$ if and only if
\begin{equation}
x_1=z_1,\qquad x_t=z_{t-1}\oplus z_t,\ \ 2\le t\le r-1,\qquad x_r=z_{r-1}.
\label{eq:kernelB_recursion}
\end{equation}
Therefore, every choice of the left block $\mathbf{z}\in\mathbb{F}_2^{r-1}$ uniquely determines the right block $\mathbf{x}$, which implies that
\[
\dim\bigl(\ker(\mathbf{B})\bigr)=r-1.
\]
Hence, the total number of nonzero vectors in $\ker(\mathbf{B})$ equals $2^{r-1}-1$.

Eq.~(\ref{eq:kernelB_recursion}) also makes clear how the Hamming weight of a kernel vector is determined. The first $r-1$ entries can take any binary pattern. Let $k$ denote the number of ones among them. For $k=1$, i.e., a single nonzero bit in the left block, there are exactly $r-1$ vectors of total weight $3$, which are the minimum-weight nonzero vectors in $\ker(\mathbf{B})$. For $k=2$, one must distinguish whether the two nonzero bits are consecutive or not. If they are consecutive, they form a single run and trigger only two data bits, producing vectors of weight $4$. If they are non-consecutive, they form two disjoint runs and trigger four data bits, resulting in vectors of weight $6$. Hence, there are $r-2$ vectors of weight $4$ and $\frac{(r-2)(r-3)}{2}$ vectors of weight $6$. The following lemma generalizes this counting.

Lemma~\ref{lemma1} gives the complete weight distribution of $\ker(\mathbf{B})$. It is the starting point of our analysis, since every temporal codeword is built from lifted copies of these vectors, and after lifting it directly yields the weights of all temporal codewords of a single variable node in Eq.~(\ref{eq:enumerateOneVN}).

\begin{lemma}
Let $\mathbf{B}$ be the base matrix obtained for $r$ measurement rounds. Define 
\[
k = 0,1,\dots,r-1
\] 
as the number of ones in the first $r-1$ positions (left block), and 
\[
t = 0,1,\dots,k-1
\] 
as the number of consecutive pairs among these ones. Then the bivariate weight enumerator $W(X,Y)$ of $\ker(\mathbf{B})$ is
\[
W(X,Y) = 1+  \sum_{k=1}^{r-1} \sum_{t=0}^{k-1} \binom{k-1}{t} \binom{r-k}{k-t} \, X^{k} \, Y^{2(k-t)},
\]

where $X$ tracks the weight in the first $r-1$ positions (left weight) and  $Y$ tracks the weight in the next $r$ positions (right weight).
\label{lemma1}
\end{lemma}

\begin{IEEEproof}
The proof is given in Appendix~\ref{app:lemma1}.
\end{IEEEproof}
We  denote a kernel vector with left weight $k$ and right weight $2(k-t)$ by $\mathbf{c}^{k,2(k-t)}$. We also refer to such vectors as $\mathbf{c} \in \ker{\mathbf{B}}$. Not every kernel vector needs to be analyzed separately. A kernel vector that is not minimal is the sum of kernel vectors with strictly smaller supports, so the elementary objects are the minimal vectors~\cite{705584}. Lemma~\ref{lemma:minimal_xky2} identifies them.
\begin{lemma}
Any $\mathbf{c}^{k,2}$ vector with $1 \le k \le r-1$ is minimal, i.e., its support does not strictly contain the support of any other nonzero vector in $\ker(\mathbf{B})$.
\label{lemma:minimal_xky2}
\end{lemma}
\begin{IEEEproof}
 The proof is given in Appendix~\ref{app:lemma:minimal_xky2}.
\end{IEEEproof} 

Conversely, any kernel vector whose first $(r-1)$ positions contain more than one run of ones is not minimal, since each run defines a nonzero kernel vector with strictly smaller support. Hence,
the minimal vectors are exactly the vectors with $k$ consecutive ones in their first $r-1$ positions, and for each $k$ there are $r-k$ such vectors of weight $k+2$. 
 Clearly, the weight-$3$ vectors are the {minimum-weight minimal vectors of the base matrix, and in the following lemma\red, we  show that they form a basis of $\ker(\mathbf{B})$.

Lemma~\ref{lemma:decompose_xky2} is the key property of the base matrix: the weight-$3$ vectors form a basis of $\ker(\mathbf{B})$. After lifting, it becomes Corollary~\ref{cor:dv2GeneratesALL}, which states that weight-$(d_v+2)$ temporal codewords generate all temporal codewords of a variable node, regardless of the number of rounds.

\begin{lemma}
\label{lemma:decompose_xky2}
Let $\mathbf{c}^{k,2(k-t)}\in \ker(\mathbf{B})$ be a vector with left weight $k$ and right weight $2(k-t)$, so that its total Hamming weight is
\[
k+2(k-t)=3k-2t.
\]
Then $\mathbf{c}^{k,2(k-t)}$ can be expressed as a modulo-$2$ sum of weight-$3$ minimal vectors $\mathbf{c}^{1,2}_i$. Equivalently, the set
\[
\{\mathbf{c}^{1,2}_i\}_{i=1}^{r-1}
\]
forms a basis of $\ker(\mathbf{B})$.
\end{lemma}
\begin{IEEEproof}
The proof is given in Appendix~\ref{app:lemma:decompose_xky2}.
\end{IEEEproof}

\subsection{One variable node repeated in time}
We now consider the induced subgraph obtained by restricting $\mathbf{H}'$ to a single variable node of $\mathbf{H}$ and its $d_v$ neighboring checks repeated across $r$ rounds. 
Consider the parity-check matrix $\mathbf{H}'$ of Eq.~(\ref{eq:phenomPCM}).  
Let $v\in[n]$ denote a variable node of $\mathbf{H}$. Assume that in $\mathbf{H}'$ the data  columns are ordered round-by-round after the measurement columns. We define the set of data columns corresponding to the $r$ copies of variable node $v$ across consecutive measurement rounds as
\[
\mathcal{D}_r(v)
=
\{\, m(r-1)+v+jn : 0\le j\le r-1 \,\}.
\]
and denote the corresponding data variable nodes by
\[
v^{(1)},v^{(2)},\dots,v^{(r)}.
\]
Let $\mathcal{R}_r(v)$ denote the set of all check nodes adjacent to the variable nodes in $\mathcal{D}_r(v)$, and let $\mathcal{S}_r(v)$ denote the set of degree-$2$ measurement variable nodes adjacent to the check nodes in $\mathcal{R}_r(v)$. The parity-check matrix corresponding to this induced subgraph is the submatrix of $\mathbf{H}'$ induced by rows $\mathcal{R}_r(v)$ and columns $\mathcal{S}_r(v)\cup \mathcal{D}_r(v)$, which we denote by
\[
\mathbf{H}_1^{(r)}(v)
=
\mathbf{H}'\bigl(\mathcal{R}_r(v),\,\mathcal{S}_r(v)\cup \mathcal{D}_r(v)\bigr)
\in
\mathbb{F}_2^{\,d_v r \times \bigl(d_v(r-1)+r\bigr)}.
\]
Since the choice of $v$ is arbitrary, we  omit the dependence on $v$ and simply write $\mathbf{H}_1^{(r)}$.

\begin{figure}[htbp]
     \centering
     \begin{subfigure}[]
     {
        \centering
        \includegraphics[width=0.4\linewidth]{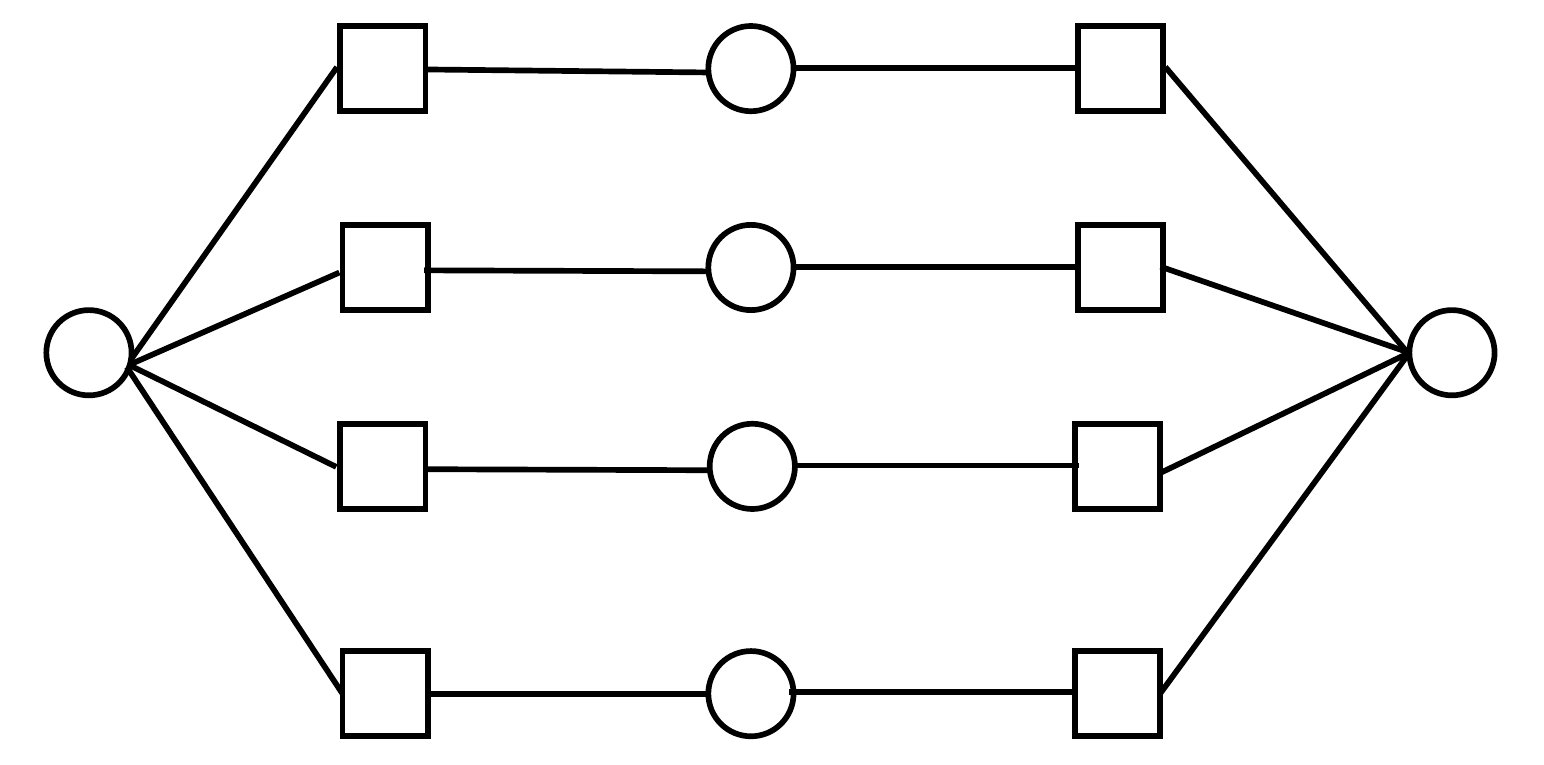}
         \label{fig:temporalStabilizer0}
    }
    \end{subfigure}
    \begin{subfigure}[]
    {
        \centering
       \includegraphics[width=0.4\linewidth]{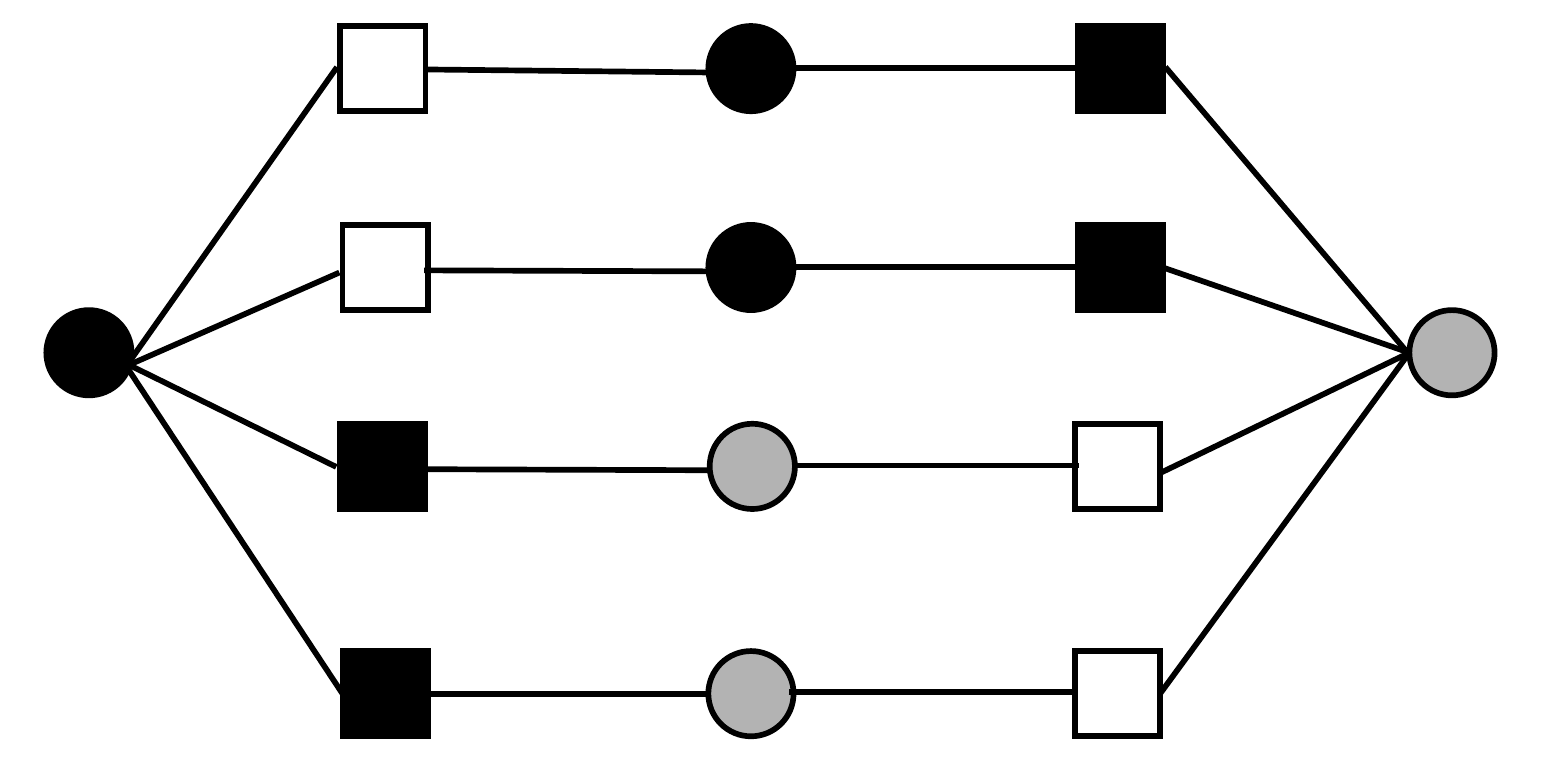}
        \label{fig:temporalStabilizer1}
    }
    
    \end{subfigure}
\caption{The induced subgraph  of a temporal codeword formed by a variable node ($d_v=4$) repeated in two measurement rounds~\subref{fig:temporalStabilizer0}. Two temporal degenerate error patterns of the same weight (\vnode{black} and \vnode{figgray}) in the support of the temporal codeword~\subref{fig:temporalStabilizer1}; unsatisfied check nodes are shown as \cnode{black}. Both error patterns are valid but problematic for an isotropic iterative decoder (as defined in Section~\ref{sec:secv}).}
    \label{fig:temporalStabilizer}
\end{figure}

In Fig.~\ref{fig:temporalStabilizer} we illustrate the induced subgraph of $\mathbf{H}_1^{(2)}$ for $d_v=4$, which shows that for two consecutive measurement rounds, the resulting subset of nodes corresponds to a weight-($d_v+2$) temporal codeword. 
Fig.~\ref{fig:temporalStabilizer}~\subref{fig:temporalStabilizer1} illustrates two temporal degenerate error patterns that are equivalent according to Eq.~(\ref{eq:codespace}) and, like degenerate errors originating from stabilizers  generated by $H_X$, are problematic for iterative decoders~\cite{chytasminsum}.
 
We now show that the temporal degeneracy of $\mathbf{H}^{(r)}_1$ can be characterized by the kernel vectors of $\mathbf{H}^{(r)}_1$, which are lifted versions of the kernel vectors of the base matrix $\mathbf{B}$ and correspond to temporal codewords. Lemma~\ref{lemma:vectors} and Corollary~\ref{cor:bijectionH1} transfer the analysis of the small matrix $\mathbf{B}$ to the actual induced subgraph of a variable node in $\mathbf{H}'$. The lifting shows how kernel vectors of $\mathbf{B}$ become temporal codewords, and the bijection guarantees that no temporal codeword of $\mathbf{H}_1^{(r)}$ is missed.
\begin{lemma}
Let $\mathbf{B}=[\mathbf{B}_L\; \mathbf{B}_R]$ be the base matrix of the phenomenological noise model , corresponding to $r$ measurement rounds. Form the  matrix $\mathbf{H}_1^{(r)}$  by replacing each $1$ in $\mathbf{B}_L$ with an identity block $I_{d_v}$ and each $1$ in $\mathbf{B}_R$ with a $d_v\times 1$ column of ones. 
Then any vector $\mathbf{c} = \begin{bmatrix} \mathbf{z} \\ \mathbf{x} \end{bmatrix} 
$ $\in \operatorname{ker}(\mathbf{B})$ lifts to a vector 
\[
\widehat {\mathbf{c}} = \begin{bmatrix} \mathbf{z}\otimes \mathbf{1}_{d_v} \\[0.5em] \mathbf{x} \end{bmatrix}
\]  $\in \operatorname{ker}(\mathbf{H}_1^{(r)})$.
\label{lemma:vectors}
\end{lemma}
\begin{IEEEproof}
The proof is given in Appendix~\ref{app:lemma:vectors}.
\end{IEEEproof}

\begin{corollary}
The lifting map induces a bijection between $\ker(\mathbf{B})$ and $\ker(\mathbf{H}_1^{(r)})$.
\label{cor:bijectionH1}
\end{corollary}
\begin{IEEEproof}
The proof is given in Appendix~\ref{app:cor:bijectionH1}.
\end{IEEEproof}
Since, by Corollary~\ref{cor:bijectionH1}, there is a one-to-one mapping between the kernel vectors of the base matrix and those of $\mathbf{H}^{(r)}_1$, we can extend the weight enumerator  of Lemma~\ref{lemma1} to the kernel vectors of $\mathbf{H}^{(r)}_1$: 
\begin{equation}
{W(X,Y)} = 1+ \sum_{k=1}^{r-1} \sum_{t=0}^{k-1} \binom{k-1}{t} \binom{r-k}{k-t} \, X^{k d_v} \, Y^{2(k-t)},
\label{eq:enumerateOneVN}
\end{equation}
where  $X$ tracks the weight of the measurement part and $Y$ tracks the weight of the data part. The weight enumerator formula in Eq.~(\ref{eq:enumerateOneVN})  leads  to the following corollary. Corollary~\ref{cor:evenWeight} ensures that the kernel vectors of $\mathbf{H}_1^{(r)}$ do not change the total data error. Hence the two halves of such a vector are equivalent corrections in the sense of Eq.~(\ref{eq:codespace}), and a decoder that cannot choose between them fails even though either choice would be correct.
\begin{corollary}
\label{cor:evenWeight}
All vectors in the kernel of $\mathbf{H}_1^{(r)}$ are temporal codewords.
\end{corollary}
\begin{IEEEproof}
The proof is given in Appendix~\ref{app:cor:evenWeight}.
\end{IEEEproof}

Following the previous notation, we  denote by $\mathbf{c}^{kd_v,2(k-t)}$ a vector in $\ker(\mathbf{H}_1^{(r)})$. 
Weight-($d_v+2$) vectors $\in \operatorname{ker}(\mathbf{H}_1^{(r)})$ are \emph{minimum-weight vectors}, i.e., \emph{minimum-weight temporal codewords} of the code described by $\mathbf{H}_1^{(r)}$, since they are lifted versions of $\mathbf{c}^{1,2} \in \operatorname{ker}(\mathbf{B})$. As a result  of Lemma~\ref{lemma:minimal_xky2},  Lemma~\ref{lemma:decompose_xky2}, and Lemma~\ref{lemma:vectors},  we obtain the following corollaries.
\begin{corollary}
Any $\mathbf{c}^{kd_v,2}$ vector of $\mathbf{H}_1^{(r)}$ with $1\le k\le r-1$ is minimal.
\label{cor:minimalONEvn}
\end{corollary}
\begin{IEEEproof}
    Follows from Lemma~\ref{lemma:minimal_xky2}, lifting arguments of Lemma~\ref{lemma:vectors} and Corollary~\ref{cor:bijectionH1}.
\end{IEEEproof}
\begin{corollary}
The space $\ker(\mathbf{H}_1^{(r)})$ is spanned by the weight-($d_v+2$) temporal codewords. Hence, increasing the number of rounds $r$ does not introduce lower-weight minimal vectors.
    \label{cor:dv2GeneratesALL}
\end{corollary}
\begin{IEEEproof}
    Follows from Lemma~\ref{lemma:decompose_xky2}, lifting arguments of Lemma~\ref{lemma:vectors} and Corollary~\ref{cor:bijectionH1}.
\end{IEEEproof}

In Fig.~\ref{fig:lattice} we illustrate all temporal codewords in the kernel  of $\mathbf{H}_1^{(4)}$. Notice that they are  all linear combinations of weight-($d_v+2$) temporal codewords. Fig.~\ref{fig:vector11} illustrates the induced subgraph  of $\mathbf{H}_1^{(4)}$ for a $d_v=3$ code. Based on Eq.~(\ref{eq:enumerateOneVN}), $\mathbf{H}_1^{(4)}$ contains weight-$5$, weight-$8$, weight-$10$, and weight-$11$ temporal codewords. In particular, Fig.~\ref{fig:vector11}~\subref{fig:h14_cdw} illustrates the induced subgraph of the weight-$11$ temporal codeword, which has the highest weight. Notice that it is a minimal vector, since its support does not strictly contain the support of any other nonzero kernel vector.

\begin{figure}
    \centering
\includegraphics[width=1\linewidth]{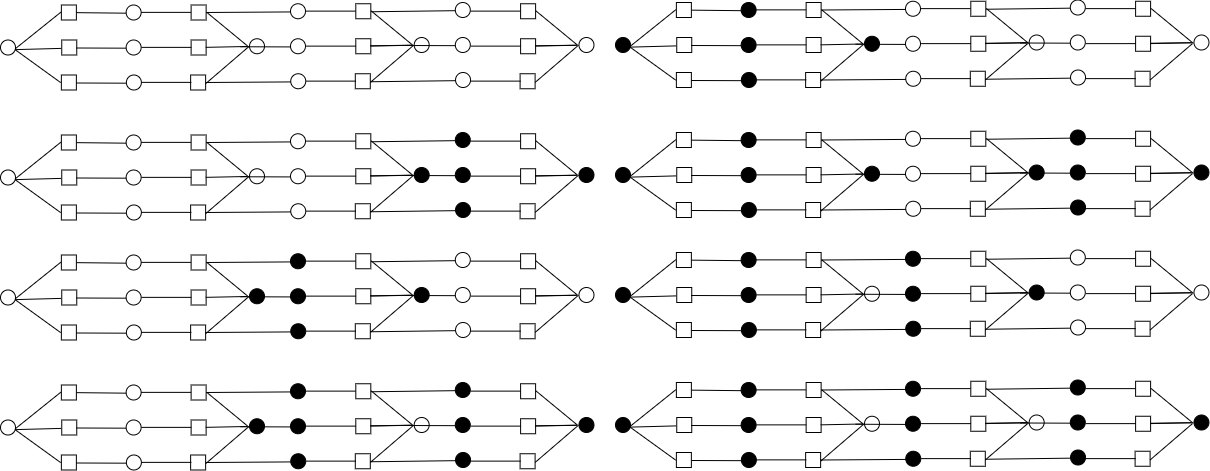}
    \caption{All temporal codewords of $\mathbf{H}_1^{(4)}$. The configurations are arranged in two columns (four on the left, four on the right). Traversing them top-to-bottom, left-to-right, each vector is uniquely labeled by a binary sequence of length $r-1=3$, starting from $000$ to $111$ in lexicographic order.} 
    \label{fig:lattice}
\end{figure} 
\begin{figure}[h]
     \centering
     \begin{subfigure}[]
     {
        \centering
        \includegraphics[width=0.8\linewidth]{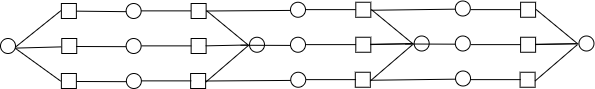}
         \label{fig:h14}
    }
    \end{subfigure}
    \begin{subfigure}[]
    {
        \centering
       \includegraphics[width=0.8\linewidth]{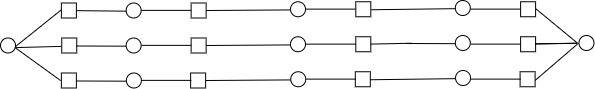}
        \label{fig:h14_cdw}
    }
    
    \end{subfigure}
\caption{The induced subgraph  of a temporal codeword formed by a variable node ($d_v=3$) repeated in four measurement rounds~\subref{fig:h14}. According to Eq.~(\ref{eq:enumerateOneVN}), the highest weight temporal codeword contained in the subgraph is a weight-$11$ vector, which is depicted in~\subref{fig:h14_cdw}.}
    \label{fig:vector11}
\end{figure}

The results so far describe which temporal codewords exist. To quantify their effect on decoding, we next count the error patterns they support. The relevant patterns are those of half the weight of the codeword that cannot be reduced by adding any of its constituent generators, since such a pattern and its complement are equally likely and produce the same syndrome. Definition~\ref{def:locallymin} makes this precise, and Theorem~\ref{theorem1} shows that their number depends only on $k$ and $d_v$, and not on the blocklength or the distance of the code.
\begin{defn}[Locally minimum half-weight error patterns]
Let $\mathbf D=\bigoplus_{i=1}^{k}\mathbf c_i\in\ker(\mathbf H')$
be an even-weight temporal codeword with a specified decomposition into
elementary temporal codewords $\mathbf c_i$. Define
\[
    \mathcal G_D
      :=\operatorname{span}_{\mathbb F_2}
        \{\mathbf c_1,\ldots,\mathbf c_k\},
    \qquad
    \mathcal R_D
      :=\bigcup_{i=1}^{k}\operatorname{supp}(\mathbf c_i).
\]
An error pattern $\mathbf e$ is a locally minimum half-weight error pattern associated with this decomposition if $\operatorname{supp}(\mathbf e)\subseteq\operatorname{supp}(\mathbf D)$, and
\[
    w(\mathbf e)=\frac{w(\mathbf D)}{2},
    \qquad
    w(\mathbf e\oplus\mathbf f)\geq w(\mathbf e),~~\forall \mathbf f\in\mathcal G_D.
\]
\label{def:locallymin}
\end{defn}

Every competing pattern $\mathbf e\oplus\mathbf f$ has the same syndrome as $\mathbf e$ and is supported within $\mathcal R_D$. Thus, the comparison includes coordinates shared by constituent codewords that cancel from $\mathbf D$, while testing only the subspace $\mathcal G_D$ specified by its decomposition. The following theorem counts locally minimum half-weight error patterns in the support of even-weight temporal codewords of $\mathbf{H}_1^{(r)}$.

For $2\le k\le r-1$, let $\mathcal{A}_k$ denote the set of admissible weight distributions $(w_1,\dots,w_k)\in\mathbb{Z}_{\ge 0}^{k}$ that satisfy
\begin{equation}
{\begin{aligned}
&\sum_{i=1}^{k}w_i=\frac{kd_v+2}{2},\\
&\frac{d_v}{2}\le w_1,w_k\le \frac{d_v+2}{2},\\
&\frac{d_v-2}{2}\le w_i\le \frac{d_v+2}{2},\quad 2\le i\le k-1.
\end{aligned}}
\label{eq:admissibleX}
\end{equation}
\begin{theorem}
Let $2\le k\le r-1$, let $\mathbf{c}^{kd_v,2}$ be a minimal vector of even weight $k d_v + 2$, and let $N(k)$ denote the number of locally minimum half-weight error patterns of weight $(k d_v + 2)/2$ contained in $\mathbf{c}^{kd_v,2}$.
Then
\begin{equation}
{N(k)\le
\sum_{(w_1,\dots,w_k)\in\mathcal{A}_k}
\binom{d_v+1}{w_1}
\biggl(\prod_{i=2}^{k-1}\binom{d_v}{w_i}\biggr)
\binom{d_v+1}{w_k}.}
\label{eq:exactNk}
\end{equation}
Consequently,
for fixed $d_v$,
\begin{equation}
N(k)\leq2^{O(kd_v)}.
\end{equation}

\textbf{Remark.} Over $r$ measurement rounds, the total number of locally minimum half-weight patterns of this type satisfies
\[
N_{\mathrm{total}}(k,r) \le (r-k)\,N(k),
\]
and therefore
\[
N_{\mathrm{total}}(k,r) \leq (r-k)\,2^{O(kd_v)}.
\]
\label{theorem1}
\end{theorem}
\begin{IEEEproof}[Proof sketch]
From Corollary~\ref{cor:dv2GeneratesALL}, a weight-$(kd_v+2)$ minimal
vector can be written as
\begin{equation}
\mathbf{c}^{kd_v,2}
=
\mathbf{c}_1\oplus \mathbf{c}_2\oplus \cdots \oplus \mathbf{c}_k,
\label{eq:decomposition}
\end{equation}
where each $\mathbf{c}_i$ is a weight-$(d_v+2)$ temporal codeword of
$\mathbf{H}_1^{(r)}$.

Let $\mathcal{V}_i:=\operatorname{supp}(\mathbf{c}_i)$, $i\in[k]$.
The support of $\mathbf{c}^{kd_v,2}$ can be decomposed into disjoint sets
\begin{equation}
\operatorname{supp}(\mathbf{c}^{kd_v,2})
=
\mathcal{V}
=
\bigsqcup_{i=1}^{k}\mathcal{U}_i,
\label{eq:Vdisjoint}
\end{equation}
where
\[
|\mathcal{U}_1|=|\mathcal{U}_k|=d_v+1,
\qquad
|\mathcal{U}_i|=d_v,\quad 2\le i\le k-1.
\]

For an error pattern $\mathbf{e}$ of weight
$w(\mathbf{c}^{kd_v,2})/2$, define
\[
w_i
:=
|\operatorname{supp}(\mathbf e)\cap\mathcal U_i|,
\qquad i\in[k].
\]
Flipping any proper
combination of the constituent temporal codewords must not produce a
lower-weight syndrome-equivalent error. Applying this condition to the
singleton partitions yields the necessary constraints
\[
\frac{d_v}{2}\le w_1,w_k\le \frac{d_v+2}{2},
\]
and
\[
\frac{d_v-2}{2}\le w_i\le \frac{d_v+2}{2},
\qquad 2\le i\le k-1,
\]
together with
\[
\sum_{i=1}^{k}w_i=\frac{kd_v+2}{2}.
\]

Since the sets $\mathcal U_i$ are disjoint, a fixed weight distribution
$(w_1,\ldots,w_k)$ satisfying these constraints contributes
\[
\binom{d_v+1}{w_1}
\left(\prod_{i=2}^{k-1}\binom{d_v}{w_i}\right)
\binom{d_v+1}{w_k}
\]
possible error supports. Summing over all such distributions gives
Eq.~(\ref{eq:exactNk}). The complete derivation of the constraints and
the asymptotic bound is provided in Appendix~\ref{app:theorem1}.
\end{IEEEproof}

An illustration of the mechanism of the proof is shown in Fig.~\ref{fig:proofMechanism}. According to Theorem~\ref{theorem1}, there are $2\binom{5}{3}\binom{5}{2}=200$ weight-$5$ temporal degenerate error patterns that satisfy Definition~\ref{def:locallymin} and are supported on the weight-$10$ temporal codeword of the figure. The count here is exact since there is only one partition.
\begin{figure}[htpb]
     \centering
     \begin{subfigure}[]
     {
        \centering
        \includegraphics[width=0.8\linewidth]{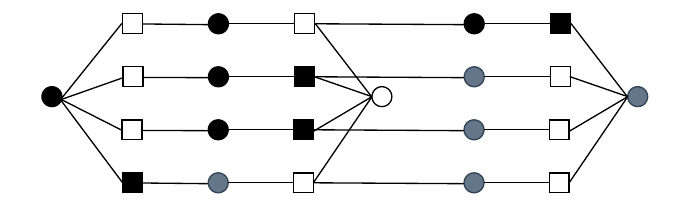}
         \label{fig:proof1}
    }
    \end{subfigure}
    \begin{subfigure}[]
    {
        \centering
       \includegraphics[width=0.8\linewidth]{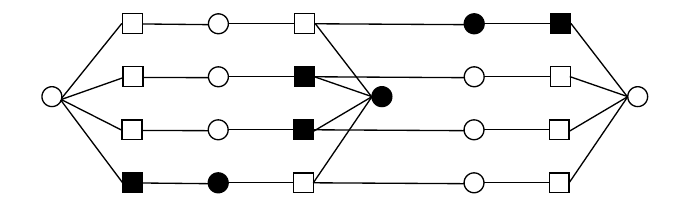}
        \label{fig:proof2}
    }
    
    \end{subfigure}
\caption[Example of two temporal degenerate error patterns.]{Example of two temporal degenerate error patterns $\mathbf{e}, \mathbf{f}$ (colored  black and gray, respectively) that together form the support of a weight-10 temporal codeword~\subref{fig:proof1}. As shown in the proof of Theorem~\ref{theorem1}, since the supports of $\mathbf{e}, \mathbf{f}$ cover more than  half of each weight-$6$ temporal codeword, there exists a weight-$3$ error pattern $\mathbf{f}'$~\subref{fig:proof2}\ that matches the syndrome. In that case,  even though the weight of $\mathbf{e}$ is half the weight of the vector, $\mathbf{e}$ is not considered locally minimum half-weight.}
    \label{fig:proofMechanism}
\end{figure}
The proof technique of Theorem~\ref{theorem1} is also used in the following sections. 

\subsection{Harmful temporal degenerate error patterns}
Based on the analysis of the previous subsection, we characterize all temporal codewords in the kernel of $\mathbf{H}'$ (Definition~\ref{definition:definitionSpatialTemporal}).

Lemma~\ref{lemma:temporalgenerator} extends the single-variable analysis to the full space-time matrix $\mathbf{H}'$. It shows that every temporal codeword of $\mathbf{H}'$ is a sum of embedded weight-$(d_v+2)$ generators. Consequently, temporal degeneracy is determined locally by the variable-node degree, independently of the code construction.

\begin{lemma}
Every temporal codeword $\mathbf D\in\ker(\mathbf H')$ is a linear
combination of embedded weight-$(d_v+2)$ temporal codeword generators.
\label{lemma:temporalgenerator}
\end{lemma}

\begin{IEEEproof}
The proof is given in Appendix~\ref{app:lemma:temporalgenerator}.
\end{IEEEproof}

The preceding lemma shows that all temporal codewords in $\ker(\mathbf H')$ are generated by embedded copies of the weight-$(d_v+2)$ temporal codewords of $\mathbf H_1^{(r)}$. We refer to these embedded weight-$(d_v+2)$ vectors as temporal codeword generators.
We now count the lowest-weight potentially problematic temporal degenerate error patterns. Based on the previous analysis, such errors are supported on even-weight temporal codewords. For odd $d_v$, which includes the BB codes of Section~\ref{sec:secvi}, the generators have odd weight $d_v+2$ and cannot be split into two halves of equal weight. Lemma~\ref{lem:min-even-temporal} identifies the smallest even-weight temporal codewords, of weight $2d_v+2$. This is why the failing weight for odd $d_v$ is $d_v+1$, e.g., four for BB codes. An example of a weight-$(2d_v+2)$ temporal codeword found in $d=3$ codes, including BB codes, is shown in Fig.~\ref{fig:SmallestForOdd}.
The following lemma establishes a result used in the subsequent theorems.
\begin{lemma}
\label{lem:min-even-temporal}
Let $\mathbf H$ be a $(d_v,d_c)$-regular parity-check matrix with $r\geq 2$,
$d_v>2$ odd and a Tanner graph of girth larger than four. Then every
nonzero even-weight temporal codeword
$\mathbf D\in\ker(\mathbf H')$ satisfies
\[
    w(\mathbf D)\ge 2d_v+2.
\]
\end{lemma}

\begin{IEEEproof}
The proof is given in Appendix~\ref{app:lem:min-even-temporal}.
\end{IEEEproof}
We are interested in locally minimum half-weight error patterns (Definition~\ref{def:locallymin}) supported on temporal codewords of even weight and
 nonzero measurement part, 
when such codewords exist.  We refer to such error patterns as \emph{locally minimum half-weight errors supported on minimum even-weight temporal codewords}. Theorem~\ref{theorem:upper boundsTemporal} turns the preceding results into explicit counts. Its upper bounds, $\overline N_{\mathrm{even}}$ and $\overline N_{\mathrm{odd}}$, are matched by lower bounds of the same order in Section~\ref{sec:secv} (Corollary~\ref{cor:decoder_independent_temporal}), which together establish that the number of harmful patterns grows as $\Theta(nr)$.
 \begin{figure}
     \centering
     \includegraphics[width=0.5\linewidth]{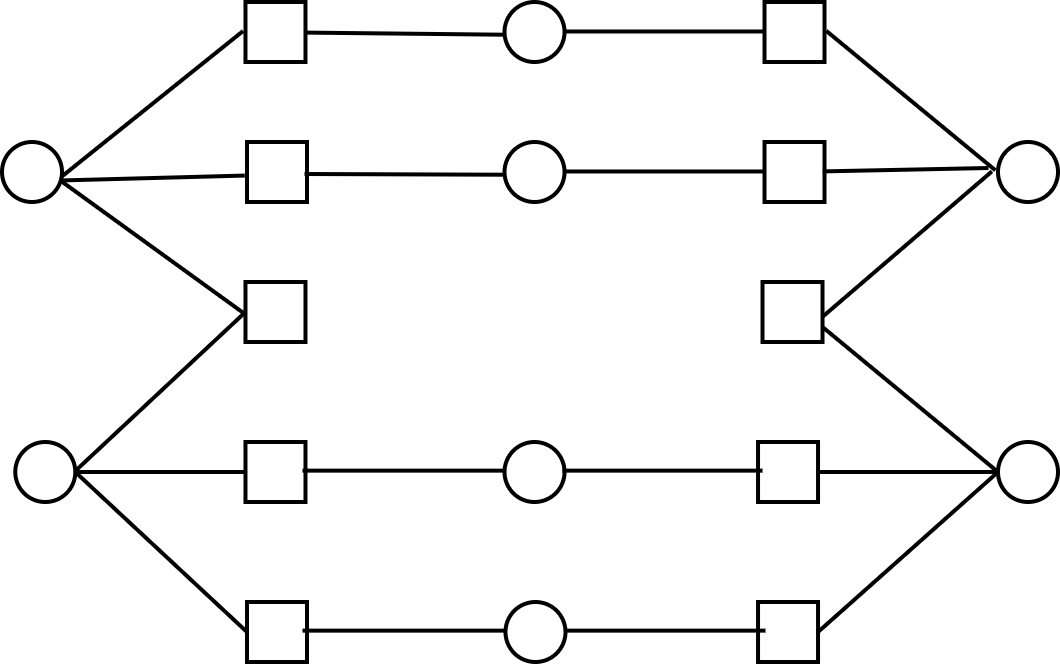}
     \caption{A weight-($2d_v+2$) temporal codeword ($d_v=3$) obtained by two temporal codeword generators sharing a measurement node, or alternatively by two variable nodes of $\mathbf{H}$ sharing a check node repeated for two measurement rounds.}
     \label{fig:SmallestForOdd}
 \end{figure}
\begin{theorem}
   Consider the parity-check matrix $\mathbf{H}'$ obtained from 
a $(d_v,d_c)$-regular matrix $\mathbf{H}$ with girth larger than four under the phenomenological noise model with $r \geq 2$ measurement rounds and $d_v,d_c>2$. We enumerate  locally minimum half-weight errors supported on minimum even-weight temporal codewords with nonzero measurement parts. If $d_v$ is even\red{,} these errors have weight $\frac{d_v+2}{2}$, and the number of incidences is \[\overline N_{\mathrm{even}}=\binom{d_v+2}{(d_v+2)/2}n(r-1)\textcolor{red}{.}\]\red{\st{.}} If the variable degree of $\mathbf{H}$ is odd, these errors have weight $d_v+1$, and  the number of incidences  is \[\overline N_{\mathrm{odd}}=\binom{d_v+1}{(d_v+1)/2}^2n(r-2)+\frac{1}{2}\binom{d_v+1}{(d_v+1)/2}^2nd_v(d_c-1)(r-1)\textcolor{red}{.}\]\red{\st{.}}
    
Consequently, $\overline N_{\mathrm{even}}$ and
$\overline N_{\mathrm{odd}}$ upper-bound the corresponding numbers of
distinct error patterns.
\label{theorem:upper boundsTemporal}
\end{theorem}
\begin{IEEEproof}
The proof is given in Appendix~\ref{app:theorem:upper boundsTemporal}.
\end{IEEEproof}
Intuitively, just as spatial degeneracy arises from linear combinations of rows of $\mathbf{H}_X$,
temporal degeneracy  arises from linear combinations of weight-($d_v+2$) temporal codeword generators. In the following subsection, we study the spatial-temporal case ($\mathbf{h}\neq \mathbf{0}$ in Definition~\ref{definition:definitionSpatialTemporal}).
\subsection{Spatial-Temporal Degeneracy}
\label{sec:sectemporalspatial}
We now focus on a more general case. Let
$\mathbf D
=
(\mathbf z^{(1)},\ldots,\mathbf z^{(r-1)}
\mid
\mathbf x^{(1)},\ldots,\mathbf x^{(r)})
\in \ker(\mathbf H')$
be a  kernel vector. We define $\Pi(\mathbf D)
:=
\sum_{t=1}^{r}\mathbf x^{(t)} \pmod 2.$
Temporal codewords correspond to the case $\Pi(\mathbf D)=\mathbf 0$.
In this subsection, we consider the case in which
$\mathbf D$ has nonzero measurement part and $\Pi(\mathbf D)=\mathbf h\neq \mathbf 0,$
where $\mathbf h$ is a stabilizer of the underlying code. We refer
to such vectors as \emph{mixed spatial-temporal codewords}.

Lemma~\ref{lem:spatiotemporal_generator} shows that mixed spatial-temporal codewords are sums of a stabilizer placed in one round and a temporal codeword. Hence no new type of generator arises, but spatial and temporal degeneracy combine, and a decoder must break the symmetries of both. This motivates the space-time anisotropic decoders of Section~\ref{sec:secvi}.

\begin{lemma}
\label{lem:spatiotemporal_generator}
Let $\mathbf D\in\ker(\mathbf H')$ satisfy
\[
\Pi(\mathbf D)=\mathbf h,
\]
where $\Pi$ is the parity of the data part and $\mathbf h$ is a  stabilizer, so that $\mathbf H\mathbf h=0$.
Then
\[
\mathbf D=\mathbf S_{\mathbf h}\oplus \mathbf T,
\]
where $\mathbf S_{\mathbf h}$ is a  stabilizer placed in one data
round and $\mathbf T$ is a temporal codeword generated by embedded
weight-$(d_v+2)$ temporal codeword generators.
\end{lemma}
\begin{IEEEproof}
The proof is given in Appendix~\ref{app:lem:spatiotemporal_generator}.
\end{IEEEproof}
Fig.~\ref{fig:spatialtemporalOdd} depicts a  mixed spatial-temporal codeword for a regular two-block code with $d_v=3$ and girth larger than four.

\begin{figure}[htbp]
     \centering
     \begin{subfigure}[]
     {
        \centering
        \includegraphics[width=0.4\linewidth]{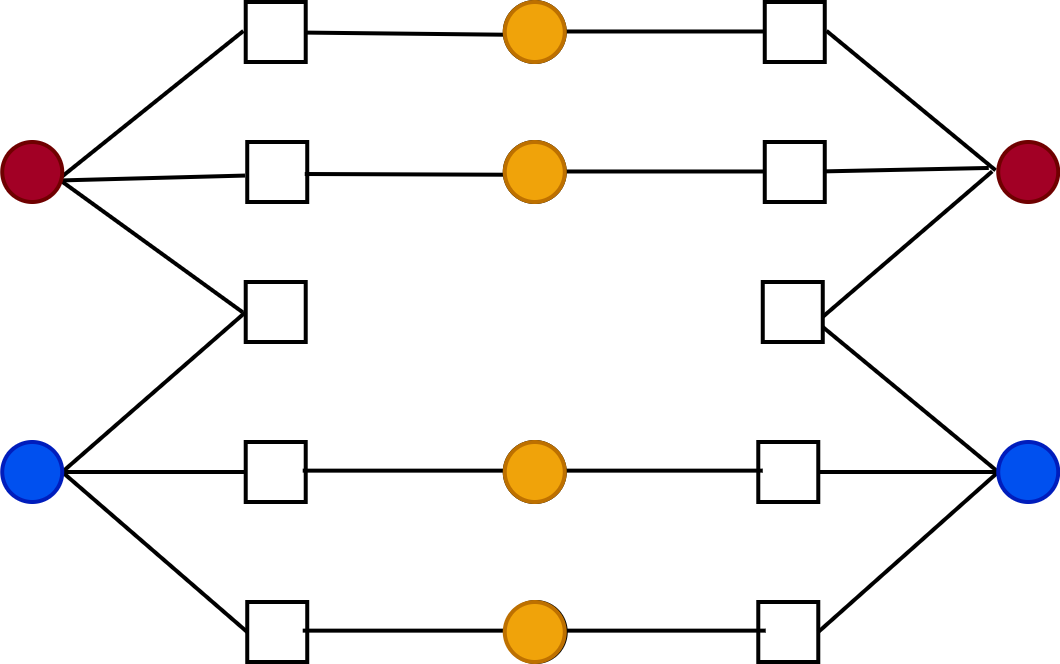}
         \label{fig:spatialtemporalOdd0}
    }
    \end{subfigure}
    \begin{subfigure}[]
    {
        \centering
       \includegraphics[width=0.4\linewidth]{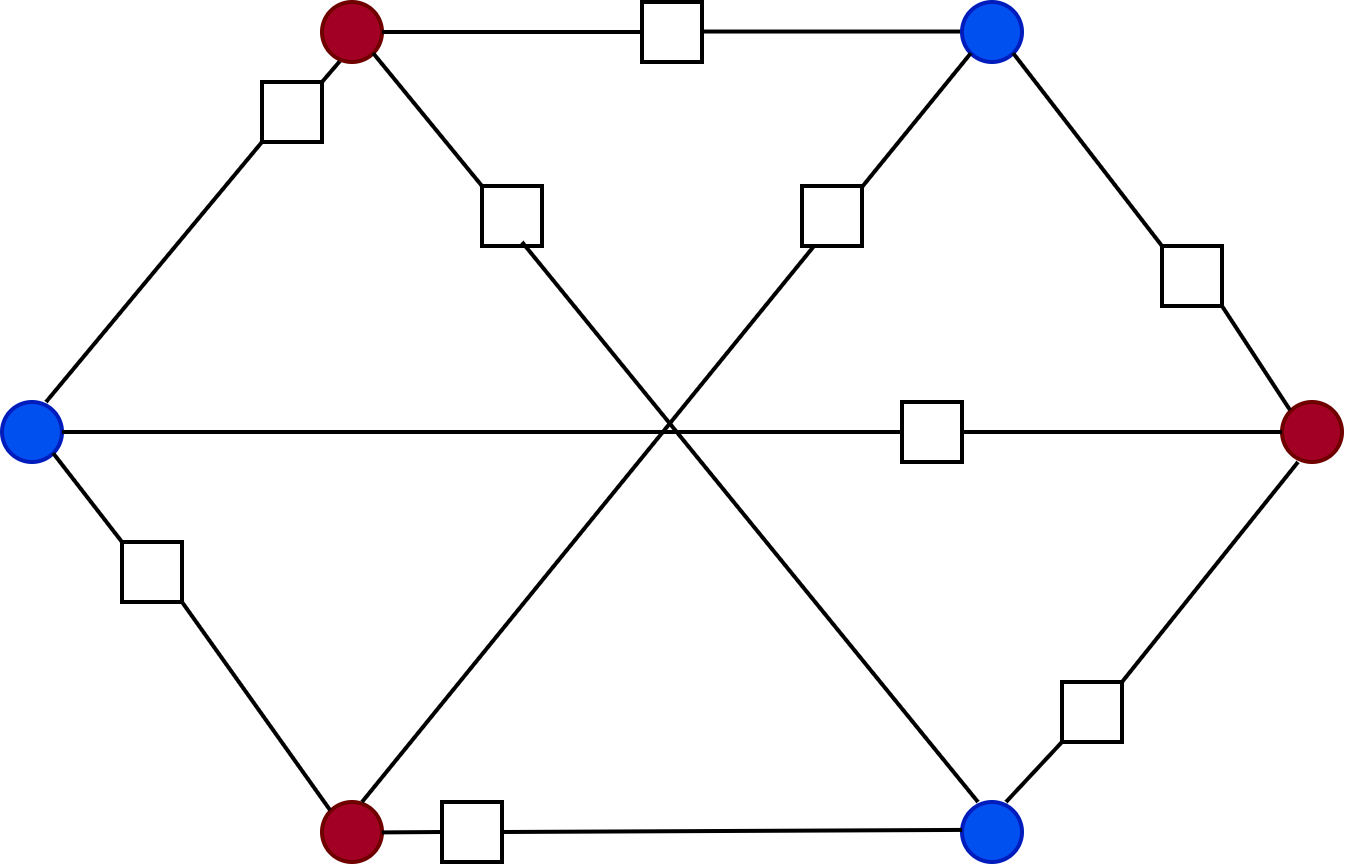}
        \label{fig:spatialtemporalOdd1}
    }
    
    \end{subfigure}

    \begin{subfigure}[]
     {
        \centering
        \includegraphics[width=0.4\linewidth]{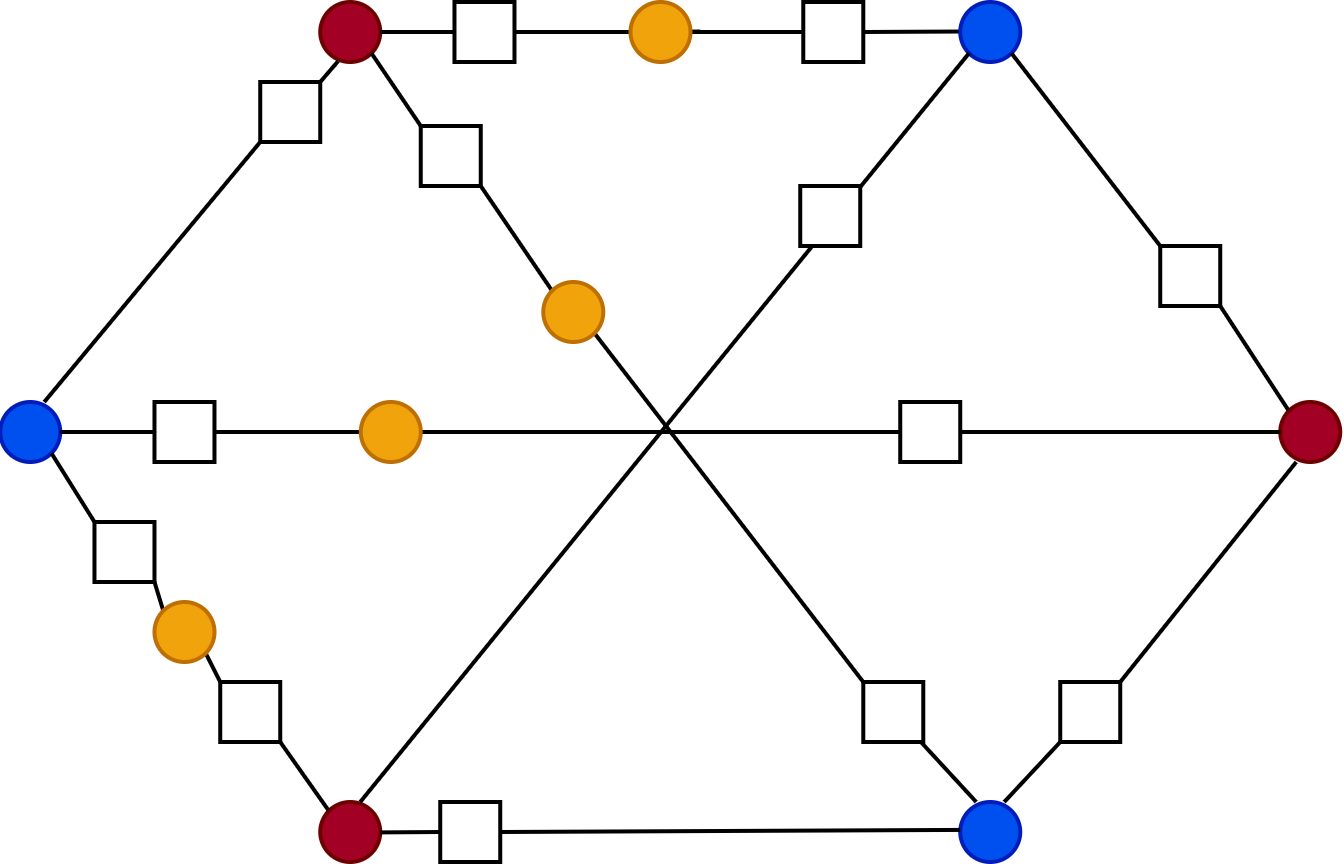}
         \label{fig:spatialtemporalOdd2}
    }
    \end{subfigure}
\caption{A weight-$10$ mixed spatial-temporal codeword~\subref{fig:spatialtemporalOdd2} obtained by combining a weight-$8$ temporal codeword~\subref{fig:spatialtemporalOdd0} and a weight-$6$  stabilizer~\subref{fig:spatialtemporalOdd1}. This configuration is found in BB codes. Data variable nodes are represented by \vnode{figred}, \vnode{figblue}, depending on which block they belong to, and measurement nodes are represented by \vnode{figyellow}. }
    \label{fig:spatialtemporalOdd}
\end{figure}
\section{Effect of temporal degeneracy on iterative decoding}
\label{sec:secv}
Theorem~\ref{theorem:upper boundsTemporal} gives decoder-independent upper bounds on the number of locally minimum half-weight error patterns supported on minimum even-weight temporal codewords. We now identify subsets of these patterns that cause symmetry-induced decoding failures for deterministic isotropic message-passing decoders under flooding scheduling, symmetric initialization, and the isolation assumption. We first introduce graph automorphisms and isotropic decoding to explain the mechanism responsible for these failures.
\begin{defn}[Graph automorphisms]
Let $G=(\mathcal V,\mathcal C,\mathcal E)$ be an induced subgraph of a Tanner graph. A \emph{type-preserving automorphism} of $G$ is a bijection
\[
\psi:\mathcal V\sqcup\mathcal C \to \mathcal V\sqcup\mathcal C
\]
such that
\begin{enumerate}
    \item $\psi(\mathcal V)=\mathcal V$ and $\psi(\mathcal C)=\mathcal C$,
    \item $(v,c)\in\mathcal E$ if and only if $(\psi(v),\psi(c))\in\mathcal E$, for $v \in \mathcal V$ and $c \in \mathcal C$ .
\end{enumerate}
The set of all such automorphisms is denoted by $\operatorname{Aut}(G)$.
If the check nodes carry syndrome labels $\boldsymbol{\sigma}=(\sigma_c), {c\in\mathcal C}$, we define the syndrome-label-preserving automorphism group as
\[
\operatorname{Aut}(G,\boldsymbol{\sigma})
:=
\{
\psi\in\operatorname{Aut}(G):
\sigma_{\psi(c)}=\sigma_c
\text{ for every }c\in\mathcal C
\}.
\]
\end{defn}
We consider iterative message-passing decoders operating on a Tanner graph $G=(\mathcal V,\mathcal C,\mathcal E)$  associated with a
parity-check matrix $\mathbf{H}$ and a syndrome vector $\boldsymbol{\sigma}$. For a check node $c \in \mathcal{C}$ and a variable node $v \in \mathcal{V}$, we denote their neighborhoods by $\mathcal{N}(c)$ and $\mathcal{M}(v)$, respectively.
A \emph{syndrome-based message-passing decoder} is
defined by update functions $\Phi_v$, $\Phi_c$, and a hard-decision function $\Psi_v$ such that, for all iterations $\ell\ge1$,
the messages evolve according to the flooding schedule (all variable and subsequently check nodes are updated in parallel):
\begin{align}
\nu^{(\ell)}_{v\to c}
&=
\Phi_v\!\left(
\lambda,\,
\{\mu^{(\ell-1)}_{c'\to v} : c' \in \mathcal{M}(v)\setminus\{c\}\}
\right),
\label{eq:isotropic_v}
\\
\mu^{(\ell)}_{c\to v}
&=
\Phi_c\!\left(
\sigma_c,\,
\{\nu^{(\ell)}_{v'\to c} : v' \in \mathcal{N}(c)\setminus\{v\}\}
\right),
\label{eq:isotropic_c}
\\
\widehat e_v^{(\ell)}
&=
\Psi_v\left(
\lambda,
\left\{
\mu_{c\to v}^{(\ell)}:c\in\mathcal M(v)
\right\}
\right)
\in\{0,1\},
\label{eq:isotropic_dec}
\end{align}
where $\nu^{(\ell)}_{v\to c}$ and $\mu^{(\ell)}_{c\to v}$ denote the variable-to-check and
check-to-variable messages, respectively,  and
$\widehat e_v^{(\ell)}$ denotes the binary hard decision at variable node $v$. The scalar $\lambda$ denotes a \emph{uniform prior} shared by all variable nodes. Initialize $\nu^{(1)}_{v\to c}=\lambda$. At iteration $\ell=1$,
compute the check-to-variable messages using Eq.~(\ref{eq:isotropic_c}) and hard decisions using Eq.~(\ref{eq:isotropic_dec}). For every $\ell\ge2$, first compute the variable-to-check messages using Eq.~(\ref{eq:isotropic_v}),
then the check-to-variable messages using Eq.~(\ref{eq:isotropic_c}), and finally the hard decisions using Eq.~(\ref{eq:isotropic_dec}).

\begin{defn}[Isotropic syndrome-based decoder]
\label{defn:isotropic}
A syndrome-based message-passing decoder defined by Eqs.~(\ref{eq:isotropic_v})--(\ref{eq:isotropic_dec}) is called \emph{isotropic} if:

\begin{enumerate}
    \item 
    The same update functions $\Phi_v$, $\Psi_v$, and $\Phi_c$ are used at all variable nodes and
    check nodes, respectively.

    \item 
    The update functions $\Phi_v$, $\Psi_v$, and $\Phi_c$ depend only on the \emph{multiset} of
    incoming messages and are invariant under permutations of their arguments.
    \item All update and decision functions are deterministic. In particular, any tie in the hard decision rule is resolved using the same deterministic tie-breaking rule at every variable node.
\end{enumerate}
\end{defn}

Examples of  isotropic message-passing decoders used in practice are BP, MS~\cite{05CDEFH} and Gallager~A/B~\cite{63Gallager} under flooding schedule.
The failures of iterative decoders are characterized by the notion of trapping sets, which we consider here from a purely graph-theoretic perspective independent of the channel and the decoder (based on~\cite{ontology}). An $(a, b)$ trapping set is a configuration
of $a$ variable nodes, for which the induced subgraph is connected and has $b$ odd-degree check
nodes.
For the trapping sets under consideration here, it holds that $b=0$, since they correspond to codewords. 

For the analysis below we will employ the isolation assumption~\cite{iso}. This assumption states that all nodes outside
the trapping set of interest, which in our case will correspond to a temporal codeword, are error-free and have converged, meaning that their 
 variable-to-check messages are of high positive amplitudes. This means that external messages preserve equivariance under the induced-subgraph
automorphisms considered below. In our case external messages are excluded during the check-to-variable updates of the decoder (Eq.~(\ref{eq:isotropic_c})).

In the following Lemma we show a property of isotropic decoding that is going to be useful for the next Theorem.

\begin{lemma}[Automorphism equivariance of isotropic decoding]
\label{lem:aut_equivariance}
Let $G=(\mathcal V,\mathcal C,\mathcal E)$ be an induced Tanner subgraph with syndrome labels
$\boldsymbol{\sigma}$ on check nodes, and consider a syndrome-based isotropic decoder
defined by \eqref{eq:isotropic_v}--\eqref{eq:isotropic_dec} under flooding schedule.

Assume the isolation assumption holds for an induced subgraph under test.
Let $\psi$ be a type-preserving automorphism of $G$ that preserves syndrome labels, i.e.,
$\sigma_c=\sigma_{\psi(c)}$ for all check nodes $c\in\mathcal C$.

Initialize messages consistently with $\psi$ (e.g., all variable priors equal to $\lambda$,
and boundary messages invariant under $\psi$).
Then for every iteration $\ell\ge 1$ and every edge $(v,c)\in\mathcal E$,
\[
\nu^{(\ell)}_{v\to c}=\nu^{(\ell)}_{\psi(v)\to \psi(c)},
\qquad
\mu^{(\ell)}_{c\to v}=\mu^{(\ell)}_{\psi(c)\to \psi(v)},
\]
and for every variable node $v\in\mathcal V$,
\[
\widehat e_v^{(\ell)}
=
\widehat e_{\psi(v)}^{(\ell)}.
\]
Equivalently, the hard decision vector restricted to
$\mathcal V$ is invariant under $\psi$ at every iteration, i.e., if $\psi(v)=v'$ then the hard decisions at $v$ and $v'$ coincide
for all iterations.
\end{lemma}

\begin{IEEEproof}
We will use induction on the $\ell$-th iteration.

\emph{Base step.} At $\ell=1$, variable-to-check messages are initialized as
$\nu^{(1)}_{v\to c}=\lambda$
so they are identical on all edges and hence invariant under $\psi$.
Using $\sigma_c=\sigma_{\psi(c)}$ and that $\psi$ preserves neighborhoods,
the multisets
$
\{\nu^{(1)}_{v'\to c}: v'\in\mathcal N(c)\setminus\{v\}\}
$
and
$
\{\nu^{(1)}_{\psi(v')\to \psi(c)}:\psi(v')\in\mathcal N(\psi(c))\setminus\{\psi(v)\}\}
$
agree up to permutation. Since $\Phi_c$ depends only on the multiset of its arguments,
we obtain
$
\mu^{(1)}_{c\to v}=\mu^{(1)}_{\psi(c)\to \psi(v)}.
$

\emph{Induction step.} Assume the statement holds at iteration $\ell \geq 1$ for all messages.
Fix $(v,c)\in\mathcal E$. By the induction hypothesis and the fact that $\psi$ maps
$\mathcal M(v)\setminus\{c\}$ bijectively to $\mathcal M(\psi(v))\setminus\{\psi(c)\}$,
the multiset
$
\{\mu^{(\ell)}_{c'\to v}: c'\in\mathcal M(v)\setminus\{c\}\}
$
equals the multiset
$
\{\mu^{(\ell)}_{\psi(c')\to \psi(v)}: \psi(c')\in\mathcal M(\psi(v))\setminus\{\psi(c)\}\}
$
up to permutation. By permutation invariance of $\Phi_v$,
$
\nu^{(\ell+1)}_{v\to c}=\nu^{(\ell+1)}_{\psi(v)\to \psi(c)}.
$
Applying the same argument to the check update \eqref{eq:isotropic_c} and using
$\sigma_c=\sigma_{\psi(c)}$ yields
$
\mu^{(\ell+1)}_{c\to v}=\mu^{(\ell+1)}_{\psi(c)\to \psi(v)}.
$

Finally, the hard decisions \eqref{eq:isotropic_dec} satisfy
$\widehat e_v^{(\ell)} = \widehat e_{\psi(v)}^{(\ell)}$ by the same multiset argument.
\end{IEEEproof}

The next theorem establishes a symmetry-induced local failure of syndrome-based isotropic decoders. Under the isolation assumption, the decoder's hard decision
cannot coincide with either of the two syndrome-compatible half-weight error patterns associated with each counted temporal codeword. Consequently, it cannot produce a syndrome-consistent hard decision supported entirely on the corresponding isolated subgraph.

To state the resulting lower bounds compactly, define
\begin{equation}
\underline{N}_{\mathrm{even}}
:=
2\binom{d_v}{d_v/2}n(r-1),
\label{eq:lowerEven}
\end{equation}
and
\begin{equation}
\begin{split}
\underline{N}_{\mathrm{odd}}
:={}&
2\binom{d_v}{(d_v-1)/2}^{2}n(r-2)\\
&+
\left[
2
\binom{d_v-1}{(d_v+1)/2}
\binom{d_v-1}{(d_v-3)/2}
+
4
\binom{d_v-1}{(d_v-1)/2}^{2}
\right]\\
&\hspace{1.5cm}\times
\frac{1}{2}nd_v(d_c-1)(r-1).
\end{split}
\label{eq:lowerOdd}
\end{equation}

\begin{theorem}
Let $\mathbf H$ be a $(d_v,d_c)$-regular parity-check matrix
with $d_v>2$ and girth larger than four, and let $\mathbf H'$
be the corresponding phenomenological parity-check matrix for
$r\geq 2$ measurement rounds. Consider a deterministic syndrome-based isotropic decoder under the flooding schedule, as in Definition~\ref{defn:isotropic}. Under the isolation assumption, there exist at least $\underline N_{\mathrm{even}}$ distinct locally minimum weight-$(d_v+2)/2$ error patterns if $d_v$ is even, and at
least $\underline N_{\mathrm{odd}}$ distinct locally minimum weight-$(d_v+1)$ error patterns if $d_v$ is odd, with the
following property.

For every counted error pattern $\mathbf e$, let $\mathbf D$ be its associated minimum even-weight temporal codeword, let
$\mathcal V:=\operatorname{supp}(\mathbf D)$, and
define its complementary degenerate error pattern by
\[
\mathbf f:=\mathbf e\oplus\mathbf D.
\]
For the input syndrome
\[
\boldsymbol\sigma
=
\mathbf H'\mathbf e
=
\mathbf H'\mathbf f,
\]
the restriction of the decoder's hard decision to
$\mathcal V$ satisfies
\[
\left.
\widehat{\mathbf e}^{(\ell)}
\right|_{\mathcal V}
\notin
\{\mathbf e,\mathbf f\}
\]
at every iteration $\ell$. Moreover, since $\mathbf e$ and $\mathbf f$ are the only error patterns supported on
$\mathcal V$ that satisfy this syndrome, the decoder cannot produce a syndrome-consistent hard decision supported entirely on the isolated temporal-codeword subgraph.
\label{theorem:isotropic-obstruction}
\end{theorem}
\begin{IEEEproof}[Proof sketch]
Let $\mathbf e$ be one of the counted half-weight error patterns, with its support contained in the support of $\mathbf D$, its associated minimum even-weight temporal codeword, and define
\[
\mathbf f:=\mathbf e\oplus\mathbf D.
\]
Since $\mathbf D\in\ker(\mathbf H')$, the two error patterns produce the same syndrome:
\[
\mathbf H'\mathbf e=\mathbf H'\mathbf f.
\]

Let
\[
\mathcal V_e:=\operatorname{supp}(\mathbf e),
\qquad
\mathcal V_f:=\operatorname{supp}(\mathbf f)
             =\mathcal V\setminus\mathcal V_e.
\]
For every pattern counted in
Appendix~\ref{app:isotropicCounting}, there exists a type-preserving, syndrome-label-preserving automorphism $\psi$ of the induced
subgraph satisfying
\[
\psi(\mathcal V_e)=\mathcal V_f,
\qquad
\psi(\mathcal V_f)=\mathcal V_e.
\]
The explicit constructions and enumeration are given in Appendix~\ref{app:isotropicCounting}.

By Lemma~\ref{lem:aut_equivariance}, the hard decisions at variable nodes related by $\psi$ are equal. Consequently, the support of the
hard decision within $\mathcal V$ is invariant under $\psi$:
\[
\psi\!\left(
\operatorname{supp}(\widehat{\mathbf e}^{(\ell)})
\cap\mathcal V
\right)
=
\operatorname{supp}(\widehat{\mathbf e}^{(\ell)})
\cap\mathcal V.
\]
It therefore cannot equal $\mathcal V_e$, because $\psi(\mathcal V_e)=\mathcal V_f\ne\mathcal V_e$. The same argument shows that it cannot equal $\mathcal V_f$.
Hence,
\[
\left.
\widehat{\mathbf e}^{(\ell)}
\right|_{\mathcal V}
\notin
\{\mathbf e,\mathbf f\}
\]
at every iteration.
Furthermore, the subgraph induced by
$\mathcal V$ is connected, and every adjacent check
node has degree two within this subgraph. Therefore,
\[
\ker\!\left(
\left.\mathbf H'\right|_{\mathcal V}
\right)
=
\operatorname{span}_{\mathbb F_2}\{\mathbf D\}.
\]
Indeed, if an error pattern $\mathbf g$ is supported on
$\mathcal V$ and satisfies
$
\mathbf H'\mathbf g
=
\mathbf H'\mathbf e,
$
then
\[
\mathbf g\oplus\mathbf e
\in
\ker\!\left(
\left.\mathbf H'\right|_{\mathcal V}
\right).
\]
Consequently,
\[
\mathbf g\oplus\mathbf e
\in\{\mathbf 0,\mathbf D\},
\]
and hence
\[
\mathbf g\in\{\mathbf e,\mathbf f\}.
\]
Thus, the decoder cannot produce a syndrome-consistent hard decision supported entirely on the isolated subgraph.

Finally, Appendix~\ref{app:isotropicCounting} constructs the required syndrome-label-preserving automorphisms and shows that they exist for at least $\underline N_{\mathrm{even}}$ distinct error patterns when $d_v$ is even and at least
$\underline N_{\mathrm{odd}}$ distinct error patterns when
$d_v$ is odd.
\end{IEEEproof}

Theorem~\ref{theorem:upper boundsTemporal} gives decoder-agnostic
upper bounds on the numbers of distinct locally minimum half-weight
error patterns supported on minimum even-weight temporal codewords. Theorem~\ref{theorem:isotropic-obstruction}, on the other hand, gives
lower bounds on the subset of these patterns that produce a local decoding failure for isotropic decoders under the isolation assumption. Here we note that Theorem~\ref{theorem:isotropic-obstruction} does not predict the logical error rate of the full graph which depends on the underlying decoding dynamics. Combining the two theorems
gives the following bounds.

Let $X_w$ denote the number of distinct weight-$w$ locally minimum half-weight errors supported on minimum even-weight temporal codewords with nonzero measurement parts for which, under the isolation and initialization assumptions of Lemma~\ref{lem:aut_equivariance}, an isotropic decoder cannot produce a syndrome-consistent hard decision supported on the isolated subgraph.

\begin{corollary}
\label{cor:decoder_independent_temporal}
Under the phenomenological noise model and the isolation assumption,
\[
\underline{N}_{\mathrm{even}}\;\le\; X_{\frac{d_v+2}{2}} \;\le\; \overline{N}_{\mathrm{even}}
\quad\text{for even } d_v,
\]
and
\[
\underline{N}_{\mathrm{odd}} \;\le\; X_{d_v+1} \;\le\; \overline{N}_{\mathrm{odd}}
\quad\text{for odd } d_v.
\]
In particular, $X_w=\Theta(nr)$ for fixed $d_v$ and $d_c$.
\end{corollary}

\begin{IEEEproof} 
This follows directly from Theorem~\ref{theorem:upper boundsTemporal} and Theorem~\ref{theorem:isotropic-obstruction} by observing that
the counted degenerate supports have constant weight and occur with multiplicity
$\Theta(nr)$ for $(d_v,d_c)$-regular codes of girth larger than four and $d_v>2,r\geq2$.
\end{IEEEproof}

In other words, there exist constant-weight temporal degenerate error patterns whose number grows linearly with the blocklength $n$ and the number of measurement rounds $r$, and on which the hard decision of syndrome-based isotropic iterative decoders fails. These patterns are imposed by repeated syndrome measurements and persist independently of the decoder realization, as long as it is isotropic (Definition~\ref{defn:isotropic}), and of the code, as long as it is $(d_v,d_c)$-regular with girth larger than four.

The temporal codewords responsible for these failures are introduced by the measurement model and occur across codes that satisfy the stated assumptions. Since the sparse Tanner graphs of QLDPC codes make message passing attractive for low-complexity decoding, we address these symmetries directly within the message update rules. Assigning different update rules to nodes exchanged by the relevant syndrome-preserving automorphisms can break these symmetries, realized via \emph{anisotropic decoding} as defined below. We note that symmetry breaking does not by itself guarantee successful decoding.

\begin{defn}[Anisotropic syndrome-based decoder]
\label{defn:anisotropic}
Let $G=(\mathcal V,\mathcal C,\mathcal E)$ be a Tanner graph with syndrome labels
$\boldsymbol{\sigma}$. Let $\kappa:\mathcal V\to\{1,\dots,K\}$ be a labeling
(coloring) of variable nodes. An anisotropic syndrome-based decoder
uses update functions $\Phi_v^{(k)}$ at variable node $v$ whenever
$\kappa(v)=k$, $k \in [K]$ and check and hard decision updates $\Phi_c,\Psi_v$ as before.
\end{defn}
Define the label-preserving automorphism group
\begin{equation}
\mathrm{Aut}(G,\boldsymbol{\sigma},\kappa)
=
\left\{
\psi\in\mathrm{Aut}(G,\boldsymbol{\sigma}) :
\kappa(v)=\kappa(\psi(v))\ \forall v\in\mathcal V
\right\}.
\end{equation}
\begin{prop}[Symmetry breaking via variable node labeling]
\label{prop:prop1}
Let $G=(\mathcal V,\mathcal C,\mathcal E)$ be a Tanner graph with
syndrome labels $\boldsymbol{\sigma}$.
Let $\kappa:\mathcal V\to\{1,\dots,K\}$ be a labeling (coloring)
of the variable nodes, and consider a syndrome-based decoder
that applies update rule $\Phi_v^{(\kappa(v))}$ at variable node $v$,
while check and hard decision updates use $\Phi_c$ and $\Psi_v$ as before.

Assume symmetric initialization and the isolation assumption.
Then, message evolution is equivariant with respect to the
label-preserving automorphism group
\[
\mathrm{Aut}(G,\boldsymbol{\sigma},\kappa)
=
\{\psi\in\mathrm{Aut}(G,\boldsymbol{\sigma}) :
\kappa(v)=\kappa(\psi(v))\ \forall v\}.
\]

In particular, for a temporal codeword
support,
if the labeling $\kappa$ is chosen such that no
$\psi\in\mathrm{Aut}(G,\boldsymbol{\sigma},\kappa)$
satisfies
\[
\psi(\mathcal V_e)=\mathcal V\setminus \mathcal V_e
\]
for the corresponding harmful supports, then the condition for the symmetry-based local failure of
Theorem~\ref{theorem:isotropic-obstruction} does not apply.
\label{prop:coloring}
\end{prop}
\begin{IEEEproof}
The proof follows the same induction argument as in
Lemma~\ref{lem:aut_equivariance}. Since
$\kappa(v)=\kappa(\psi(v))$ for every
$\psi\in\mathrm{Aut}(G,\boldsymbol{\sigma},\kappa)$,
corresponding variable nodes apply identical update rules.
Permutation invariance of the update functions then yields
equivariance under $\mathrm{Aut}(G,\boldsymbol{\sigma},\kappa)$.
\end{IEEEproof}

In the following section we will demonstrate a case study illustrating the limitations of iterative decoding under repeated measurements, and a proposed solution based on anisotropic decoding. We note that the design  of an optimal decoding algorithm for the phenomenological noise model is beyond the scope of this paper and is left for future work.

\section{Numerical Results}
\label{sec:secvi}
In this section, we  provide a case study on decoding $X$ errors and measurement errors on the family of BB codes. Both types of errors are independent and introduced by a binary symmetric channel with crossover probability $\alpha$. Decoding is performed via MS-based algorithms which we run for $100$ iterations. We note that in terms of notation, we assume that decoding is applied on the parity check matrix $\mathbf{H}'$. 

Throughout this section, $r$ denotes the total number of measurement rounds, consisting of $r-1$ noisy syndrome rounds followed by one ideal syndrome round. A decoding trial is counted as a failure if no syndrome-consistent estimate is found, or if a syndrome-consistent estimate is returned whose total data residual is not in the rowspace of $\mathbf{H}_X$.

\subsection{Space-time anisotropic min-sum decoding}
BB codes belong to the family of two-block CSS codes~\cite{bravyi2024high,mackay_quantum,PhysRevA.109.022407}. Such codes are defined by two commuting binary matrices $\mathbf{A}_1$ and $\mathbf{A}_2$, satisfying $\mathbf{A_1A_2}=\mathbf{A_2A_1}$. The corresponding parity-check matrices are
\begin{equation}
\mathbf{H}_X = [\mathbf{A}_1, \mathbf{A}_2], \qquad
\mathbf{H}_Z = [\mathbf{A}_2^{T} , \mathbf{A}_1^{T}],
\end{equation}
which ensures $\mathbf{H}_X \mathbf{H}_Z^{T} = \mathbf{A_1A_2} + \mathbf{A_2A_1} = \mathbf{0}$ (mod $2$), and hence the CSS commutativity condition.

In BB codes, the blocks $\mathbf{A_1}$ and $\mathbf{A_2}$ are constructed from binary circulant matrices and as a result, they are sparse with constant row and column weight.

For the BB codes considered here, each of the two blocks has row
weight three and the induced subgraph of each stabilizer generator corresponds to a subdivided $K_{3,3}$ complete bipartite graph. Hence, in order to mitigate spatial degeneracy, anisotropic decoding rules can be assigned independently to variable nodes belonging to the $\mathbf{A_1}$ and $\mathbf{A_2}$ blocks (such as~\cite{chytasminsum}) and across repeated measurement instances.

We use the syndrome-based MS decoder, a low-complexity variant of BP~\cite{05CDEFH}.
A check-to-variable message sent from check node $c$ to variable node $v$ at the $\ell$-th iteration is denoted by $\mu^{(\ell)}_{c \to v}$, whereas a variable-to-check message sent from variable node $v$ to check node $c$ at the $\ell$-th iteration is denoted by $\nu^{(\ell)}_{v \to c}$. The variable node update rule is computed as: 
\begin{equation}
    \nu^{(\ell)}_{v \to c} = \lambda_v + \sum\limits_{c' \in \mathcal{M}(v) \backslash \{ c \}}\mu^{(\ell-1)}_{c' \to v},
    \label{eq:vnu}
\end{equation}
and the check node update rule is computed as: 
\begin{equation}
   \mu^{(\ell)}_{c \to v} = (1-2\sigma_c)\prod\limits_{v' \in \mathcal{N}(c) \backslash \{v\}} \mathrm{sgn}(\nu^{(\ell)}_{v'\to c})\min\limits_{v' \in \mathcal{N}(c) \backslash \{v\}}|\nu^{(\ell)}_{v'\to c}|,
    \label{eq:cnu}
\end{equation}
where $\mathrm{sgn}$ is the sign function, $\sigma_c \in \{0,1\}$,  is the value of the measured syndrome of check node $c$ and $\lambda_v=\lambda=\log({\frac{1-\alpha}{\alpha}})$ is the a priori log-likelihood ratio (LLR) of every variable node. Since MS follows Definition~\ref{defn:isotropic}, we refer to it as isotropic MS decoder.

Authors in~\cite{chytasminsum} analyzed the evolution of the variable-to-check messages of the MS decoder under the condition that $\frac{w}{2}$-weight error patterns are applied to weight-$w$ stabilizers whose induced subgraphs are isomorphic to complete bipartite graphs. Under the isolation assumption, the authors showed that the signs of the variable-to-check messages flip in every iteration and to counter that, they proposed the following update rule which damps oscillating messages only for certain variable nodes:
\begin{equation}
\nu = \lambda_v + \sum\limits_{c' \in \mathcal{M}(v) \backslash \{ c \}}\mu^{(\ell-1)}_{c' \to v},
\label{eq:temp}
\end{equation}
and then, after sign flips are checked, the variable-to-check message is computed as follows:
\begin{equation}
\nu^{(\ell)}_{v \to c} = \begin{cases}
\nu, & \text{if ($\mathrm{sgn}(\nu)=\mathrm{sgn}(\nu^{(\ell-1)}_{v \to c}))$}\\
\nu+\nu^{(\ell-1)}_{v \to c}, & \text{otherwise}.
\end{cases}
\label{eq:Pi}
\end{equation}
In particular, for the case of BB codes the new update rule is applied to variable nodes of one block only. Since, that MS variant follows Definition~\ref{defn:anisotropic}, we  refer to it as space-anisotropic MS decoder. 

Based on Proposition~\ref{prop:coloring}, the analysis of spatial and temporal degenerate error patterns naturally motivates coloring data
variable nodes in both space and time. Since we decode $X$ errors, we use $\mathbf H=\mathbf H_Z=[\mathbf A_2^T,\mathbf A_1^T]$.
For an arbitrary number $r$ of measurement rounds, the diagonal data-check blocks of $\mathbf H'$ are
\begin{equation}
\mathbf H^{(j)}
=
\left[(\mathbf A_2^{(j)})^T,(\mathbf A_1^{(j)})^T\right],
\qquad j\in[r],
\end{equation}
where the superscript identifies a copy of the same code in round $j$. We denote by $\mathcal V_{A_2}^{(j)}$ and $\mathcal V_{A_1}^{(j)}$ the sets
of data variable nodes corresponding to the first and second column blocks, respectively.

For a data variable node $v$ in round $j\in[r]$, we assign the fixed flag
\begin{equation}
f_v=
\begin{cases}
1, & v\in\mathcal V_{A_2}^{(j)},\quad j\text{ odd},\\
1, & v\in\mathcal V_{A_1}^{(j)},\quad j\text{ even},\\
0, & \text{otherwise}.
\end{cases}
\label{eq:spaceTimeFlags}
\end{equation}
All measurement variable nodes are assigned $f_v=0$. Thus, the selected data block alternates between consecutive rounds, while the flags remain
fixed throughout the decoding iterations.

Let $\nu$ denote the raw variable-to-check message in~\eqref{eq:temp}. For every variable node $v$ and neighboring check node $c$, the outgoing
message is computed as
\begin{equation}
\nu^{(\ell)}_{v\to c}
=
\begin{cases}
\nu+f_v\nu^{(\ell-1)}_{v\to c},
& \mathrm{sgn}(\nu)\ne\mathrm{sgn}(\nu^{(\ell-1)}_{v\to c}),\\
\nu, & \text{otherwise}.
\end{cases}
\label{eq:spaceTimeUpdate}
\end{equation}
The sign-change condition is tested separately for each outgoing edge. Consequently, flagged data variable nodes apply~\eqref{eq:Pi}, whereas
unflagged data variable nodes and all measurement variable nodes use the standard MS update~\eqref{eq:vnu}. Check-to-variable messages are computed
using~\eqref{eq:cnu}. As in isotropic MS, messages are initialized as $\nu^{(0)}_{v\to c}=\lambda_v$ and $\mu^{(0)}_{c\to v}=0$, with
$\mathrm{sgn}(0)=0$.

We refer to this deterministic decoder as \emph{space-time anisotropic MS}. For $r=2$, the flagged data nodes are precisely
$\mathcal V_{A_2}^{(1)}\cup\mathcal V_{A_1}^{(2)}$. The coloring distinguishes the two data blocks within each round and
copies of the same data variable in consecutive rounds, with the aim of breaking symmetries associated with spatial and temporal degeneracy.

Fig.~\ref{fig:threshold} shows the performance of isotropic and space-time anisotropic MS decoding for the family of BB codes under phenomenological noise for the minimal setting where one noisy measurement is followed by an ideal one.
Evidently, no threshold behavior is observed under the isotropic MS decoder. On the other hand, space-time anisotropic MS decoding exhibits evidence of threshold-like behavior as it is able to correct multiple low-weight spatial and temporal degenerate error patterns.

\begin{figure}
    \centering
    \includegraphics[width=1\linewidth]{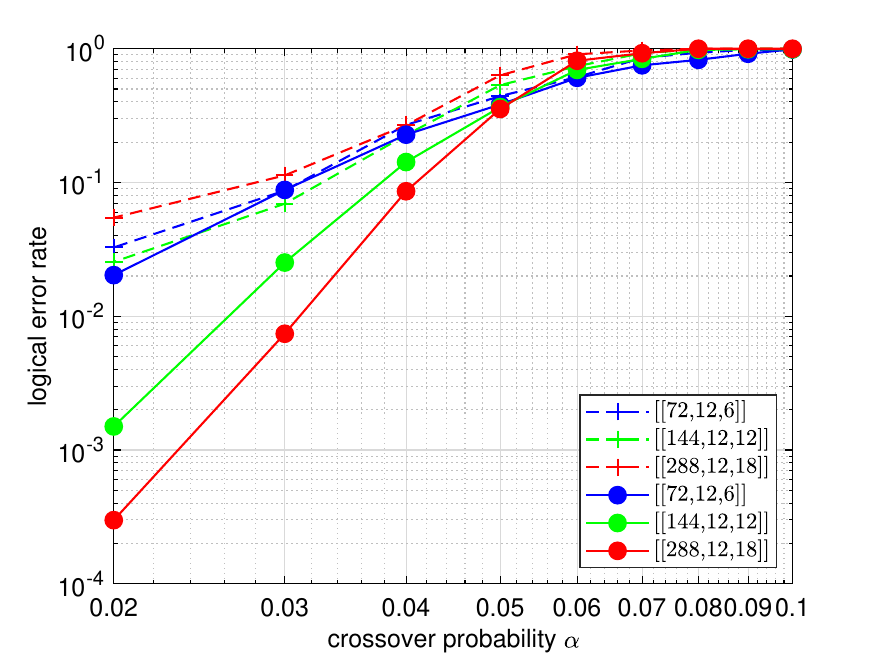}
    \caption{Performance of the family of BB codes under repeated syndrome measurements. Dashed lines correspond to isotropic MS decoding and solid lines to space-time anisotropic MS decoding. Notably, threshold-like behavior can be achieved in linear time by introducing anisotropy in MS decoding. Decoding is performed for $100$ iterations.}
    \label{fig:threshold}
\end{figure}
\begin{figure}[htbp]
     \centering
     \begin{subfigure}[]
     {
        \centering
        \includegraphics[width=0.4\linewidth]{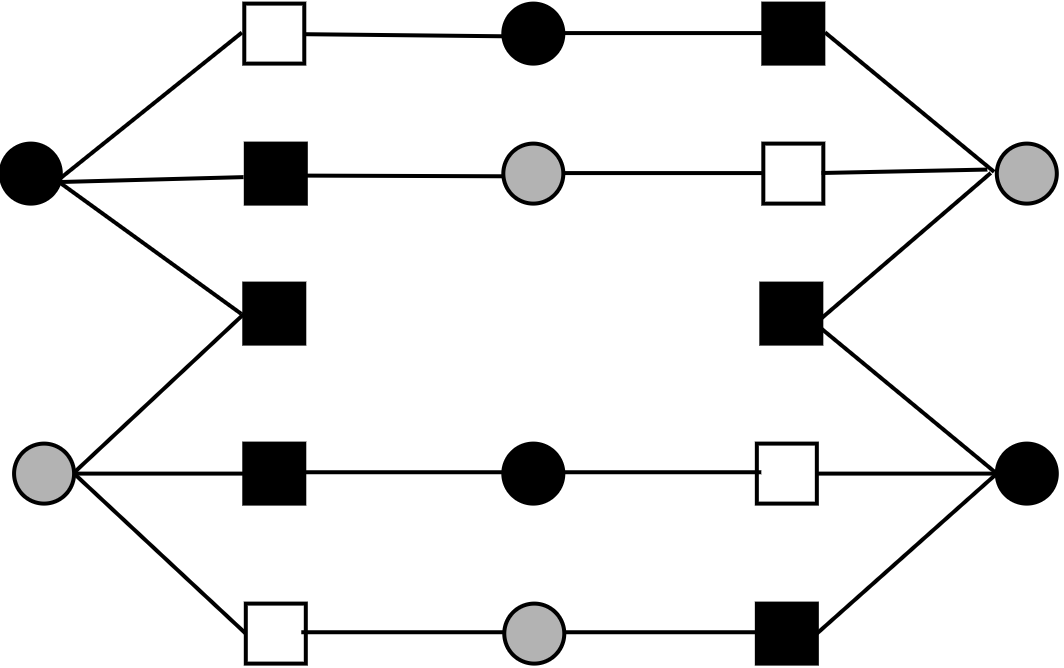}
         \label{fig:problem1}
    }
    \end{subfigure}
    \begin{subfigure}[]
    {
        \centering
       \includegraphics[width=0.4\linewidth]{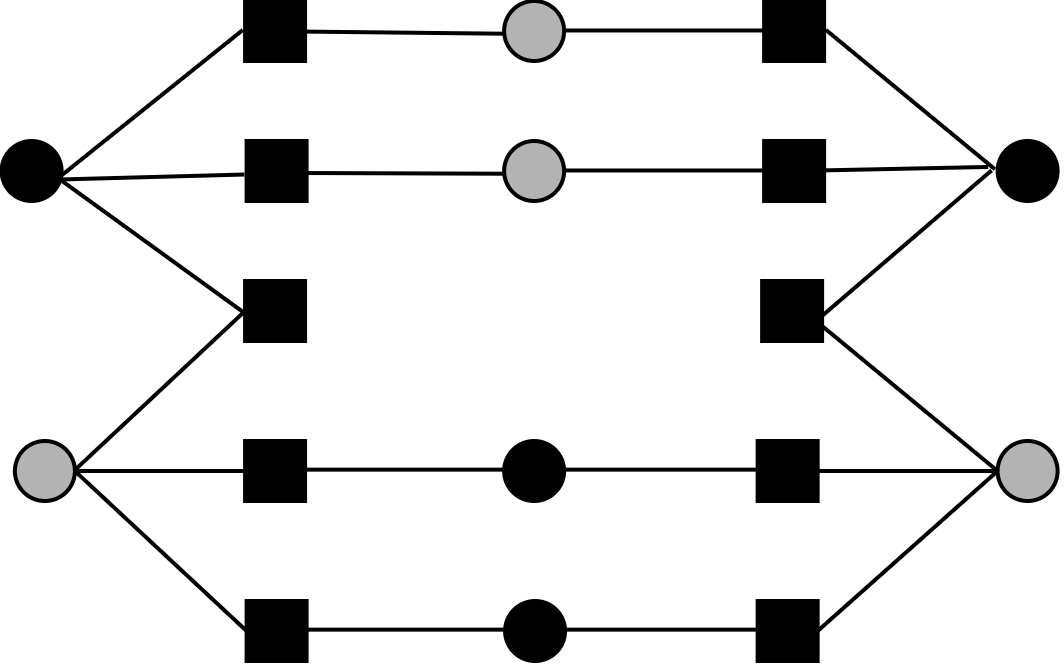}
        \label{fig:problem2}
    }
    
    \end{subfigure}

   \begin{subfigure}[]
     {
        \centering
        \includegraphics[width=0.6\linewidth]{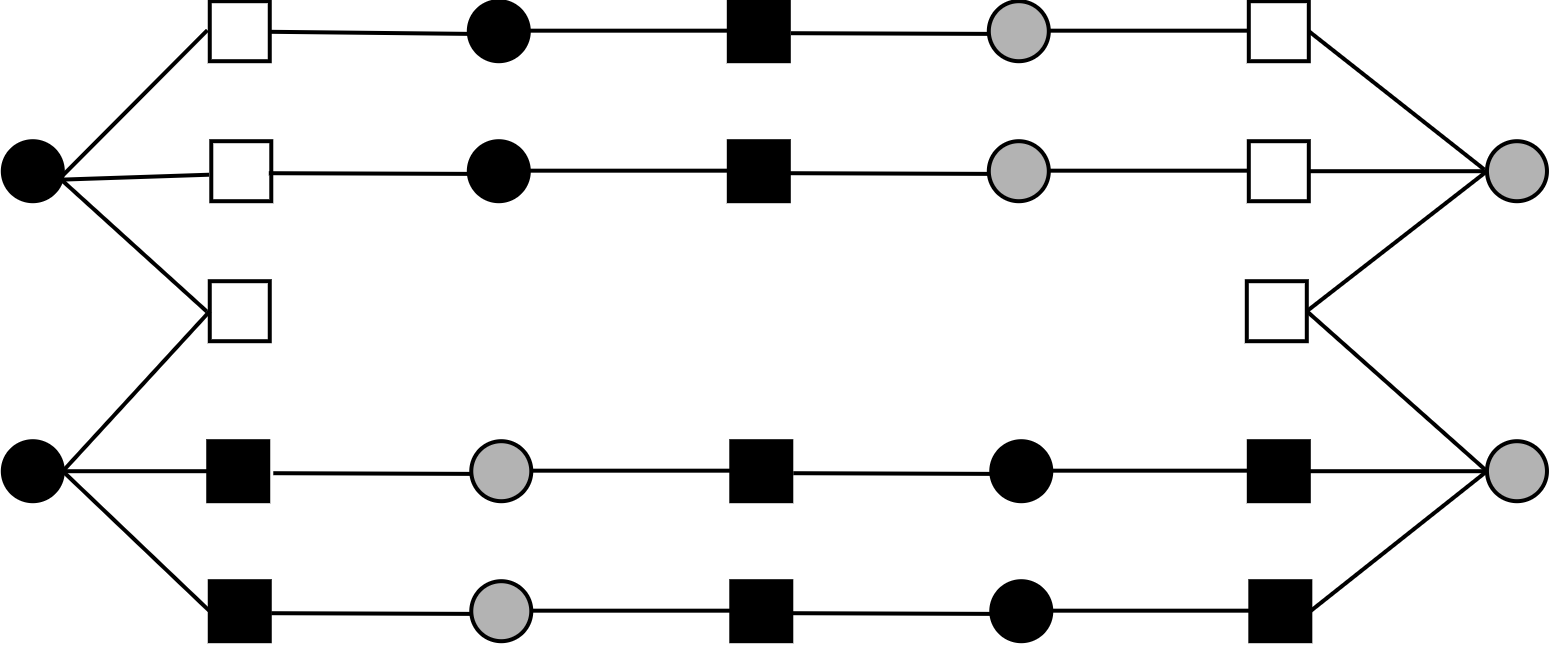}
         \label{fig:problem3}
    }
    \end{subfigure}
\caption{Pairs of degenerate error patterns supported on temporal codewords of a BB code: (a), (b) weight-$4$ patterns on weight-$8$ codewords spanning two measurement rounds, and (c) weight-$6$ patterns on a weight-$12$ codeword spanning three measurement rounds. In each example, \vnode{black} and \vnode{figgray} denote the two complementary error supports, and \cnode{black} denotes
an unsatisfied check node. All data variable nodes in each panel belong to the same block. The deterministic space-time coloring
preserves the automorphisms relating the paired patterns. In (c), all data nodes have identical colors because
the coloring repeats every two rounds. These configurations motivate additional coloring of measurement variable nodes.}
\label{fig:problematic}
\end{figure}

The deterministic space-time coloring does not eliminate all automorphisms relating degenerate error patterns. In
Fig.~\ref{fig:problematic}(a), (b), the paired weight-$4$ patterns are supported on weight-$8$ temporal codewords whose
data variable nodes belong to the same block. Within each round, these data nodes receive the same color, so the automorphisms relating the paired patterns remain
label-preserving.
Similar limitations arise for higher-weight temporal codewords. Fig.~\ref{fig:problematic}(c) shows two complementary weight-$6$
patterns supported on a weight-$12$ temporal codeword spanning three measurement rounds. Since the deterministic coloring of data nodes
repeats every two rounds, all four data nodes have identical colors, and are therefore related by a graph automorphism.

These examples illustrate symmetries that remain unresolved by the deterministic coloring of data nodes and motivate additional
coloring of measurement variable nodes. We next examine the effect of low-weight degenerate error patterns under full-graph
decoding and introduce a syndrome-guided extension of the space-time anisotropic MS decoder to address residual
temporal symmetries.

\subsection{Low-weight error patterns and syndrome-guided decoding}
\label{subsec:lowweight}
We investigate low-weight decoding failures of the $[[288,12,18]]$ BB code using the full space-time Tanner graph of $H'$, without imposing
the isolation assumption. Theorem~\ref{theorem:isotropic-obstruction} identifies failures for isotropic message-passing decoders
under isolation; it does not predict the failure probability on the full graph. Here, we examine the practical behavior of isotropic MS and  space-time anisotropic MS, and use the remaining failures to motivate syndrome-guided coloring of measurement
variable nodes. We consider one and four noisy measurement rounds, each followed by an ideal round, corresponding to $r=2$ and $r=5$, respectively. With
$n=288$ and $m=144$, a complete error history has length $N=nr+m(r-1)=720$ and $2016$, respectively. All MS realizations use normalization factor $1$, the same initialization, and at most $100$ iterations. The failure probability is denoted by $P_{\mathrm{fail}}$.

To examine the temporal degenerate error patterns identified in Theorem~\ref{theorem:isotropic-obstruction} beyond the isolation
assumption, we exhaustively tested the constructed weight-four family on the full space-time Tanner graph, with no additional
errors and a common prior $\alpha_0=10^{-4}$. For $r=2$, isotropic MS fails on $19\,008$ of the $38\,880$ patterns, whereas deterministic space-time anisotropic MS fails on $912$ patterns. For $r=5$, the corresponding number of failures  are $93\,312$ and $6912$  out of $171\,072$ patterns. Since these tests do not impose isolation, the measured failure fractions need not coincide with the isolated-subgraph behavior established by the theorem. At the same prior $\alpha_0=10^{-4}$, we also tested all $n(r-1)$ weight-two patterns consisting of errors on data variable nodes of the weight-$(d_v+2)$ temporal codeword generator. Isotropic MS fails on all $288$ patterns for $r=2$ and all $1152$ patterns for $r=5$, whereas space-time anisotropic MS corrects every tested pattern. These are additional empirical
failures, distinct from the equal-weight configurations analyzed in Theorem~\ref{theorem:isotropic-obstruction}.

To resolve rare low-weight failures under independent data and measurement noise, we use fixed-weight importance sampling~\cite{beverland2025failfasttechniquesprobe}.
Assume a weight-$W$ error pattern, at a crossover probability $\alpha$, its distribution and the failure probability satisfy
\begin{equation}
\begin{aligned}
B_w(\alpha)&=\binom{N}{w}\alpha^w(1-\alpha)^{N-w},\\
P_{\mathrm{fail},\mathcal D}(\alpha)
&=\sum_{w=0}^{N}B_w(\alpha)g_{w,\mathcal D}(\alpha),
\end{aligned}
\label{eq:weightStratumContribution}
\end{equation}
where
$g_{w,\mathcal D}(\alpha)=\Pr(\mathrm{failure}\mid W=w)$ for decoder $\mathcal D$ initialized with the prior corresponding to $\alpha$.
Conditioned on $W=w$, all $\binom{N}{w}$ supports are equally likely. We enumerate these supports exhaustively where feasible; otherwise,
each trial selects $w$ distinct locations uniformly among all $N$ locations, with every other location error-free. Sampled supports are
drawn independently, with replacement across trials. This samples the entire weight stratum, including patterns outside the constructed
temporal family.

For $K_{w,\mathcal D}$ failures in $T_w$ tests, we use $\widehat g_{w,\mathcal D}=K_{w,\mathcal D}/T_w$ and multiply by
$B_w(\alpha)$ to estimate the weight-$w$ contribution. This fraction is exact when all supports of that weight have been tested. For a set
of included weights $\mathcal W$, the partial estimate is
\begin{equation}
\widehat P_{\mathrm{fail},\mathcal D}^{(\mathcal W)}(\alpha)
=\sum_{w\in\mathcal W}B_w(\alpha)
\widehat g_{w,\mathcal D}(\alpha).
\label{eq:partialFailureEstimate}
\end{equation}

For $r=2$, space-time anisotropic MS corrects all error patterns up to weight-$3$ at $\alpha_{\mathrm{IS}}=10^{-3}$ . At weights four, five, and six, we observe $99$, $23$, and $78$
failures in $10^8$, $10^7$, and $10^7$ trials, respectively. The corresponding
estimated contributions are
\[
(\widehat P_4,\widehat P_5,\widehat P_6)
=(5.3705,\,1.7885,\,0.72351)\times10^{-9},
\]
where $\widehat P_w=B_w\widehat g_{w,\mathcal D}$ for the decoder and
prior under consideration. Their sum is
\[
\widehat P_{\mathrm{fail},\mathcal D}^{(\{4,5,6\})}(10^{-3})
=7.8826\times10^{-9},
\]
with a conservative $95\%$ statistical confidence interval
\[
[\,5.7255\times10^{-9},\;1.0652\times10^{-8}\,].
\]
This interval uses a joint Bonferroni correction across the three sampled weights, including those evaluated in separate runs.
Weight four accounts for approximately $68.1\%$ of the measured sum.

For $r=5$ at $\alpha_{\mathrm{IS}}=5\times10^{-3}$, we observe $142$, $28$, and $84$ failures at weights four, five, and six in
$10^8$, $10^7$, and $10^7$ trials, respectively. The estimated contributions are
\[
(\widehat P_4,\widehat P_5,\widehat P_6)
=(0.253901,\,1.012369,\,5.115263)\times10^{-7}.
\]
Their sum is
\[
\widehat P_{\mathrm{fail},\mathcal D}^{(\{4,5,6\})}
(5\times10^{-3})=6.3815\times10^{-7},
\]
with a conservative $95\%$ statistical confidence interval
\[
[\,4.6931\times10^{-7},\;8.4962\times10^{-7}\,],
\]
using a joint Bonferroni correction across the three sampled weights. Weight six accounts for approximately $80.2\%$ of this partial sum,
whereas weight four accounts for approximately $4.0\%$. At this probability, $N\alpha=10.08$, so contributions beyond weight
six may be substantial. The reported sum therefore quantifies the included low-weight contributions rather than a complete
failure-probability estimate.

Fig.~\ref{fig:msExclusionCodes} compares isotropic MS and space-time anisotropic MS. Ordinary Monte Carlo estimates under
independent data and measurement noise are supplemented by the fixed-weight estimates described above, marked IS. These IS values contain only the specified weight contributions.

\begin{figure*}[t]
    \centering
    \begin{minipage}[t]{0.49\textwidth}
        \centering
        \IfFileExists{figures/msSpaceTime_1plus1.pdf}{%
            \includegraphics[width=\linewidth]{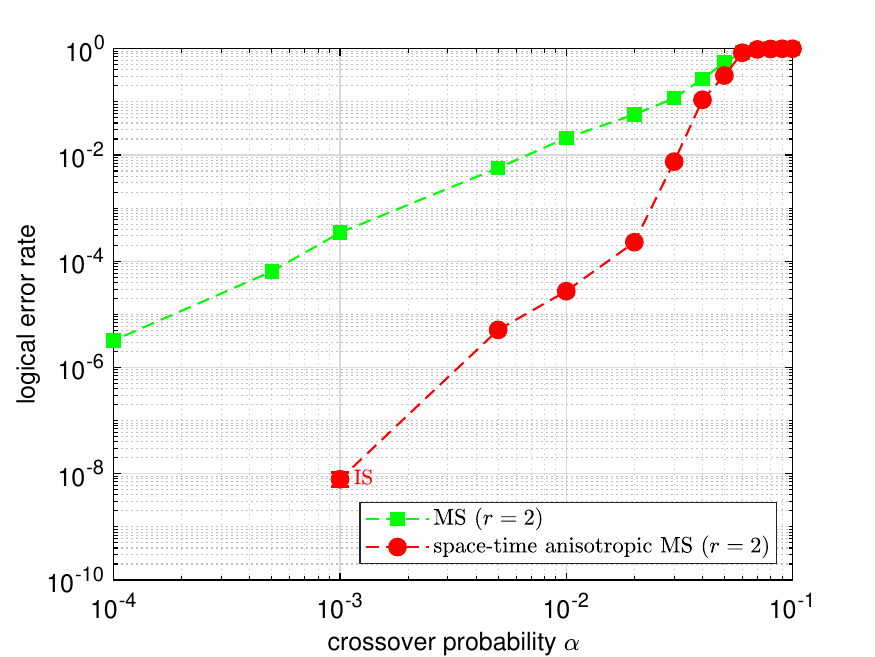}%
        }{%
            \fbox{\parbox[c][0.18\textheight][c]{0.90\linewidth}{%
                \centering\textbf{[TODO: insert the $1+1$ ordinary-noise
                MS versus space-time MS plot with fitted slopes.]}}}%
        }

        \small (a) Two measurement rounds
    \end{minipage}
    \hfill
    \begin{minipage}[t]{0.49\textwidth}
        \centering
        \IfFileExists{figures/msSpaceTime_4plus1.pdf}{%
            \includegraphics[width=\linewidth]{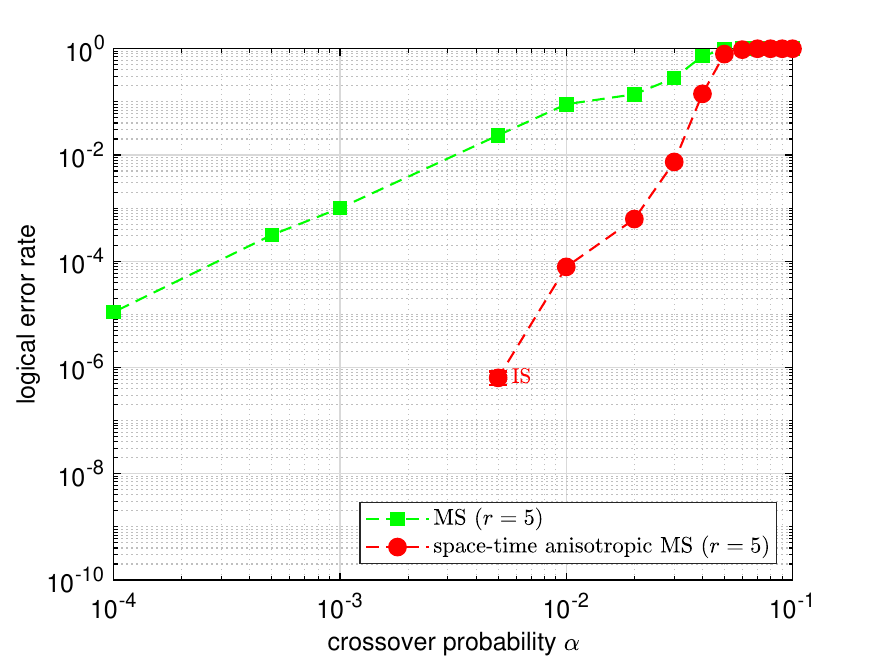}%
        }{%
            \fbox{\parbox[c][0.18\textheight][c]{0.90\linewidth}{%
                \centering\textbf{[TODO: insert the $4+1$ ordinary-noise
                MS versus space-time MS plot with fitted slopes.]}}}%
        }

        \small (b) Five measurement rounds
    \end{minipage}
    \caption{Decoding performance of MS and space-time anisotropic MS for the $[[288,12,18]]$ BB code under (a) one and (b) four noisy
    measurement rounds, each followed by an ideal round. Every realization uses
    normalization factor $1$ and at most $100$ iterations. The points marked IS estimate only the contributions from error weights four through six. Conservative $95\%$ statistical confidence intervals for these partial sums are reported in the text and do not account for contributions from omitted weights.}
    \label{fig:msExclusionCodes}
\end{figure*}

The remaining low-weight failures, together with the residual symmetries illustrated in Fig.~\ref{fig:problematic}, motivate additional coloring
of measurement variable nodes. We retain the deterministic data-node coloring and update rule in~\eqref{eq:spaceTimeFlags}
and~\eqref{eq:spaceTimeUpdate}. The selection of measurement nodes is guided by the input syndrome. Let
\[
\mathcal U_0=\operatorname{supp}(\boldsymbol{\sigma})
\]
be the set of initially unsatisfied checks, and let $\mathcal S$ be the set of measurement variable nodes. Define
\[
\mathcal S^{\mathrm{act}}
=\left\{v\in\mathcal S:
\mathcal M(v)\cap\mathcal U_0\ne\varnothing\right\}.
\]
For each guided realization, every $v\in\mathcal S^{\mathrm{act}}$ is assigned an independent coefficient
\[
\xi_v\sim\operatorname{Unif}(0,1),
\]
sampled once and kept fixed throughout the realization. The same coefficient is used on both outgoing edges of measurement node $v$.

Let $\nu$ denote the raw variable-to-check message in~\eqref{eq:temp}. The update at an active measurement node is
\[
\nu_{v\to c}^{(\ell)}=
\begin{cases}
\nu, & \mathrm{sgn}(\nu)=
\mathrm{sgn}\!\left(\nu_{v\to c}^{(\ell-1)}\right),\\
\nu+\xi_v\nu_{v\to c}^{(\ell-1)}, & \text{otherwise}.
\end{cases}
\]
The sign-reversal condition is tested separately for each outgoing edge, with $\mathrm{sgn}(0)=0$. Measurement nodes outside
$\mathcal S^{\mathrm{act}}$ use the standard MS update. The node-dependent coefficients provide additional coloring while
preserving the deterministic data-node coloring. They can also help break residual symmetries of temporal codewords spanning several
rounds, such as the example in Fig.~\ref{fig:vector11} and Fig.~\ref{fig:problematic}~\subref{fig:problem3}. Symmetry breaking alone does not guarantee successful decoding. We refer to the resulting
decoder as \emph{guided space-time anisotropic MS}.

The diversity set consists of the space-time anisotropic MS decoder and $N_D$ independent guided realizations. Each guided realization samples new coefficients on the same active measurement nodes. Thus, $N_D=0$ recovers the space-time anisotropic MS. All realizations use the original priors, the same initialization, and the same maximum iteration limit. Although the realizations can be executed in parallel, we execute them sequentially in our simulations and stop at the first estimate whose space-time syndrome matches the input syndrome. If no such estimate is found, decoding is declared unsuccessful.

We next evaluate guided space-time anisotropic MS under ordinary phenomenological noise and compare it with MS followed by OSD
post-processing and Relay-BP, including the effects of diversity size and decoding resources.

\begin{figure}[htbp]
    \centering
    \includegraphics[width=\linewidth]{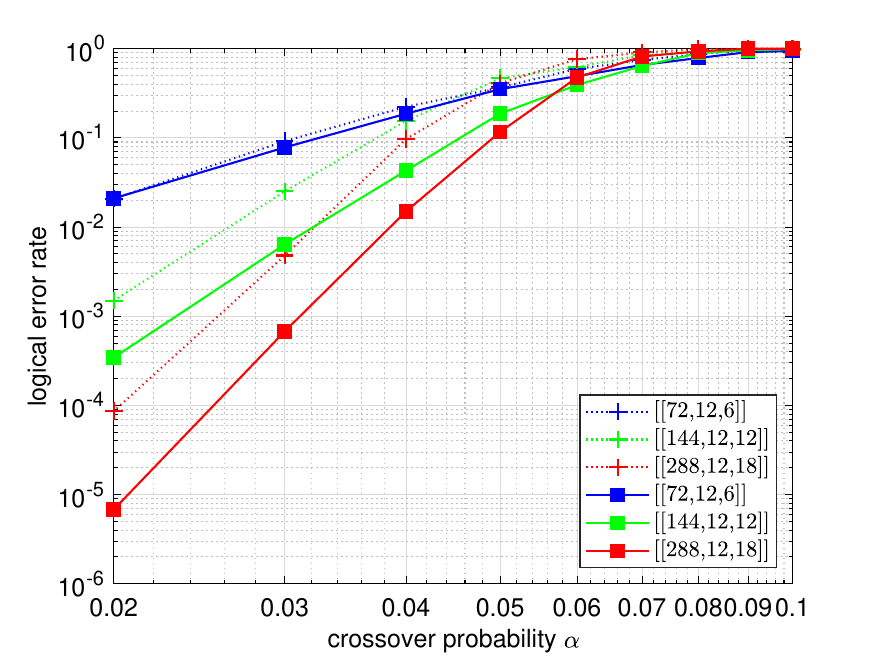}
    \caption{Logical error rates for the family of BB codes under one noisy measurement round followed by an ideal round. Solid lines correspond to guided space-time anisotropic MS with $N_D=9$, and dotted lines to MS followed by OSD-$0$ post-processing. Each MS decoder realization uses a maximum of $100$ iterations.}
    \label{fig:thresholdProposed}
\end{figure}

\subsection{Comparison with OSD Post-Processing and Relay-BP}
\label{subsec:osdcomparison}

\begin{figure*}[t]
    \centering
    \begin{minipage}[t]{0.49\textwidth}
        \centering
        \includegraphics[width=\linewidth]{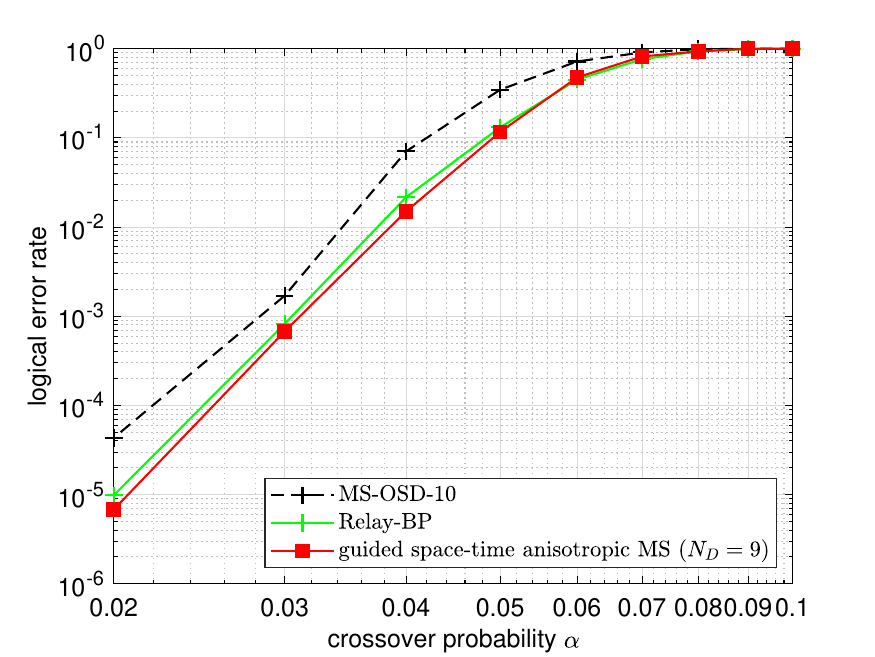}

        \small (a) Two measurement rounds
    \end{minipage}
    \hfill
    \begin{minipage}[t]{0.49\textwidth}
        \centering
        \includegraphics[width=\linewidth]{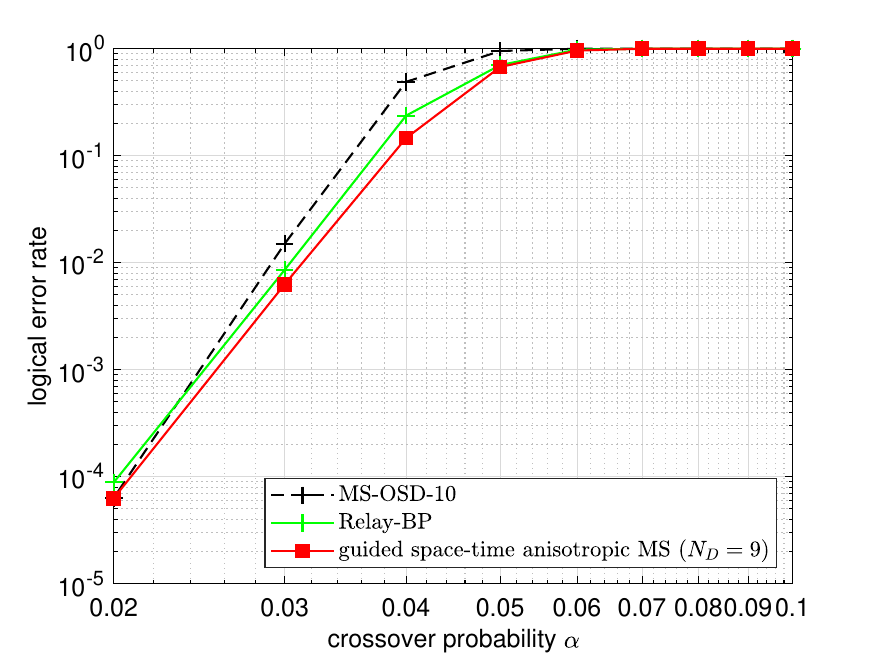}

        \small (b) Five measurement rounds
    \end{minipage}
    \caption{Logical error rates for the $[[288,12,18]]$ BB code under (a) one and (b) four noisy measurement rounds, each followed by an ideal round. Comparison of guided space-time anisotropic MS with $N_D=9$, MS followed by OSD-$10$, and Relay-BP-$1$. Each guided MS realization uses at most $100$ iterations; Relay-BP-$1$ has a total iteration limit of $1000$ across sequential decoding legs.}
    \label{fig:guidedVrelayVosd10}
\end{figure*}
\begin{figure*}[htbp]
    \centering
    \begin{minipage}[t]{0.49\textwidth}
        \centering
        \includegraphics[width=\linewidth]{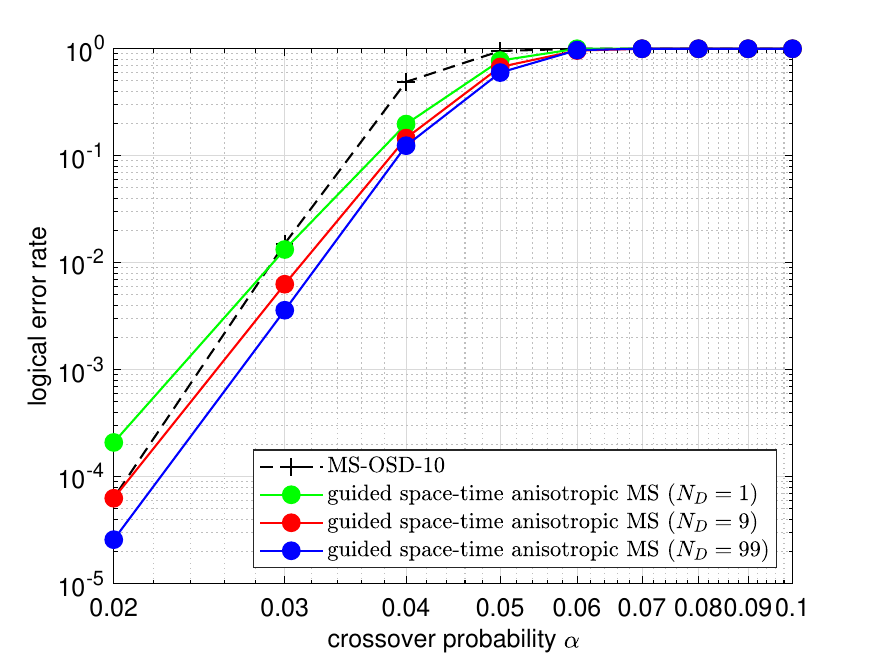}

        \small (a) Guided MS: increasing $N_D$
    \end{minipage}
    \hfill
    \begin{minipage}[t]{0.49\textwidth}
        \centering
        \includegraphics[width=\linewidth]{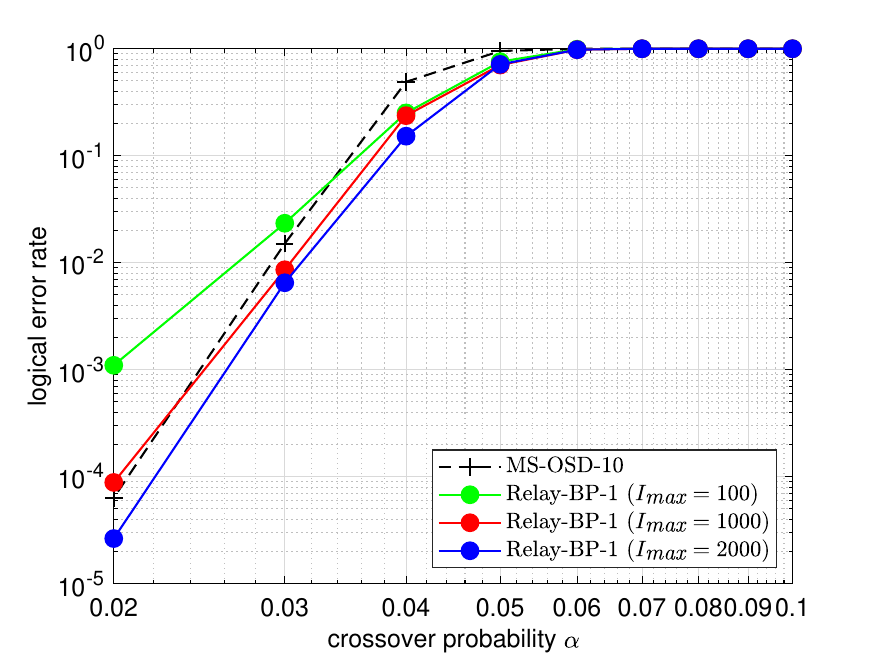}

        \small (b) Relay-BP-$1$: increasing iteration budget
    \end{minipage}
    \caption{Effect of decoding budget on the logical error rate for the $[[288,12,18]]$ BB code under four
noisy measurement rounds followed by an ideal round. (a) Guided space-time anisotropic MS with $N_D=1,9,99$, each realization limited
to $100$ iterations, and MS followed by OSD-$10$ as a reference. (b) Relay-BP-$1$ with maximum total iteration budgets of $100$,
$1000$, and $2000$, and MS followed by OSD-$10$ as a reference.}
    \label{fig:decoderBudget}
\end{figure*}

Having investigated the effect of low-weight degenerate error patterns, we now compare guided space-time anisotropic MS with MS followed by OSD post-processing and Relay-BP~\cite{müller2025improvedbeliefpropagationsufficient}. All comparisons in this subsection use ordinary independent noise sampling.

Fig.~\ref{fig:thresholdProposed} shows the logical error rates for the family of BB codes under one noisy measurement round followed by an ideal round. The guided ensemble uses $N_D=9$ additional decoder realizations, each with a maximum of $100$ iterations. At low crossover probabilities, guided space-time anisotropic MS achieves lower logical error rates than MS followed by OSD-$0$. In particular, for the $[[288,12,18]]$ code, the proposed decoder offers an improvement of one order of magnitude compared to OSD-$0$ post processing at crossover probability $\alpha=0.02$. The behavior across the considered code family is consistent with threshold-like scaling.

We next consider the $[[288,12,18]]$ BB code under one and four noisy measurement rounds, each followed by an ideal round. In addition to MS followed by OSD-$10$, for which we use the combination-sweep strategy
of~\cite{roffe_decoding_2020}, we compare with Relay-BP-$1$, the first-solution variant of Relay-BP. Relay-BP combines node-dependent memory strengths with sequentially connected MS decoding legs. We use the two-gross circuit-noise memory parameters of~\cite{müller2025improvedbeliefpropagationsufficient}: the initial leg uses a uniform memory strength $\gamma=0.125$ and at most $80$ iterations, while subsequent legs use at most $60$ iterations, with memory strengths independently sampled from $\operatorname{Unif}(-0.161,0.815)$ and held fixed throughout each leg. Each leg receives the final marginal beliefs of the preceding leg.

Fig.~\ref{fig:guidedVrelayVosd10} compares the three approaches. The guided ensemble uses $N_D=9$ and at most $100$ iterations per realization, while Relay-BP-$1$ is assigned a total iteration budget of at most $1000$ across its sequential legs. The proposed decoder achieves the lowest logical error rates at low crossover probabilities in both measurement-round settings, although Relay-BP-$1$ attains comparable performance with the larger sequential budget. Unlike the proposed guided ensemble, Relay-BP introduces anisotropy through randomized memory strengths without explicitly using the temporal-degeneracy structures analyzed here. 

 Fig.~\ref{fig:decoderBudget} considers the $[[288,12,18]]$ BB code under four noisy measurement rounds followed by an ideal round. In Fig.~\ref{fig:decoderBudget}(a), we vary the number of additional guided realizations, $N_D\in\{1,9,99\}$, while keeping the per-realization limit at $100$ iterations, and include MS followed by OSD-$10$ as a reference. The results show that additional guided realizations can reduce the logical error rate in the low-noise regime: at $\alpha=0.02$, the rate decreases from $2.09\times10^{-4}$ for $N_D=1$ to $2.58\times10^{-5}$ for $N_D=99$. In Fig.~\ref{fig:decoderBudget}(b), we vary the maximum total iteration budget of Relay-BP-$1$ over $100$, $1000$, and $2000$ iterations, retaining the same per-leg update rules and memory distributions. MS followed by OSD-$10$ is included as a reference.

For fixed $N_D$ and iteration limits, the message-passing complexity of the guided ensemble scales linearly with the number of edges in the space-time Tanner graph. In contrast, OSD involves Gaussian elimination and, at higher orders, the examination of additional candidate corrections. The guided realizations can also be executed in parallel, whereas Relay-BP's successive legs depend on the output of preceding legs and therefore have to be executed in sequence. Thus, similar aggregate iteration budgets do not imply identical sequential decoding depth or computational cost.

These comparisons motivate investigating whether the structural characterization of temporal degeneracy can inform more general anisotropic iterative decoders. In particular, the locations or distributions of the memory strengths in Relay-BP could be chosen using the space-time code structure or the observed syndrome, rather than being sampled independently of those features. In that way, anisotropy could be assigned in a less agnostic way and improve latency. More broadly, the systematic design of node- and edge-dependent update rules, building on~\cite{chytas2026edgebasedanisotropicdecodinggeneralized}, represents a possible direction for iterative decoding under constrained latency.

\section{conclusions}
In this work, we characterize temporal and spatial-temporal degeneracy introduced by repeated noisy syndrome measurements under the phenomenological noise model. The temporal kernel structures are imposed by the measurement model and arise independently of the specific quantum code construction. We identify the temporal codeword generators and derive decoder-independent upper bounds on the number of locally minimum half-weight error patterns supported on minimum even-weight temporal codewords. For $(d_v,d_c)$-regular parity-check matrices with $d_v>2$ and girth greater than four, we identify families of constant-weight error patterns whose multiplicity scales as $\Theta(nr)$ for fixed code degrees. Under symmetric initialization and the isolation assumption, these patterns create a local decoding failure for deterministic isotropic message-passing decoders under a flooding schedule. In particular, degenerate error patterns supported on temporal codewords are harmful for iterative decoding when they are related by graph automorphisms. Consequently, increasing the underlying code distance is not by itself a remedy for these local failures. Motivated by this analysis, we propose space-time anisotropic min-sum decoding with syndrome-guided realizations to mitigate both spatial and temporal symmetries. For fixed iteration and ensemble budgets, decoding complexity is linear in the number of Tanner graph edges. Numerical results on BB codes demonstrate practical improvements and show that the proposed approach can outperform MS followed by OSD post-processing. These results motivate addressing the recurring symmetries introduced by repeated measurements directly within low-complexity message passing. These results also suggest the  study of the joint design of codes, measurement protocols, and decoders. Single-shot error correction~\cite{bombin}, in which each correction cycle uses one round of noisy syndrome measurements, offers an alternative to joint decoding of repeated syndrome measurements. Code properties such as confinement~\cite{Quintavalle_2021} can support this approach by controlling residual data errors caused by measurement errors. However, confinement does not by itself remove the temporal codewords of the repeated-measurement construction considered here. How syndrome redundancy, confinement, and alternative measurement protocols can complement decoder-level symmetry breaking remains an open question. Future work will focus on extending the proposed techniques to the circuit-level noise model.

\appendices
\section{Proofs of the Results in Section~\ref{sec:sectemporal}}
\label{app:sec3proofs}

\subsection{Proof of Lemma~\ref{lemma1}}
\label{app:lemma1}
Each subset of $k$ ones in the first $r-1$ positions can be uniquely partitioned into $b = k-t$ blocks of consecutive ones.  

1.  There are $\binom{k-1}{t}$ ways to select which adjacent ones form consecutive pairs (integer composition).  

2.  The $b=k-t$ blocks should be separated by at least one zero. One zero is placed in each of the $b-1$ gaps, so there are $(r-1-k)-(b-1)=r-k-b$ zeros left. These zeros can be distributed into the $b+1$ gaps, which by stars and bars gives  $\binom{r-k}{k-t}$.

3.  Each block triggers exactly two data bits in the right block, giving weight $2(k-t)$ for the right block.  

4.  The left block  has weight $k$.  

Multiplying the combinatorial counts by $X^k Y^{2(k-t)}$ and summing over all possible $k$ and $t$ enumerates all vectors (bivariate enumeration). 
\subsection{Proof of Lemma~\ref{lemma:minimal_xky2}}
\label{app:lemma:minimal_xky2}
\begin{IEEEproof}
       We base the proof on Lemma 2.1 of~\cite{705584} and~\cite{Vontobel_2005}.
    Let $\mathcal{U} \subset [2r-1]$ be the support of a $\mathbf{c}^{k,2}$ vector of $\mathbf{B}$. We need to show that $\operatorname{rank}(\mathbf{B}(\mathcal{U}))$ equals $k+1$, where $\mathbf{B}(\mathcal{U})$ is the restriction of $\mathbf{B}$ to columns indexed by $\mathcal{U}$.  Since the support consists of $k$ consecutive $1$'s in the measurement part and two $1$'s in the data part,  we have 
     \begin{equation}
     \mathbf{B}(\mathcal{U})=
\left[
\begin{array}{c|c}
\begin{array}{ccccc}
1 \\
1 & 1 \\
& 1 & 1 & & \\
& & \vdots & \ddots & 1 \\
& & & & 1
\end{array}
&
\begin{array}{cc}
1 & 0\\
0 & 0 \\
0 & 0 \\
\vdots & \vdots \\
0 & 1 \\
\end{array}
\end{array}
\right],
\end{equation}
or $\mathbf{B}(\mathcal{U})=[\mathbf{A}_1|\mathbf{A}_2]$. $\mathbf{B}(\mathcal{U})$ has  $k+1$ nonzero rows and $\mathbf{A}_1$ has $k$ columns, while $\mathbf{A}_2$ has $2$ columns. 
We now show that the $k+1$ rows of $\mathbf{B}(\mathcal{U})$ are linearly independent.  
Assume that a binary linear combination of the rows equals zero:
\[
\sum_{i=1}^{k+1}\lambda_i \mathbf{\rho}_i=\mathbf{0},
\qquad \lambda_i\in\mathbb{F}_2,
\]
where $\mathbf{\rho}_i$ denotes the $i$-th row of $\mathbf{B}(\mathcal{U})$.

Looking at the first column of $\mathbf{A}_2$, only the first row has a nonzero entry; hence, $\lambda_1=0$.  
Then, looking at the first column of $\mathbf{A}_1$, only the first two rows have nonzero entries there, and since $\lambda_1=0$, it follows that $\lambda_2=0$.  
Proceeding inductively, for each $j=2,\dots,k$, the $j$-th column of $\mathbf{A}_1$ has nonzero entries only in rows $j$ and $j+1$. Since $\lambda_j=0$, we obtain $\lambda_{j+1}=0$. Thus,
\[
\lambda_1=\lambda_2=\cdots=\lambda_{k+1}=0.
\]
Therefore the rows of $\mathbf{B}(\mathcal{U})$ are linearly independent, and
\[
\operatorname{rank}(\mathbf{B}(\mathcal{U}))=k+1.
\]

By Lemma 2.1 in~\cite{705584}, this implies that $\mathbf{c}^{k,2}$ is a minimal vector of $\mathbf{B}$.
\end{IEEEproof}
\subsection{Proof of Lemma~\ref{lemma:decompose_xky2}}
\label{app:lemma:decompose_xky2}
\begin{IEEEproof}
We have shown that $\mathbf{c}_i^{1,2}\in\ker(\mathbf{B})$ for every
$i\in\{1,\ldots,r-1\}$. To show linear independence, suppose that
\[
\sum_{i=1}^{r-1} a_i\mathbf{c}_i^{1,2}
=
\mathbf{0},
\qquad
a_i\in\mathbb{F}_2.
\]
The first $r-1$ coordinates of $\mathbf{c}_i^{1,2}$ are zero except for a one in position $i$. Restricting the above equality to these coordinates therefore gives $a_i=0$ for every $i$.
Hence, $\{\mathbf{c}_i^{1,2}\}_{i=1}^{r-1}$ is linearly independent.

Moreover, since $\dim(\ker(\mathbf{B}))=r-1$, the set
$\{\mathbf{c}_i^{1,2}\}_{i=1}^{r-1}$ is a linearly independent
subset of $\ker(\mathbf{B})$ whose cardinality equals
$\dim\ker(\mathbf{B})$, and hence it forms a basis.
\end{IEEEproof}
\subsection{Proof of Lemma~\ref{lemma:vectors}}
\label{app:lemma:vectors}
\begin{IEEEproof}
    Let
\[
\mathbf{c}=
\begin{bmatrix}
\mathbf{z}\\
\mathbf{x}
\end{bmatrix}
\in \ker(\mathbf{B}),
\]
so that
\[
\mathbf{B}_L\mathbf{z}+\mathbf{B}_R\mathbf{x}=\mathbf{0}
\]
over $\mathbb{F}_2$.

By definition,
\[
\mathbf{H}_1^{(r)}
=
\big[\mathbf{B}_L\otimes \mathbf{I}_{d_v}\;\big|\;\mathbf{B}_R\otimes \mathbf{1}_{d_v}\big],
\qquad
\widehat{\mathbf{c}}
=
\begin{bmatrix}
\mathbf{z}\otimes \mathbf{1}_{d_v}\\[0.3em]
\mathbf{x}
\end{bmatrix}.
\]
Therefore,
\[
\mathbf{H}_1^{(r)}\widehat{\mathbf{c}}
=
(\mathbf{B}_L\otimes \mathbf{I}_{d_v})(\mathbf{z}\otimes \mathbf{1}_{d_v})
+
(\mathbf{B}_R\otimes \mathbf{1}_{d_v})\mathbf{x}.
\]
Using the mixed-product property of the Kronecker product,
\[
(\mathbf{B}_L\otimes \mathbf{I}_{d_v})(\mathbf{z}\otimes \mathbf{1}_{d_v})
=
(\mathbf{B}_L\mathbf{z})\otimes \mathbf{1}_{d_v},
\]
and
\[
(\mathbf{B}_R\otimes \mathbf{1}_{d_v})\mathbf{x}
=
(\mathbf{B}_R\mathbf{x})\otimes \mathbf{1}_{d_v}.
\]
Hence,
\[
\mathbf{H}_1^{(r)}\widehat{\mathbf{c}}
=
\big(\mathbf{B}_L\mathbf{z}+\mathbf{B}_R\mathbf{x}\big)\otimes \mathbf{1}_{d_v}
=
\mathbf{0}\otimes \mathbf{1}_{d_v}
=
\mathbf{0}.
\]
Thus $\widehat{\mathbf{c}}\in \ker(\mathbf{H}_1^{(r)})$.
\end{IEEEproof}

\subsection{Proof of Corollary~\ref{cor:bijectionH1}}
\label{app:cor:bijectionH1}
\begin{IEEEproof}
Lemma~\ref{lemma:vectors} shows that every vector in $\ker(\mathbf{B})$ lifts to a vector in $\ker(\mathbf{H}_1^{(r)})$.
For the converse, let
\[
\widehat{\mathbf c}=
\bigl(
\mathbf{z}^{(1)},\dots,\mathbf{z}^{(r-1)}
\;\big|\;
\mathbf{x}
\bigr)
\in \ker(\mathbf{H}_1^{(r)}),
\]
where
\[
\mathbf{z}^{(t)} \coloneqq (z_{1,t},\dots,z_{d_v,t}),
\qquad
t=1,\dots,r-1,
\]
and
\[
\mathbf{x}\coloneqq (x_1,\dots,x_r).
\]

From the structure of $\mathbf{H}_1^{(r)}$, for each $a\in[d_v]$ the parity-check equations are
\[
z_{a,1}+x_1=0,
\]
\[
z_{a,t-1}+z_{a,t}+x_t=0,\qquad 2\le t\le r-1,
\]
\[
z_{a,r-1}+x_r=0.
\]
Therefore, for every $a\in[d_v]$,
\[
z_{a,1}=x_1,
z_{a,2}=x_1+x_2,
\dots,
z_{a,t}=x_1+\cdots+x_t.
\]
Hence, the measurement coordinates are independent of $a$.
Define
\[
z_t:=\sum_{j=1}^{t}x_j,\qquad
\mathbf{z}:=(z_1,\ldots,z_{r-1}).
\]
Then
\[
\widehat{\mathbf{c}}
=
\begin{bmatrix}
\mathbf{z}\otimes\mathbf{1}_{d_v}\\
\mathbf{x}
\end{bmatrix}.
\]
The final parity-check equation gives
\[
\sum_{j=1}^{r}x_j=0.
\]
Together with the preceding equations, this shows that
\[
\mathbf{c}
=
\begin{bmatrix}\mathbf{z}\\\mathbf{x}\end{bmatrix}
\in\ker(\mathbf{B}).
\]
Thus $\widehat{\mathbf{c}}$ is the lift of a vector in
$\ker(\mathbf{B})$. This vector is unique because the lift retains
$\mathbf{x}$ and repeats each coordinate of $\mathbf{z}$.
Therefore the lifting map is bijective.
\end{IEEEproof}
\subsection{Proof of Corollary~\ref{cor:evenWeight}}
\label{app:cor:evenWeight}
\begin{IEEEproof}
    The weight enumerator of vectors in $\ker(\mathbf{H}_1^{(r)})$ shows that the weight of the data part of any vector is even. Therefore, for any
\[
\widehat{\mathbf c}=
\bigl(
\mathbf{z}^{(1)},\dots,\mathbf{z}^{(r-1)}
\;\big|\;
\mathbf{x}
\bigr)
\in \ker(\mathbf{H}_1^{(r)}),
\]
where $\mathbf{z}^{(1)},\dots,\mathbf{z}^{(r-1)}\in\mathbb{F}_2^{d_v}$ denote the measurement blocks and
$\mathbf{x}=(x_1,\dots,x_{r})\in\mathbb{F}_2^r$ denotes the data coordinates, it holds that
\[
\sum_{i=1}^{r} x_i = 0 \pmod 2.
\]
Hence every vector in $\ker(\mathbf{H}_1^{(r)})$ is a temporal codeword.
\end{IEEEproof}
\subsection{Detailed Proof of Theorem~\ref{theorem1}}
\label{app:theorem1}

From Corollary~\ref{cor:dv2GeneratesALL}, a weight-$(kd_v+2)$ minimal
vector can be written as
\begin{equation}
\mathbf{c}^{kd_v,2}
=
\mathbf{c}_1\oplus \mathbf{c}_2\oplus \cdots \oplus \mathbf{c}_k,
\label{eq:decompositionApp}
\end{equation}
where each $\mathbf{c}_i$ is a weight-$(d_v+2)$ temporal codeword of
$\mathbf{H}_1^{(r)}$.

Let $\mathcal{V}_i:=\operatorname{supp}(\mathbf{c}_i)$, $i\in[k]$.
Since modulo-$2$ addition of binary vectors corresponds to symmetric
difference of supports, it follows from~\eqref{eq:decompositionApp} that
\[
\operatorname{supp}(\mathbf{c}^{kd_v,2})
=
\mathcal{V}_1\triangle\mathcal{V}_2\triangle\cdots\triangle\mathcal{V}_k.
\]
This can be written as
\begin{equation}
\operatorname{supp}(\mathbf{c}^{kd_v,2})
=
\mathcal{V}
=
\bigsqcup_{i=1}^{k}\mathcal{U}_i,
\label{eq:VdisjointApp}
\end{equation}
where $\bigsqcup$ denotes disjoint union, and
\[
\mathcal{U}_1
=
\mathcal{V}_1\setminus \mathcal{V}_2,
\]
\[
\mathcal{U}_i
=
\mathcal{V}_i\setminus
(\mathcal{V}_{i-1}\cup\mathcal{V}_{i+1}),
\qquad 2\le i\le k-1,
\]
and
\[
\mathcal{U}_k
=
\mathcal{V}_k\setminus \mathcal{V}_{k-1}.
\]
Hence,
\[
|\mathcal{U}_1|=|\mathcal{U}_k|=d_v+1,
\qquad
|\mathcal{U}_i|=d_v,\quad 2\le i\le k-1.
\]

Let us assume an error pattern
$\mathbf{e}$ of weight
$\frac{w(\mathbf{c}^{kd_v,2})}{2}$ with
$\operatorname{supp}(\mathbf{e})
\subset\operatorname{supp}(\mathbf{c}^{kd_v,2})$.
Then, for a temporal degenerate error pattern $\mathbf{f}$ of the same
weight with disjoint support, we have
\begin{equation}
\mathbf{e}\oplus\mathbf{f}
=
\mathbf{c}^{kd_v,2},
\qquad
w(\mathbf{e})+w(\mathbf{f})
=
w(\mathbf{c}^{kd_v,2}).
\end{equation}

Decomposing $\mathbf{c}^{kd_v,2}$ in terms of its constituent
weight-$(d_v+2)$ temporal codewords, we can write
\begin{equation}
\mathbf{e}\oplus\mathbf{f}
=
\sum_{i\in A}\mathbf{c}_i
\;\oplus\;
\sum_{j\in B}\mathbf{c}_j
\pmod{2},
\label{eq:generalDecompCompact}
\end{equation}
where $A,B\subset[k]$ are disjoint and $A\cup B=[k]$, so that
$\{A,B\}$ is a partition of $[k]$. Equivalently,
\[
\mathbf{e}\oplus\mathbf{f}_A=\mathbf{c}_B,
\qquad
\mathbf{e}\oplus\mathbf{f}_B=\mathbf{c}_A,
\]
with
\[
\mathbf{f}_A
:=
\mathbf{f}\oplus\sum_{i\in A}\mathbf{c}_i
\pmod{2},
\]
and
\[
\mathbf{f}_B
:=
\mathbf{f}\oplus\sum_{j\in B}\mathbf{c}_j
\pmod{2}.
\]

\textbf{Remark:}
There are exactly $2^{k-1}-1$ distinct ways to split a minimal vector
$\mathbf{c}^{kd_v,2}$ into sums
$\mathbf{c}_A\oplus\mathbf{c}_B=\mathbf{c}^{kd_v,2}$, since each
nonempty proper subset $A\subset[k]$ has a unique complementary subset
$B=[k]\setminus A$, and exchanging $A$ and $B$ gives the same
partition. The supports of $\mathbf{c}_A$ and $\mathbf{c}_B$ are denoted
by $\mathcal{V}_A$ and $\mathcal{V}_B$, respectively.

For $\mathbf{e}$ to be locally minimum half-weight, it should hold that
\begin{equation}
w(\mathbf{f}_A)\geq w(\mathbf{e}),
\label{eq:fa}
\end{equation}
and
\begin{equation}
w(\mathbf{f}_B)\geq w(\mathbf{e}),
\label{eq:fb}
\end{equation}
for all partitions $\{A,B\}$.

Equation~(\ref{eq:fa}) implies
\begin{equation}
\begin{split}
w(\mathbf{e}\oplus\mathbf{c}_B)\geq w(\mathbf{e})
&\implies
w(\mathbf{e})+w(\mathbf{c}_B)-2|\mathcal{T}_B|
\geq w(\mathbf{e})
\\
&\implies
|\mathcal{T}_B|
\leq
\frac{w(\mathbf{c}_B)}{2},
\end{split}
\end{equation}
where $\mathcal{T}_B$ denotes the support of $\mathbf{e}$ on
$\mathbf{c}_B$, so that
\[
\mathcal{T}_B
\subset
\mathcal{V}\cap\mathcal{V}_B
=
\mathcal{V}_B\setminus\mathcal{V}_A.
\]
Similarly, Eq.~(\ref{eq:fb}) gives
\[
|\mathcal{T}_A|
\leq
\frac{w(\mathbf{c}_A)}{2}.
\]

Since $\mathcal{T}_A$ and $\mathcal{T}_B$ are disjoint,
\begin{equation}
|\mathcal{T}_A|+|\mathcal{T}_B|
=
w(\mathbf{e})
=
\frac{w(\mathbf{c}^{kd_v,2})}{2}.
\label{eq:TATB}
\end{equation}
Combining Eq.~(\ref{eq:TATB}) with the previous bounds gives
\begin{equation}
\frac{w(\mathbf{c}^{kd_v,2})-w(\mathbf{c}_B)}{2}
\leq
|\mathcal{T}_A|
\leq
\frac{w(\mathbf{c}_A)}{2},
\label{eq:TAsimple}
\end{equation}
and
\begin{equation}
\frac{w(\mathbf{c}^{kd_v,2})-w(\mathbf{c}_A)}{2}
\leq
|\mathcal{T}_B|
\leq
\frac{w(\mathbf{c}_B)}{2}.
\label{eq:TBsimple}
\end{equation}

Moreover,
\[
w(\mathbf{c}^{kd_v,2})
=
w(\mathbf{c}_A)+w(\mathbf{c}_B)
-2w(\mathbf{c}_{AB}),
\]
where $w(\mathbf{c}_{AB})$ denotes the weight of the intersection of the
two vectors, equivalently
\[
w(\mathbf{c}_{AB})
=
|\mathcal{V}_A\cap\mathcal{V}_B|.
\]

Let us now denote
\[
w(\mathbf{c}_A)=|A|d_v+l_A,
\qquad
w(\mathbf{c}_B)=|B|d_v+l_B,
\]
where, because of Corollary~\ref{cor:evenWeight}, $l_A$ and $l_B$ are
even. By setting $w(\mathbf{c}_{AB})=l_{AB}$, since
$\mathbf{c}^{kd_v,2}$ is a minimal vector,
\[
l_A+l_B-2l_{AB}=2.
\]
Hence,
\begin{equation}
\frac{kd_v+2-|B|d_v-l_B}{2}
\leq
|\mathcal{T}_A|
\leq
\frac{|A|d_v+l_A}{2},
\label{eq:TA}
\end{equation}
and
\begin{equation}
\frac{kd_v+2-|A|d_v-l_A}{2}
\leq
|\mathcal{T}_B|
\leq
\frac{|B|d_v+l_B}{2}.
\label{eq:TB}
\end{equation}

Recall that the bounds in Eqs.~(\ref{eq:TAsimple}) and~(\ref{eq:TBsimple})
should hold for all partitions $(A,B)$. In particular, they must hold for
the singleton partitions
\[
A_i=\{i\},
\qquad
B_i=[k]\setminus\{i\},
\qquad
i\in[k].
\]
These singleton constraints are only necessary, not sufficient, for
locally minimum half-weight error patterns.

For the singleton partition $A_i=\{i\}$, we have
$\mathbf{c}_{A_i}=\mathbf{c}_i$ and
\[
\mathcal{V}_{A_i}\cap\mathcal{V}
=
\mathcal{U}_i.
\]
Hence,
\[
|\mathcal{T}_{A_i}|
=
|\operatorname{supp}(\mathbf{e})\cap\mathcal{U}_i|.
\]
Define
\[
w_i
:=
|\operatorname{supp}(\mathbf{e})\cap\mathcal{U}_i|,
\qquad i\in[k].
\]
Applying Eq.~(\ref{eq:TA}) to the singleton partition gives
\[
\frac{
w(\mathbf{c}^{kd_v,2})-w(\mathbf{c}_{B_i})
}{2}
\leq
w_i
\leq
\frac{w(\mathbf{c}_i)}{2}.
\]
Since $w(\mathbf{c}_i)=d_v+2$, the upper bound is
\[
w_i
\leq
\frac{d_v+2}{2}.
\]

It remains to evaluate the lower bound. If $i=1$ or $i=k$, then
$\mathbf{c}_{B_i}$ is a contiguous sum of $k-1$ temporal codewords,
and therefore
\[
w(\mathbf{c}_{B_i})
=
(k-1)d_v+2.
\]
Thus,
\[
w_i
\geq
\frac{kd_v+2-((k-1)d_v+2)}{2}
=
\frac{d_v}{2},
\qquad
i\in\{1,k\}.
\]

If $2\leq i\leq k-1$, then removing $\mathbf{c}_i$ splits the complement
into two contiguous components. Hence,
\[
w(\mathbf{c}_{B_i})
=
(k-1)d_v+4,
\]
and therefore
\[
w_i
\geq
\frac{kd_v+2-((k-1)d_v+4)}{2}
=
\frac{d_v-2}{2},
\qquad
2\leq i\leq k-1.
\]

Consequently, every locally minimum half-weight error pattern satisfies the local constraints
\[
\frac{d_v}{2}
\leq
w_1,w_k
\leq
\frac{d_v+2}{2},
\]
and
\[
\frac{d_v-2}{2}
\leq
w_i
\leq
\frac{d_v+2}{2},
\qquad
2\leq i\leq k-1.
\]
Moreover, since $\mathbf{e}$ has half the weight of
$\mathbf{c}^{kd_v,2}$,
\[
\sum_{i=1}^{k}w_i
=
\frac{kd_v+2}{2}.
\]

For compactness, define the set of admissible weight distributions
\begin{equation}
\mathcal{A}_k :=
\left\{
\mathbf{w}\in\mathbb{Z}_{\geq0}^{k}:
\begin{array}{l}
\sum_{i=1}^{k}w_i=(kd_v+2)/2,\\
d_v/2\leq w_1,w_k\leq(d_v+2)/2,\\
(d_v-2)/2\leq w_i\leq(d_v+2)/2,\\
\hfill 2\leq i\leq k-1
\end{array}
\right\}.
\label{eq:admissibleSet}
\end{equation}

We now count all half-weight error patterns satisfying these necessary
singleton constraints. Since the sets $\mathcal{U}_i$ are disjoint, a
weight distribution
$\mathbf{w}=(w_1,\ldots,w_k)\in\mathcal{A}_k$
contributes
\[
\binom{d_v+1}{w_1}
\left(
\prod_{i=2}^{k-1}\binom{d_v}{w_i}
\right)
\binom{d_v+1}{w_k}
\]
choices, because
$|\mathcal{U}_1|=|\mathcal{U}_k|=d_v+1$
and
$|\mathcal{U}_i|=d_v$
for $2\leq i\leq k-1$.
Summing over all admissible weight distributions gives
\begin{equation}
N(k)
\leq
\sum_{\mathbf{x}\in\mathcal{A}_k}
\binom{d_v+1}{w_1}
\left(
\prod_{i=2}^{k-1}\binom{d_v}{w_i}
\right)
\binom{d_v+1}{w_k}.
\label{eq:upperNkproof}
\end{equation}

Finally, the asymptotic scaling follows from the crude bound obtained by
ignoring all local constraints:
\[
N(k)
\leq
\binom{kd_v+2}{(kd_v+2)/2}
\leq
2^{kd_v+2}
=
2^{O(kd_v)}.
\]
For fixed $k$, the contiguous weight-$(kd_v+2)$ vector spans $k$
consecutive rounds, so over $r$ total rounds there are $r-k$ possible
contiguous placements. Hence,
\begin{equation}
N_{\mathrm{total}}(k,r)
\leq
(r-k)N(k),
\end{equation}
and therefore
\begin{equation}
N_{\mathrm{total}}(k,r)
\leq
(r-k)\,2^{O(kd_v)}.
\end{equation}
Thus, for fixed $k$ and $d_v$, the number of such patterns grows at most
linearly with the number of measurement rounds $r$.

\subsection{Proof of Lemma~\ref{lemma:temporalgenerator}}
\label{app:lemma:temporalgenerator}
\begin{IEEEproof}
Write
\[
\mathbf D
=
(\mathbf z^{(1)},\ldots,\mathbf z^{(r-1)}
\mid
\mathbf x^{(1)},\ldots,\mathbf x^{(r)})
\in \ker(\mathbf H'),
\]
where $\mathbf z^{(t)}\in\mathbb F_2^m$ denotes the measurement part at
time $t \in [r-1]$ and $\mathbf x^{(t)}\in\mathbb F_2^n$ denotes the data part at
time $t \in [r]$.

Since $\mathbf D$ is temporal, its total data parity is zero:
\[
\sum_{t=1}^{r}\mathbf x^{(t)}=\mathbf 0 \pmod 2.
\]
Equivalently, for every variable node $v\in[n]$,
\[
\sum_{t=1}^{r} x_v^{(t)}=0 \pmod 2.
\]

Fix $v\in[n]$ and define
\[
\chi_t=x_v^{(t)},\qquad t\in[r].
\]
The above parity condition gives $\chi_1+\cdots+\chi_r=0$. Hence the sequence
$(\chi_1,\ldots,\chi_r)$ determines a unique vector
\[
\mathbf V^{(v)}\in \ker(\mathbf H_1^{(r)}(v)),
\]
whose data part is $(\chi_1,\ldots,\chi_r)$ and whose measurement blocks are
\[
\mathbf z_v^{(t)}
=
(\chi_1+\cdots+\chi_t)\mathbf 1_{d_v},
\qquad 1\le t\le r-1.
\]
The final parity-check equation is satisfied because
$\chi_1+\cdots+\chi_r=0$.

Let
\[
\iota_v:\ker(\mathbf H_1^{(r)}(v))\to \ker(\mathbf H')
\]
denote the natural embedding that places a local vector on the data
coordinates corresponding to the repeated copies of variable node $v$
and on the measurement-variable coordinates associated with the neighboring
checks of $v$ in the full space-time Tanner graph. Then
\[
\mathbf D
=
\sum_{v=1}^{n}\iota_v(\mathbf V^{(v)}).
\]
Indeed, the data parts agree coordinatewise by construction. Moreover,
the measurement part of a kernel vector is uniquely determined by its data
part through the detector recursion:
\[
\mathbf z^{(t)}
=
\mathbf H(\mathbf x^{(1)}+\cdots+\mathbf x^{(t)}),
\qquad 1\le t\le r-1.
\]
The embedded local syndrome contributions therefore sum to the syndrome
blocks of $\mathbf D$.

Thus every temporal codeword lies in the span of the embedded
single-variable kernels $\iota_v(\ker(\mathbf H_1^{(r)}(v)))$. By
Corollary~\ref{cor:dv2GeneratesALL}, each vector in
$\ker(\mathbf H_1^{(r)}(v))$ is a linear combination of
weight-$(d_v+2)$ temporal codeword generators. Hence every temporal
codeword in $\ker(\mathbf H')$ is generated by linear combinations of
embedded weight-$(d_v+2)$ temporal codeword generators.
\end{IEEEproof}

\subsection{Proof of Lemma~\ref{lem:min-even-temporal}}
\label{app:lem:min-even-temporal}
\begin{IEEEproof}
Write
\[
\mathbf D=
\left(
\mathbf z^{(1)},\ldots,\mathbf z^{(r-1)}
\,\middle|\,
\mathbf x^{(1)},\ldots,\mathbf x^{(r)}
\right)
\in\ker(\mathbf H').
\]
Since $\mathbf D$ is a temporal codeword,
\[
    \bigoplus_{j=1}^{r}\mathbf x^{(j)}=\mathbf 0.
\]
For $t\in[r-1]$, define the cumulative data parity
\[
    \mathbf g_t:=
    \bigoplus_{j=1}^{t}\mathbf x^{(j)},
\]
and set $\mathbf g_0=\mathbf g_r=\mathbf 0$.
From the detector equations of $\mathbf H'$, it follows that
\[
    \mathbf z^{(t)}=\mathbf H\mathbf g_t,
    \qquad 1\le t\le r-1,
\]
and
\[
    \mathbf x^{(t)}
    =
    \mathbf g_{t-1}\oplus\mathbf g_t,
    \qquad 1\le t\le r.
\]
Therefore,
\begin{equation}
\label{eq:temporal-weight-decomp}
w(\mathbf D)
=
\sum_{t=1}^{r-1}w(\mathbf H\mathbf g_t)
+
\sum_{t=1}^{r}
w(\mathbf g_{t-1}\oplus\mathbf g_t).
\end{equation}

Let
\[
    q:=\max_{1\le t\le r-1}w(\mathbf g_t),
\]
and choose $t^\star$ such that
$w(\mathbf g_{t^\star})=q$. Since $\mathbf{D}\neq \mathbf{0}$, we have $q\geq 1$.
For every variable node
$v\in \operatorname{supp}(\mathbf g_{t^\star})$, we have
$g_0(v)=g_r(v)=0$ and $g_{t^\star}(v)=1$.
Hence the cumulative parity of the $v$-th data coordinate must
have changed at least once before or at $t^\star$, and must change
at least once more after $t^\star$. Thus, it undergoes at least two
parity transitions in total.
Each such transition corresponds to a nonzero entry in
$\mathbf g_{t-1}\oplus\mathbf g_t$. Since the $q$ variable nodes in
$\operatorname{supp}(\mathbf g_{t^\star})$ are distinct, the data part therefore
satisfies
\begin{equation}
\label{eq:data-part-bound}
    \sum_{t=1}^{r}
    w(\mathbf g_{t-1}\oplus\mathbf g_t)
    \ge 2q.
\end{equation}

We next lower-bound the measurement contribution at time
$t^\star$. Let
\[
    S:=\operatorname{supp}(\mathbf g_{t^\star}),\qquad |S|=q,
\]
and, for every check node $c$, define
\[
    a_c:=|\mathcal N(c)\cap S|.
\]
Since every variable node in $S$ has degree $d_v$,
\[
    \sum_c a_c=qd_v.
\]
Moreover, the $c$-th entry of
$\mathbf H\mathbf g_{t^\star}$ is nonzero exactly when $a_c$ is odd,
and hence
\[
    w(\mathbf H\mathbf g_{t^\star})
    =
    \sum_c \mathbf{1}_{\{a_c\ {\rm odd}\}}.
\]
For every nonnegative integer $a$,
\[
    a
    \le
    2\binom{a}{2}
    +
    \mathbf{1}_{\{a\ {\rm odd}\}},
\]
so that
\[
    qd_v
    \le
    2\sum_c\binom{a_c}{2}
    +
    w(\mathbf H\mathbf g_{t^\star}).
\]
Since the Tanner graph of $\mathbf H$ has girth larger than four,
any two distinct variable nodes share at most one check node.
Therefore,
\[
    \sum_c\binom{a_c}{2}
    \le
    \binom{q}{2},
\]
which gives
\begin{equation}
\label{eq:measurement-part-bound}
    w(\mathbf H\mathbf g_{t^\star})
    \ge
    q(d_v-q+1).
\end{equation}

Suppose first that
\[
    2\le q\le d_v+1.
\]
Using \eqref{eq:temporal-weight-decomp},
\eqref{eq:data-part-bound}, and
\eqref{eq:measurement-part-bound},
\begin{align}
w(\mathbf D)
&\ge
2q+q(d_v-q+1) \nonumber\\
&=
q(d_v-q+3) \nonumber\\
&=
2d_v+2+(q-2)(d_v-q+1) \nonumber\\
&\ge
2d_v+2.
\end{align}
If $q\ge d_v+2$, the data contribution alone gives
\[
    w(\mathbf D)\ge2q\ge2d_v+2.
\]

It remains to consider $q=1$. Let
\[
    L:=
    \left|
    \left\{
    t\in[r-1]:
    \mathbf g_t\neq\mathbf 0
    \right\}
    \right|.
\]
Every nonzero $\mathbf g_t$ then has weight one, and hence
\[
    w(\mathbf H\mathbf g_t)=d_v.
\]
Thus, the measurement part has weight $Ld_v$. The data part has
even weight, since for each variable coordinate the binary sequence
$g_0(v),g_1(v),\ldots,g_r(v)$ begins and ends at zero and therefore
contains an even number of transitions. Since $d_v$ is odd and
$\mathbf D$ has even weight, $L$ must be even. As
$\mathbf D\neq\mathbf 0$, we have $L\ge2$. Moreover, the data part
contains at least two nonzero coordinates. Consequently,
\[
    w(\mathbf D)\ge 2d_v+2.
\]
\end{IEEEproof}

\subsection{Proof of Theorem~\ref{theorem:upper boundsTemporal}}
\label{app:theorem:upper boundsTemporal}
\begin{IEEEproof}
Suppose first that $d_v$ is even. A temporal codeword generator has even weight $d_v+2$. In this case, every
error pattern of weight $(d_v+2)/2$ supported on the generator is locally minimum. There are $n(r-1)$ temporal codeword generators in $\mathbf{H}'$; therefore, the number of incidences equals \begin{equation}
\overline{N}_{\mathrm{even}}=\binom{d_v+2}{(d_v+2)/2}n(r-1).
\label{eq:dvEven}
\end{equation}

Now suppose that $d_v$ is odd. In this case, the temporal codeword generator has odd weight and therefore does not admit a half-weight error pattern. Hence, we need to consider errors in the support of the lowest-weight even temporal codewords.  By Lemma~\ref{lem:min-even-temporal}, every nonzero even-weight temporal codeword has weight at least $2d_v+2$.
We determine which codewords with nonzero measurement part attain this
bound by examining the equality cases in the proof of
Lemma~\ref{lem:min-even-temporal}. In the notation of that proof,
equality for $q=1$ requires exactly two consecutive, identical
singleton vectors $\mathbf g_t$. This gives a modulo-$2$ sum of two
consecutive temporal codeword generators associated with the same data
variable, whose common data coordinate cancels. Equality for $q=2$
requires a single nonzero vector $\mathbf g_t$ supported on two data
variables that share exactly one check. This gives a modulo-$2$ sum of
two temporal codeword generators whose common measurement coordinate
cancels. The remaining possible equality case, $q=d_v+1$, has zero
measurement part and is therefore excluded here. All other cases give
weight strictly larger than $2d_v+2$.
Thus, the minimum-weight even temporal codewords considered here arise
from two temporal codeword generators that share and cancel either
\begin{enumerate}
    \item one data-variable coordinate, or
    \item one measurement-variable coordinate.
\end{enumerate}
Both combinations form a weight-$(2d_v+2)$ temporal codeword $\mathbf D=\mathbf c_1\oplus\mathbf c_2$, where the two generators
share exactly one coordinate, which cancels in $\mathbf D$. Define
\[
\mathcal U_i
:=
\operatorname{supp}(\mathbf c_i)
\cap\operatorname{supp}(\mathbf D),
\qquad i=1,2.
\]
The sets $\mathcal U_1$ and $\mathcal U_2$ are disjoint, each has size $d_v+1$, and together they form $\operatorname{supp}(\mathbf D)$.
Let $\mathbf e$ be an error pattern supported on $\mathbf D$ with $w(\mathbf e)=d_v+1$, and define
\[
w_i
:=
|\operatorname{supp}(\mathbf e)\cap\mathcal U_i|,
\qquad i=1,2.
\]
Then $w_1+w_2=d_v+1$, and
\[
w(\mathbf e\oplus\mathbf c_i)-w(\mathbf e)
=
w(\mathbf c_i)
-2|\operatorname{supp}(\mathbf e)
       \cap\operatorname{supp}(\mathbf c_i)|
=
d_v+2-2w_i.
\]
The span of the two generators is $\mathcal G_D=\{\mathbf 0,\mathbf c_1,\mathbf c_2,\mathbf D\}$. Since $\mathbf e$ occupies half of $\operatorname{supp}(\mathbf D)$, adding $\mathbf D$ preserves its weight. Hence $\mathbf e$ satisfies Definition~\ref{def:locallymin} if and only if
$d_v+2-2w_i\ge0$ for $i=1,2$.

Because $d_v$ is odd and $w_i$ is an integer, these inequalities give $w_i\le(d_v+1)/2$. Together with $w_1+w_2=d_v+1$, this forces
\[
w_1=w_2=\frac{d_v+1}{2}.
\]
Thus each such temporal codeword supports exactly
\[
\binom{d_v+1}{(d_v+1)/2}^{\!2}
\]
locally minimum half-weight error patterns.
In case 1), the resulting temporal codeword belongs to the kernel of $\mathbf{H}_1^{(r)}$ and therefore to the kernel of $\mathbf{H}'$. There are exactly $n(r-2)$ such temporal codewords in $\mathbf{H}'$. 
The number of incidences equals $\binom{d_v+1}{(d_v+1)/2}^2n(r-2)$.
In case 2), the two generators are associated with distinct data variables that share a check in $\mathbf H$. The number of unordered pairs of variables sharing a check is
\[
m\binom{d_c}{2}
=
\frac{1}{2}nd_v(d_c-1),
\]
where $md_c=nd_v$. Because the Tanner graph has girth larger than four, two data variables share at most one check, so no pair is counted through two different checks. Over the $r-1$ measurement intervals,
this gives
\[
\frac{1}{2}nd_v(d_c-1)(r-1)
\]
temporal codewords of this type, and therefore
\[
\frac{1}{2}
\binom{d_v+1}{(d_v+1)/2}^{\!2}
nd_v(d_c-1)(r-1)
\]
incidences.
In summary, for odd $d_v$, the total number of incidences between the
minimum-weight even temporal codewords considered here and their
associated locally minimum half-weight errors is
\begin{equation}
\begin{split}
\overline{N}_{\mathrm{odd}}=
    \binom{d_v+1}{(d_v+1)/2}^2n(r-2)+ \\ \frac{1}{2}\binom{d_v+1}{(d_v+1)/2}^2nd_v(d_c-1)(r-1).
    \label{eq:temporalOdd}
    \end{split}
\end{equation}
\end{IEEEproof}

\subsection{Proof of Lemma~\ref{lem:spatiotemporal_generator}}
\label{app:lem:spatiotemporal_generator}
\begin{IEEEproof}
Let
\[
\mathbf D
=
(\mathbf z^{(1)},\ldots,\mathbf z^{(r-1)}
\mid
\mathbf x^{(1)},\ldots,\mathbf x^{(r)})
\in \ker(\mathbf H').
\]
Define the linear map
\[
\Pi:\ker(\mathbf H')\to \mathbb F_2^n,
\qquad
\Pi(\mathbf D)
=
\sum_{t=1}^{r}\mathbf x^{(t)} \pmod 2 .
\]
Then the kernel of $\Pi$ is precisely the set of temporal codewords,
\[
\ker(\Pi)
=
\left\{
\mathbf D\in\ker(\mathbf H')
\;\middle|\;
\sum_{t=1}^{r}\mathbf x^{(t)}=\mathbf 0
\right\}.
\]
By Lemma~\ref{lemma:temporalgenerator}, $\ker(\Pi)$ is generated by
linear combinations of embedded weight-$(d_v+2)$ temporal codeword
generators.

Now let $\mathbf D\in\ker(\mathbf H')$ satisfy
\[
\Pi(\mathbf D)=\mathbf h,
\]
where $\mathbf h$ is a  stabilizer of the underlying code. Fix a
data round $q\in[r]$ and define
\[
\mathbf S_{\mathbf h}
=
(\mathbf 0,\ldots,\mathbf 0
\mid
\mathbf 0,\ldots,\mathbf 0,\mathbf h,\mathbf 0,\ldots,\mathbf 0),
\]
where $\mathbf h$ is placed in the $q$-th data block and all other blocks
are zero.

Since $\mathbf h$ is a  stabilizer, it produces zero syndrome
under the parity-check matrix used in $\mathbf H'$, i.e.,
\[
\mathbf H\mathbf h=\mathbf 0.
\]
Therefore every detector equation of $\mathbf H'$ is satisfied by
$\mathbf S_{\mathbf h}$, and hence
\[
\mathbf S_{\mathbf h}\in\ker(\mathbf H').
\]
Moreover,
\[
\Pi(\mathbf S_{\mathbf h})=\mathbf h.
\]

Define
\[
\mathbf T:=\mathbf D\oplus \mathbf S_{\mathbf h}.
\]
Since both $\mathbf D$ and $\mathbf S_{\mathbf h}$ belong to
$\ker(\mathbf H')$, we have
\[
\mathbf T\in\ker(\mathbf H').
\]
Furthermore,
\[
\Pi(\mathbf T)
=
\Pi(\mathbf D)\oplus \Pi(\mathbf S_{\mathbf h})
=
\mathbf h\oplus \mathbf h
=
\mathbf 0.
\]
Hence $\mathbf T\in\ker(\Pi)$, so $\mathbf T$ is a temporal codeword.
By Lemma~\ref{lemma:temporalgenerator}, $\mathbf T$ can be expressed as a 
linear combination of embedded weight-$(d_v+2)$ temporal codeword
generators.
\end{IEEEproof}

\section{DETAILED PROOF OF THEOREM 3}
\label{app:isotropicCounting}
\begin{IEEEproof}
For the proof, we consider an error pattern $\mathbf{e}$ with support $\mathcal{V}_e \subset \mathcal{V}$, where $\mathcal{V}$ is the support of a temporal codeword and $|\mathcal{V}_e|=\frac{|\mathcal{V}|}{2}$. We count these error patterns for which there exists a type preserving automorphism that maps them to the degenerate error pattern, i.e., $\psi\in\mathrm{Aut}(G)$ such that $\psi(\mathcal{V}_e)=\mathcal{V}\setminus \mathcal{V}_e$.

We first consider the case where $d_v$ is even and set $w:=\frac{d_v+2}{2}$.
The lowest even-weight temporal codeword is the weight-$(d_v+2)$ temporal codeword generator.
Let $G$ denote the  Tanner graph induced by this codeword,
obtained by subdividing every edge of the complete bipartite graph $K_{2,d_v}$.
Denote by $\mathcal V=\mathcal A\,\sqcup\mathcal B$, the set of variable nodes of the codeword, with $\mathcal A=\{a_1,a_2\}$ and
$\mathcal B=\{b_1,\dots,b_{d_v}\}$.

Let $A_e=\mathcal V_e\cap\mathcal A$ and $B_e=\mathcal V_e\cap\mathcal B$, and denote
$A_f=\mathcal A\setminus A_e$ and $B_f=\mathcal B\setminus B_e$. Any graph automorphism should preserve degrees, so  $|A_e|=|A_f|=1$ and
$|B_e|=|B_f|=d_v/2$. Without loss of generality assume $A_e=\{a_1\}$ and $A_f=\{a_2\}$.

Under isolation, each subdividing check is in bijection with an edge $(a,b)\in\mathcal A\times\mathcal B$,
and its syndrome label depends on its two adjacent variables.
Equivalently, there exists a function $\ell_{\mathcal V_e}:\mathcal A\times\mathcal B\to\{0,1\}$ of the form
\begin{equation}
\ell_{\mathcal V_e}(a,b)=\mathbf 1\{a\in \mathcal V_e\}\oplus \mathbf 1\{b\in \mathcal V_e\},
\end{equation}
so that $\ell_{\mathcal V_e}(a,b)=0$ exactly for the pairs $(A_e,B_e)\cup(A_f,B_f)$ and
$\ell_{\mathcal V_e}(a,b)=1$ exactly for $(A_e,B_f)\cup(A_f,B_e)$.

Define $\psi$ on $\mathcal A$ by swapping $a_1\leftrightarrow a_2$ (i.e., $\psi(A_e)=A_f$),
and choose any bijection $\pi:B_e\to B_f$, extended to a permutation of $\mathcal B$.
Extend $\psi$ to subdividing checks by mapping the check on edge $(a,b)$ to the check on edge
$(\psi(a),\pi(b))$. Since $G$ is the subdivided $K_{2,d_v}$, this defines a type-preserving automorphism of $G$. Clearly, the syndrome labels are preserved, as if $\ell_{\mathcal V_e}(a,b)=0$ then $\ell_{\mathcal V_e}(\psi(a),\pi(b))=0$ because it is connected to $(A_e,B_e)\cup(A_f,B_f)$, similarly for $\ell_{\mathcal V_e}(a,b)=1$.
Therefore, for every such balanced support $\mathcal{V}_e$ there exists a type-preserving labeled automorphism $\psi$
mapping $\mathcal{V}_e$ to its complement. Let $v\in \mathcal{V}_e$ and
$v'=\psi(v)\in (\mathcal A\,\sqcup\mathcal B)\setminus \mathcal{V}_e$.
By Lemma~\ref{lem:aut_equivariance}, the hard decisions at $v$ and $v'$ are identical.
Thus, under the isolation assumption, the decoder's hard decision cannot match the support of either 
$\mathcal{V}_e$, or $(\mathcal A\,\sqcup\mathcal B)\setminus \mathcal{V}_e$.
Finally, the number of such balanced supports in $\mathbf{H}'$ is
\begin{equation}
\underline{N}_{\mathrm{even}}
=2\binom{d_v}{d_v/2}n(r-1).
\label{eq:lowedEven}
\end{equation}

The case of odd $d_v$ follows  by considering the lowest even-weight temporal codewords, which are formed by concatenating two temporal codeword generators. Two temporal codeword generators can either share i) a data variable node, or ii) a measurement variable node.

i) The concatenated temporal codeword is
\[
\mathcal V=\mathcal U_1 \,\sqcup\, \mathcal U_2,
\]
where
\[
\mathcal U_1=\{a_1\}\,\sqcup\,\mathcal B_1,\qquad
\mathcal U_2=\{a_2\}\,\sqcup\,\mathcal B_2,
\]
with $|\mathcal B_1|=|\mathcal B_2|=d_v$ and hence
$|\mathcal U_1|=|\mathcal U_2|=d_v+1$.
Let
\[
w=\frac{d_v+1}{2}.
\]

By the partition inequalities
\eqref{eq:TAsimple}--\eqref{eq:TBsimple} specialized to $k=2$,
any weight-$(d_v+1)$ harmful candidate
$\mathcal V_e\subseteq\mathcal V$ must satisfy
\[
|\mathcal V_e\cap\mathcal U_i|\leq w,
\qquad i\in\{1,2\}.
\]
Since
\[
|\mathcal V_e|=d_v+1=2w
\]
and $\mathcal V=\mathcal U_1\sqcup\mathcal U_2$, it follows that
\begin{equation}
\label{eq:odd_balanced_split_revised}
|\mathcal V_e\cap\mathcal U_1|
=
|\mathcal V_e\cap\mathcal U_2|
=
w.
\end{equation}

We now count those supports $\mathcal V_e$ for which there exists a
type-preserving, syndrome-label-preserving automorphism
$\psi\in\mathrm{Aut}(G)$ such that
\[
\psi(\mathcal V_e)
=
\mathcal V\setminus\mathcal V_e.
\]

The induced isolated Tanner graph has exactly two variable nodes of
degree $d_v$, namely $a_1$ and $a_2$, whereas all nodes in $\mathcal B_1\cup\mathcal B_2$ have degree $2$. Therefore any
type-preserving automorphism must map $\{a_1,a_2\}$ to itself. Since
$\psi(\mathcal V_e)=\mathcal V\setminus\mathcal V_e$, we must have
\[
a_1\in\mathcal V_e
\iff
a_2\notin\mathcal V_e,
\]
and hence
\[
|\mathcal V_e\cap\{a_1,a_2\}|=1.
\]

Assume without loss of generality that
$a_1\in\mathcal V_e$ and $a_2\notin\mathcal V_e$.
Then \eqref{eq:odd_balanced_split_revised} implies
\[
|\mathcal V_e\cap\mathcal B_1|
=
w-1
=
\frac{d_v-1}{2},
\qquad
|\mathcal V_e\cap\mathcal B_2|
=
w
=
\frac{d_v+1}{2}.
\]

Index the nodes in the two sets as
\[
\mathcal B_1=\{b_{1,1},\ldots,b_{1,d_v}\},
\qquad
\mathcal B_2=\{b_{2,1},\ldots,b_{2,d_v}\},
\]
where $b_{1,j}$ and $b_{2,j}$ are connected through the corresponding
middle check node. Define
\[
S_1
:=
\{j:b_{1,j}\in\mathcal V_e\},
\qquad
S_2
:=
\{j:b_{2,j}\in\mathcal V_e\}.
\]
Then
\[
|S_1|=\frac{d_v-1}{2},
\qquad
|S_2|=\frac{d_v+1}{2},
\]
and therefore
\begin{equation}
|S_1|+|S_2|=d_v.
\label{eq:S1S2sum}
\end{equation}
For $\alpha,\beta\in\{0,1\}$, define
\[
C_{\alpha\beta}
:=
\left\{
j:
\mathbf 1\{j\in S_1\}=\alpha,\;
\mathbf 1\{j\in S_2\}=\beta
\right\}.
\]
Thus, for example, $j\in C_{10}$ means that
$b_{1,j}\in\mathcal V_e$ and
$b_{2,j}\notin\mathcal V_e$.

Since the four sets $C_{\alpha\beta}$ partition $[d_v]$,
\[
d_v
=
|C_{00}|+|C_{01}|+|C_{10}|+|C_{11}|.
\]
On the other hand, Eq.~\eqref{eq:S1S2sum} gives
\[
d_v
=
|S_1|+|S_2|
=
|C_{10}|+|C_{01}|+2|C_{11}|.
\]
Comparing the two expressions yields
\begin{equation}
|C_{00}|=|C_{11}|.
\label{eq:C00C11}
\end{equation}
We may therefore choose a permutation $\rho$ of $[d_v]$ such that
\[
\rho(C_{00})=C_{11},
\qquad
\rho(C_{11})=C_{00},
\]
while $\rho$ maps $C_{10}$ and $C_{01}$ onto themselves.  Define $\psi$ on the variable nodes by
\[
\psi(a_1)=a_2,
\qquad
\psi(a_2)=a_1,
\]
and
\[
\psi(b_{1,j})=b_{2,\rho(j)},
\qquad
\psi(b_{2,j})=b_{1,\rho(j)}.
\]
By construction,
\[
b_{1,j}\in\mathcal V_e
\iff
b_{2,\rho(j)}\notin\mathcal V_e,
\]
and
\[
b_{2,j}\in\mathcal V_e
\iff
b_{1,\rho(j)}\notin\mathcal V_e.
\]
Together with
$a_1\in\mathcal V_e$ and $a_2\notin\mathcal V_e$, this gives
\[
\psi(\mathcal V_e)
=
\mathcal V\setminus\mathcal V_e.
\]

It remains to verify that $\psi$ extends to a
syndrome-label-preserving automorphism of the induced Tanner graph.
For each index $j$, the corresponding part of the graph is a path
between $a_1$ and $a_2$ through $b_{1,j}$ and $b_{2,j}$.
The above map sends this path to the path indexed by $\rho(j)$ with its
orientation reversed. In particular, the middle check connecting
$b_{1,j}$ and $b_{2,j}$ is mapped to the middle check connecting
$b_{1,\rho(j)}$ and $b_{2,\rho(j)}$, so adjacency is preserved.

Moreover, since $\psi$ maps nodes in $\mathcal{V}_e$ to nodes in $\mathcal{V}_f$ and vice versa at both ends of every check,
the syndrome labels are unchanged. Thus, every support satisfying the above balanced conditions admits a
type-preserving syndrome-label-preserving automorphism satisfying
\[
\psi(\mathcal V_e)
=
\mathcal V\setminus\mathcal V_e.
\]

For the choice
$a_1\in\mathcal V_e$, $a_2\notin\mathcal V_e$, the number of such
supports is
\[
\binom{d_v}{\frac{d_v-1}{2}}
\binom{d_v}{\frac{d_v+1}{2}}.
\]
The case
$a_2\in\mathcal V_e$, $a_1\notin\mathcal V_e$
is symmetric. Hence the total number of such supports per temporal
codeword is
\[
2
\binom{d_v}{\frac{d_v-1}{2}}
\binom{d_v}{\frac{d_v+1}{2}}
=
2
\binom{d_v}{\frac{d_v-1}{2}}^{2}.
\]

There are $n(r-2)$ temporal codewords of this type in $\mathbf H'$.
Therefore, any syndrome-based isotropic decoder cannot 
estimate at least
\[
2
\binom{d_v}{\frac{d_v-1}{2}}^{2}
n(r-2)
\]
weight-$(d_v+1)$ temporal degenerate error patterns under the isolation
assumption.

ii) Let $d_v$ be odd and set \[
w:=\frac{d_v+1}{2}.
\]
Consider two temporal codeword generators $\mathbf c_1$ and $\mathbf c_2$ that share a measurement variable node.
The support of the resulting
weight-$(2d_v+2)$ temporal codeword decomposes as
\[
\mathcal V=\mathcal U_1\,\sqcup\,\mathcal U_2,
\qquad
|\mathcal U_1|=|\mathcal U_2|=d_v+1,
\]
where
\[
\mathcal U_i
=
\mathcal A_i\,\sqcup\,\mathcal B_i',
\qquad
|\mathcal A_i|=2,
\qquad
|\mathcal B_i'|=d_v-1.
\]
The two nodes in each $\mathcal A_i$ have degree $d_v$ in the
induced graph, whereas the nodes in $\mathcal B_i'$ have degree $2$.

Let $\mathcal V_e\subseteq\mathcal V$ be a weight-$(d_v+1)$
error support, and define
\[
x_i:=|\mathcal V_e\cap\mathcal U_i|,
\qquad i\in\{1,2\}.
\]
Since $\mathbf c_i$ contains $\mathcal U_i$ together with the shared
measurement variable, which does not belong to $\mathcal V_e$, we have
\[
|\mathcal V_e\cap\operatorname{supp}(\mathbf c_i)|=x_i.
\]
For $\mathcal V_e$ not to be reduced in weight by adding either
constituent temporal codeword, it is necessary that
\[
w(\mathbf e\oplus\mathbf c_i)\geq w(\mathbf e),
\qquad i\in\{1,2\}.
\]
Using $w(\mathbf c_i)=d_v+2$, this gives
\[
d_v+2-2x_i\geq0,
\]
and therefore, since $d_v$ is odd and $x_i$ is integer,
\[
x_i\leq\frac{d_v+1}{2}=w.
\]
On the other hand,
\[
x_1+x_2
=
|\mathcal V_e|
=
d_v+1
=
2w.
\]
Hence
\begin{equation}
\label{eq:odd_case2_balanced}
|\mathcal V_e\cap\mathcal U_1|
=
|\mathcal V_e\cap\mathcal U_2|
=
w.
\end{equation}

Define
\[
t_i
:=
|\mathcal V_e\cap\mathcal A_i|
\in\{0,1,2\}.
\]
Equation~\eqref{eq:odd_case2_balanced} then implies
\[
|\mathcal V_e\cap\mathcal B_i'|
=
w-t_i.
\]

We now characterize the supports $\mathcal V_e$ for which there exists
a syndrome-label-preserving automorphism
$\psi\in\mathrm{Aut}(G)$ satisfying
\[
\psi(\mathcal V_e)
=
\mathcal V\setminus\mathcal V_e.
\]

\paragraph{Necessity.}
The four nodes in
$\mathcal A_1\sqcup\mathcal A_2$
are the variable nodes of degree $d_v$ in the induced graph.
Hence every type-preserving automorphism preserves the set
$\mathcal A_1\sqcup\mathcal A_2$.
If
\[
\psi(\mathcal V_e)
=
\mathcal V\setminus\mathcal V_e,
\]
then the number of selected and unselected nodes within this degree
class must be equal. Therefore
\[
|\mathcal V_e\cap(\mathcal A_1\sqcup\mathcal A_2)|
=
|(\mathcal A_1\sqcup\mathcal A_2)\setminus\mathcal V_e|,
\]
which gives
\begin{equation}
\label{eq:t1t2sum2}
t_1+t_2=2.
\end{equation}

\paragraph{Sufficiency.}
Conversely, suppose that
\[
t_1+t_2=2.
\]
Then
\[
|\mathcal A_1\cap\mathcal V_e|
=
|\mathcal A_2\setminus\mathcal V_e|,
\]
and
\[
|\mathcal A_1\setminus\mathcal V_e|
=
|\mathcal A_2\cap\mathcal V_e|.
\]
Moreover, using \eqref{eq:odd_case2_balanced},
\begin{align*}
|\mathcal B_1'\cap\mathcal V_e|
&=
w-t_1,\\
|\mathcal B_2'\setminus\mathcal V_e|
&=
(d_v-1)-(w-t_2)\\
&=
w-2+t_2\\
&=
w-t_1,
\end{align*}
where the final equality follows from $t_1+t_2=2$. Similarly,
\[
|\mathcal B_1'\setminus\mathcal V_e|
=
|\mathcal B_2'\cap\mathcal V_e|.
\]

Hence there exist bijections
\[
\pi_A:
\mathcal A_1\cap\mathcal V_e
\longrightarrow
\mathcal A_2\setminus\mathcal V_e,
\]
\[
\pi_A':
\mathcal A_1\setminus\mathcal V_e
\longrightarrow
\mathcal A_2\cap\mathcal V_e,
\]
and
\[
\pi_B:
\mathcal B_1'\cap\mathcal V_e
\longrightarrow
\mathcal B_2'\setminus\mathcal V_e,
\]
\[
\pi_B':
\mathcal B_1'\setminus\mathcal V_e
\longrightarrow
\mathcal B_2'\cap\mathcal V_e.
\]
Combining $\pi_A$ and $\pi_A'$ defines a bijection
$\pi_{\mathcal A}:\mathcal A_1\to\mathcal A_2$, while combining
$\pi_B$ and $\pi_B'$ defines a bijection
$\pi_{\mathcal B}:\mathcal B_1'\to\mathcal B_2'$.
We define $\psi$ on $\mathcal U_1$ using these bijections and on
$\mathcal U_2$ using their inverses. Therefore,
\[
\psi(\mathcal V_e)
=
\mathcal V\setminus\mathcal V_e.
\]

It remains to show that this variable node permutation extends to an
automorphism of the induced Tanner graph. Within each
$\mathcal U_i$, the graph is obtained by subdividing the edges of
$K_{2,d_v-1}$ between $\mathcal A_i$ and $\mathcal B_i'$.
Thus, a check associated with a pair $(a,b)$,
$a\in\mathcal A_1$ and $b\in\mathcal B_1'$, can be mapped to the
corresponding check associated with
\[
\bigl(\pi_{\mathcal A}(a),\pi_{\mathcal B}(b)\bigr)
\]
in $\mathcal{U}_2$, and conversely using the inverse maps.

In addition, cancellation of the shared measurement variable leaves
two check nodes connecting corresponding nodes of
$\mathcal A_1$ and $\mathcal A_2$. Label
\[
\mathcal A_1=\{a_{1,1},a_{1,2}\},
\qquad
\mathcal A_2=\{a_{2,1},a_{2,2}\},
\]
so that the two shared checks connect $a_{1,j}$ and $a_{2,j}$,
$j\in\{1,2\}$. Any bijection
$\pi_{\mathcal A}:\mathcal A_1\to\mathcal A_2$ is either the identity
or a transposition with respect to these two indices. Since both
permutations are involutions, extending the action from
$\mathcal A_2$ by $\pi_{\mathcal A}^{-1}$ maps the two shared checks
onto one another. Hence adjacency is preserved throughout the induced
graph.

Finally, the syndrome labels are also preserved. Under the isolation
assumption, every check in the induced subgraph has two neighboring
variable nodes, and its syndrome label is determined by the parity of
their membership in $\mathcal V_e$.
Since $\psi$ maps  variables in $\mathcal{V}_e$ to variables in $\mathcal{V}_f$ and vice versa,
$\psi$ preserves every syndrome label.

Thus, condition \eqref{eq:t1t2sum2} is both necessary and sufficient
for the balanced supports considered here to admit a type-preserving,
syndrome-label-preserving automorphism satisfying
\[
\psi(\mathcal V_e)
=
\mathcal V\setminus\mathcal V_e.
\]
By Lemma~\ref{lem:aut_equivariance}, the hard decisions on paired
nodes are identical, and therefore an isotropic decoder cannot estimate any  of the corresponding degenerate error patterns under the isolation
assumption.

\paragraph{Exact count.}
Let
\[
x:=|\mathcal B_i'|=d_v-1.
\]
For each $t\in\{0,1,2\}$ with
$t_1=t$ and $t_2=2-t$, the number of supports satisfying the above
conditions is
\[
\binom{2}{t}
\binom{x}{w-t}
\binom{2}{2-t}
\binom{x}{w-(2-t)}.
\]
Summing over $t=0,1,2$ gives
\begin{align}
&2\binom{x}{w}\binom{x}{w-2}
+
4\binom{x}{w-1}^{2}
\nonumber\\
&\qquad =
2\binom{d_v-1}{\frac{d_v+1}{2}}
\binom{d_v-1}{\frac{d_v-3}{2}}
+
4\binom{d_v-1}{\frac{d_v-1}{2}}^{2}.
\label{eq:exact_odd_case2_revised}
\end{align}

There are
\[
\frac{1}{2}nd_v(d_c-1)(r-1)
\]
temporal codewords of this type in $\mathbf H'$.
Combining this family with case~i), any syndrome-based isotropic
decoder cannot correctly estimate at least
\begin{equation}
\begin{split}
\underline{N}_{\mathrm{odd}} 
={}& 
\left(
2\binom{d_v-1}{\frac{d_v+1}{2}}
\binom{d_v-1}{\frac{d_v-3}{2}}
+
4\binom{d_v-1}{\frac{d_v-1}{2}}^{2}
\right)\\
\frac{1}{2}nd_v(d_c-1)(r-1)
&+
2\binom{d_v}{\frac{d_v-1}{2}}^{2}n(r-2)
\label{eq:lowerBoundOdd}
\end{split}
\end{equation}
weight-$(d_v+1)$ temporal degenerate error patterns under the
isolation assumption.

It remains to verify that the enumeration counts distinct error supports rather than incidences between errors and temporal
codewords. Let $\mathbf c_{v,t}$ denote the elementary temporal
codeword generator associated with data variable $v$ and measurement interval $t$. Its support is
\[
\operatorname{supp}(\mathbf c_{v,t})
=
\left\{
x_v^{(t)},x_v^{(t+1)}
\right\}
\cup
\left\{
z_c^{(t)}:
c\in\mathcal N_{\mathbf H}(v)
\right\}.
\]

Suppose first that $d_v$ is even. Every error support counted above
within $\mathbf c_{v,t}$ contains exactly one of the two data
coordinates
\[
x_v^{(t)},\qquad x_v^{(t+1)},
\]
and $d_v/2\geq2$ measurement coordinates belonging to measurement interval $t$. The data coordinate identifies the underlying variable
$v$, while the measurement coordinates identify the interval $t$. Consequently, the supporting generator $\mathbf c_{v,t}$ is uniquely
determined by the error support.

Now let $d_v$ be odd and consider case~(i), in which
\[
\mathbf D
=
\mathbf c_{v,t}\oplus\mathbf c_{v,t+1}.
\]
Every error support counted in this case contains exactly one data coordinate, belonging to variable $v$, together with nonzero numbers
of measurement coordinates in each of the two consecutive intervals $t$ and $t+1$. The data coordinate therefore identifies $v$, and the
two occupied measurement intervals identify $t$. Hence the supporting temporal codeword $\mathbf D$ is again uniquely determined.

Finally, consider case~(ii), in which
\[
\mathbf D
=
\mathbf c_{u,t}\oplus\mathbf c_{v,t},
\qquad u\neq v,
\]
where $u$ and $v$ share a check $c_0$ in $\mathbf H$. Define
\[
A_u
=
\left\{x_u^{(t)},x_u^{(t+1)}\right\},
\qquad
A_v
=
\left\{x_v^{(t)},x_v^{(t+1)}\right\},
\]
and let
\[
t_u
=
\left|
\operatorname{supp}(\mathbf e)\cap A_u
\right|,
\qquad
t_v
=
\left|
\operatorname{supp}(\mathbf e)\cap A_v
\right|.
\]
The supports counted above satisfy
\[
t_u+t_v=2.
\]
If $t_u=t_v=1$, the two selected data coordinates identify the unordered pair $\{u,v\}$, while the selected measurement coordinates identify the interval $t$. Thus, $\mathbf D$ is uniquely determined.

It remains to consider $\{t_u,t_v\}=\{0,2\}$. Without loss of generality, suppose that $t_u=2$ and $t_v=0$. The two selected data
coordinates identify $u$ and the interval $t$. Let
\[
\mathcal S
=
\left\{
c:
z_c^{(t)}\in\operatorname{supp}(\mathbf e),
\quad
c\notin\mathcal N_{\mathbf H}(u)
\right\}.
\]
The set $\mathcal S$ consists precisely of the selected measurement neighbors of $v$ and satisfies
\[
|\mathcal S|
=
\frac{d_v+1}{2}
\geq2.
\]
If the same error support were associated with another variable $v'\neq v$, then every check in $\mathcal S$ would be adjacent to both
$v$ and $v'$. The variables $v$ and $v'$ would therefore share at least two checks, producing a $4$-cycle in the Tanner graph of
$\mathbf H$. This contradicts the assumption that its girth is larger than four. Hence $v$ is unique. The case $t_u=0$ and $t_v=2$ is
symmetric.

Moreover, the supports counted in cases~(i) and~(ii) cannot coincide: those in case~(i) contain exactly one data coordinate, whereas those
in case~(ii) contain exactly two. Therefore, no error support included in the preceding enumeration is counted under two distinct temporal
codewords. The quantities $\underline N_{\mathrm{even}}$ and $\underline N_{\mathrm{odd}}$ consequently count distinct error supports and provide valid lower bounds.
\end{IEEEproof}

\section*{Acknowledgment}
This work was supported by the National Science Foundation under grant CCF-2420424.

\bibliographystyle{IEEEtran}
\bibliography{bib/totalRefs.bib}

\end{document}